\documentclass{article}
\usepackage{arxiv}

\usepackage[utf8]{inputenc} 
\usepackage[T1]{fontenc}    
\usepackage{hyperref}       
\usepackage{url}            
\usepackage{booktabs}       
\usepackage{amsfonts}       
\usepackage{nicefrac}       
\usepackage{microtype}      
\usepackage{lipsum}		
\usepackage{graphicx}
\usepackage{natbib}
\usepackage{doi}
\usepackage[linesnumbered,ruled,vlined]{algorithm2e}
\usepackage{amsthm}
\usepackage{amsmath}
\usepackage{amssymb}

\usepackage{amsfonts}
\usepackage{bm}
\usepackage{longtable}
\usepackage{pdflscape}
\usepackage{dsfont}
\usepackage{multirow} 
\usepackage{booktabs} 
\newcommand*{\E}{\mathop{{}\mathbb{E}}}

\newcommand*{\R}{\mathop{{}\mathbb{R}}}

\theoremstyle{dgthm}
\newtheorem{theorem}{Theorem}

\usepackage{bbm}

\usepackage[noend]{algpseudocode}
\usepackage{amsmath}

\newcommand{\xfnm}[1][]{\ifx!#1!\else\unskip,\space#1\fi}

\newcommand{\indep}{\perp \!\!\! \perp}
\newtheorem{assumption}{Assumption}

\title{A template for the \emph{arxiv} style}

\author{{\hspace{1mm}Shinpei Nakamura-Sakai}\thanks{Corresponding author} \\
	Department of Statistics and Data Science\\
	Yale University\\
	New Haven, CT \\
	\texttt{s.nakamura.sakai@yale.edu} \\
	\And
	{Laura Forastiere} \\
	Department of Biostatistics\\
	Yale University\\
	New Haven, CT \\
        \texttt{laura.forastiere@yale.edu} \\
	\\
    \And
	{Brian Macdonald} \\
	Department of Statistics and Data Science\\
	Yale University\\
	New Haven, CT \\
        \texttt{brian.macdonald@yale.edu} \\
	\\
}

\renewcommand{\headeright}{}
\renewcommand{\undertitle}{}
\renewcommand{\shorttitle}{Estimating the age-conditioned average treatment effects curves}

\hypersetup{
pdftitle={Estimating the age-conditioned average treatment effects curves: An application for assessing load-management strategies in the NBA},
pdfsubject={q-bio.NC, q-bio.QM},
pdfauthor={Shinpei Nakamura-Sakai, Laura Forastiere, Brian Macdonad},
pdfkeywords={First keyword, Second keyword, More},
}

\title{Estimating the age-conditioned average treatment effects curves: An application for assessing load-management strategies in the NBA}

\begin{document}
\maketitle

\begin{abstract}
	In the realm of competitive sports, understanding the performance dynamics of athletes, represented by the age curve (showing progression, peak, and decline), is vital. Our research introduces a novel framework for quantifying age-specific treatment effects, enhancing the granularity of performance trajectory analysis. Firstly, we propose a methodology for estimating the age curve using game-level data, diverging from traditional season-level data approaches, and tackling its inherent complexities with a meta-learner framework that leverages advanced machine learning models. This approach uncovers intricate non-linear patterns missed by existing methods. Secondly, our framework enables the identification of causal effects, allowing for a detailed examination of age curves under various conditions. By defining the Age-Conditioned Treatment Effect (ACTE), we facilitate the exploration of causal relationships regarding treatment impacts at specific ages. Finally, applying this methodology to study the effects of rest days on performance metrics, particularly across different ages, offers valuable insights into load management strategies' effectiveness. Our findings underscore the importance of tailored rest periods, highlighting their positive impact on athlete performance and suggesting a reevaluation of current management practices for optimizing athlete performance.
\end{abstract}


\section{Introduction} 
Athletic performance over an individual's career is not a constant; it is a dynamic interplay of improvement, peak performance, and eventual decline. This performance trajectory, often visualized as an "age curve," encapsulates the biological and career lifespan of athletes across various sports disciplines. While it is widely recognized that these curves vary due to a myriad of factors unique to each athlete, the scientific inquiry into the heterogeneity of treatment effects across different ages remains underexplored. The nuanced understanding of how specific interventions or treatments impact athletes at various stages of their careers can be instrumental in enhancing their performance and extending their peak years in the sport.

\cite{Schuckers2023}' review  focuses on regression-based models for age curve analysis, specifically highlighting the predominance of fixed effects models that incorporate quadratic and cubic age terms \cite{Brander2014EffectsOfAge}, and Lichtman's delta method for age effect analysis \cite{Lichtman2009BaseballAging}. Importantly, the review notes that all discussed models utilize season-level data. It also points to advancements through semiparametric methods such as spline techniques and generalized additive models \cite{WakimJin2014} \cite{Judge2020DeltaMethod} \cite{Turtoro2019FlexibleAging}, which offer more nuanced analytical capabilities.

The emphasis on season-level data in these models is critical; however, transitioning to game-level data could significantly refine the analysis by incorporating detailed variables like opponent strength, home advantage, and team dynamics. While offering a richer dataset, this shift introduces complexity due to the intricate interactions between these variables. Traditional models, designed around season-level data using regression-based models, face challenges in accurately capturing these dynamics, highlighting the need for a new analytical framework. Such a framework would efficiently handle non-linear interactions between game-level variables, representing a significant methodological advancement in age curve analysis.

Integrating Conditional Average Treatment Effect (CATE) or Heterogeneous Treatment Effect (HTE) estimation into the age-curve literature could significantly enhance the precision of understanding how various treatments impact athletes at different ages. HTE estimation is a statistical approach that identifies subgroups of individuals who may experience varying responses to a given treatment \cite{athey2016recursive}. By applying HTE methods to the age-curve literature, researchers can move beyond general trends to discover how the age of an athlete interacts with the effectiveness of interventions. Meta-learning techniques, such as those outlined in  \cite{kunzel_metalearners_2019}, can improve the power and flexibility of HTE estimation in this context. For example, HTE estimation could reveal varying benefits of injury prevention strategies across different age groups, if particular nutritional adjustments yield greater performance benefits for certain age groups, how mental health interventions can be tailored to age-specific needs, or if recovery techniques like cryotherapy demonstrate differential effects depending on an athlete's age. This precise analysis enables better resource allocation, targeting interventions where they'll be most beneficial to improve performance and health outcomes cost-effectively.  Beyond refining current age-curve analysis, it paves the way for personalized athlete care strategies.

In the past, athletic management strategies regarding rest periods and "load management" have largely relied on observed differences in how athletes of various ages respond to training and recovery. However, this field lacks rigorous quantitative analysis.  As professional sports become increasingly competitive, there's a critical need to develop data-driven, individualized strategies. Such strategies would integrate an athlete's unique profile to optimize performance outcomes while safeguarding their long-term health and well-being. Recent research emphasizes a shift away from relying solely on the concept of workload as a predictor of injury risk and performance. Instead, there's an increasing focus on understanding the complex relationship between an athlete's training load, internal responses, and external factors \cite{Gabbett2016} \cite{Drew2016}. This highlights the need for individualized load management strategies that continuously monitor a variety of physiological and subjective markers. The aim is to balance performance goals with injury prevention by acutely measuring how athletes adapt to and recover from training stress. This data-driven approach promises to enhance decision-making for optimizing training schedules, and recovery protocols, and mitigating athlete injury risk.

In response to those needs, our work introduces a novel analytical framework aimed at estimating age-specific treatment effects to construct a personalized age curve for treatment impact. This framework is particularly focused on decision-making processes regarding athlete rest and load management, a concept at the forefront of sports science and athlete longevity. The contributions of this study are threefold. Initially, we unveil a method for estimating the age curve utilizing game-level data, a departure from conventional approaches that depend on aggregated season-level data. This shift to game-level data introduces complexities due to a larger array of variables needing management. We address this by deploying a meta-learner framework that allows for the application of sophisticated machine learning models, adept at uncovering non-linear patterns that elude current leading age-curve estimation techniques. Second, our framework offers a means to discern causally identifiable treatment effects, facilitating a nuanced comparison of age curves across various covariates and interventions. By establishing the Age-Conditioned Treatment Effect (ACTE) for any given intervention, our methodology supports the exploration of causal hypotheses related to treatment impacts and outcomes for specific ages. Lastly, we apply our methodology to investigate the influence of rest days between games on key performance metrics, with a focus on athletes' age. This analysis provides insights into the effectiveness of load management strategies, highlighting the performance benefits when athletes receive adequate rest. Through these contributions, our study not only enhances our understanding of age-related performance dynamics but also offers practical guidance for optimizing athlete management strategies.

This introduction paves the way for the forthcoming sections, setting the stage for a detailed exploration of our methodology in Section \ref{sec:2}, a simulation designed to offer an in-depth understanding of the framework in Section \ref{sec:3}, an examination of the outcomes and implications of our findings in Section \ref{sec:4}, and we will conclude with a discussion in Section \ref{sec:5}. Our study proposes a significant shift towards individualized athlete management strategies, moving away from traditional, one-size-fits-all approaches. This potential shift aims to redefine how athletic performance and longevity are optimized, marking a pivotal change in the field.

\section{Methodology}
\label{sec:2}
\subsection{Framework}

Let us consider a sample of N units, represented by the indicator $i\in \{1,...,N\}$. For each unit, we observe the following variables: $A\in \mathbb{Z}$, representing unit i's age, $W\in \{0,1\}$, denoting a binary treatment of interest, a vector of covariates ${X}_i$, and the outcome of interest denoted by $Y_i^{obs}$.
Under Rubin's potential outcome framework \cite{rubin_estimating_1974}, we make the Stable Unit Treatment Value Assumption (SUTVA; see appendix \ref{app:A}) and we denote by $Y(w)$ the potential outcome of unit i under treatment $w$.  We assume that for each unit i the tuple $(Y_i(0), Y_i(1), A_i, W_i, X_i)$ is independently drawn from a distribution $\mathcal{P}$, such that $(Y_i(0), Y_i(1), A_i, W_i, X_i)\sim \mathcal{P}$.




For each potential outcome, we define the  Conditional Expectation Function (CFE), denoted by $\mu_w(a,x)$, as the expectation of the potential outcome as a function of age and covariates, that is:
    \begin{equation}
\mu_0(a, x):=\E[Y_i(0)|A_i=a, X_i=x] \text{ and } \mu_1(a,x):=\E[Y_i(1)|A_i=a, X_i=x]
\end{equation}.

Under this framework, our main causal estimand of interest is the  Age-Conditioned Treatment Effect (ACTE), that is:
\begin{equation}
\begin{aligned}
\tau(a)&:=\E[Y_i(1) - Y_i(0)|A=a]\\&=\mu_1(a)-\mu_0(a)=\textstyle\E_\mathcal{X}[\mu_1(a,x)-\mu_0(a,x)]
\end{aligned}
\end{equation}
where $\mu_w(a)=\textstyle\E_\mathcal{X}[\mu_1(a,x)]$ is referred to as the Age-Conditioned Expectation Function (ACEF).
Following \cite{Schuckers2023}, we can decompose the potential outcome as follows:
\begin{equation}\label{eq:outcome}
    Y_i(w)=g(a,w)+f_i(x, a,w) + \epsilon_i
\end{equation}
where $g(a,w)$ is the average performance at age $a$ for all players with treatment $w$, $f(x,a,w)$ represents a performance adjustment at age $a$ for player with covariates $x$, and treatment $w$, and $\epsilon$ is the mean zero model error. Note that under this decomposition, $g$ is the representation of all players, and hence $f$ represents the variation for all players at age $a$. Namely, $\E_{\mathcal{X}}[f(x,a,w)|A=a]=0$

Note that using the decomposition in Equation \ref{eq:outcome}, we have
\begin{align*}
    \tau(a)=\E[Y_i(1) - Y_i(0)|A=a]&=g(a,1)-g(a,0)+\textstyle\E_\mathcal{X}[f(x,a,1)-f(x,a,0)|A=a]\\
    &=g(a,1)-g(a,0)
\end{align*}

In order to properly identify this estimand we need the following assumptions:

\begin{assumption}[Unconfoundedness]\label{ass:uncounfoundedness}
There are no unmeasured factors affecting the probability of treatment allocation and the outcome. Formally,
        $$(Y_i(0), Y_i(1))\indep W_i|X_i$$
\end{assumption}

Under these assumptions, we can identify our estimand of interest using the following theorem.
\begin{theorem}[Identification of ACTE]\label{them:IdentifiACTE}
    Under SUTVA and 
    Assumption \ref{ass:uncounfoundedness} we have
    $$\tau(a)=\textstyle\E_{\mathcal{X}}[\E[Y_i^{obs}|W_i=1, A_i=a, X_i=x]] - \textstyle\E_{\mathcal{X}}[\E[Y_i^{obs}|W_i=0, A_i=a, X_i=x]]$$
\end{theorem}
\begin{proof}
    See appendix \ref{app:A}
\end{proof}

\subsection{Estimation of $\tau(a)$}
To estimate $\tau(a)$, we implement a causal machine learning method, particularly meta-learners, as introduced by \cite{kunzel_metalearners_2019}. This approach is well-suited for our observational study, especially considering the complexities in estimating $g(a,x)$. Meta-learners offer notable flexibility, allowing for the integration of various predictive models to suit different goals. \cite{kunzel_metalearners_2019} proposes three distinct types of meta-learners: S-learners, T-learners, and X-learners, each designed for specific causal inference challenges.

The \textbf{S-learner} is termed so due to its use of a 'single' estimator. This learner includes the treatment as one of the independent variables in the outcome model, i.e., 
$$\mu^{obs}(a, x, w):=\E[Y_i^{obs}|A=a, W=w, X=x],$$ 
which is estimated by a base learner,  which could be any supervised learning or regression estimator. 
%
%
Consequently, 
the ACEF is given by $$\hat{\mu}_w(a)= \textstyle\E_{\mathcal{X}}[\hat{\mu}^{obs}(a,x,w)],$$ 
and the ACTE is then
$$\hat{\tau}(a)=\hat{\mu}_1(a)-\hat{\mu}_0(a).$$
The pseudocode of S-learner is provided in Appendix \ref{app:B}

The \textbf{T-learner} uses 'two' separate learners to estimate $\tau$. It independently learns $\hat{\mu}_0(a,x)$ and $\hat{\mu}_1(a,x)$ for the treatment and control groups, respectively, following the setup: 

$$\mu^{obs}_0(a,x):=\E[Y_i^{obs}|A_i=a, W_i=0, X_i=x] \text{ and }\mu_1^{obs}(a,x):=\E[Y_i^{obs}|A_i=a, W_i=1, X_i=x]$$

The first expression, the control response function, is estimated by any base learner that could be any model of choice using the observation in the control group, namely $\{(A_i,X_i,Y_i)\}_{\{i:W_i=0\}}$ and we denote the estimated function as $\hat{\mu}_0^{obs}(a,x)$. Similarly, $\mu_1^{obs}(a,x)$ is learned using $\{(A_i,X_i,Y_i)\}_{\{i:W_i=1\}}$ with a potentially different base learner of choice and denoting the estimator by $\hat{\mu}^{obs}_1(x,a)$. 

Then, 
the ACEF is given by $$\hat{\mu}_w(a)= \textstyle\E_{\mathcal{X}}[\hat{\mu}_w^{obs}(a,x)],$$ 
and the ACTE is given by
$\hat{\tau}(a)=\hat{\mu}_1(a)-\hat{\mu}_0(a)$ as before.
The pseudocode of the T-learner is provided in Appendix \ref{app:B}.

\textbf{X-learner} involves a two-step process. Initially, it estimates $\hat{\mu}_1^{obs}(a, x)$ and $\hat{\mu}_0^{obs}(a, x)$ akin to the T-learner either using the same of different base-learners. In the second step, using the model trained in step one, it 'imputes' the treatment effect for the treated units, denoted by $D_i^1$, using $\hat{\mu}_0^{obs}(a, x)$, that is,
$$
\tilde{D}^1_i=Y^{obs}_i-\hat{\mu}_0^{obs}(A_i,X_i) \qquad \forall i: W_i=1,
$$
and it 'imputes' the treatment effect for the control units, denoted by $D_i^0$, using $\hat{\mu}_1^{obs}(a, x)$, that is,
$$\tilde{D}^0_i=\hat{\mu}^{obs}_1(A_i,X_i)-Y^{obs}_i \qquad \forall i: W_i=0.$$


Then we could use any predictive model including regression or supervised machine learning models to estimate $$\tau^0(a) := \textstyle\E_{\mathcal{X}}[\E[\tilde{D}_i^1|A_i=a, X_i=x]] \text{ and } \tau^1(a) := \textstyle\E_{\mathcal{X}}[\E[\tilde{D}_i^0|A_i=a, X_i=x]],$$
resulting in $\hat{\tau}^0(a)$ and $\hat{\tau}^1(a)$. 
The final conditional average treatment effect for each age is estimated by $$\hat{\tau}(a)=e(a)\hat{\tau}_0(a)+(1-e(a))\hat{\tau}_1(a),$$ with $e(a)\in [0,1]$ typically being the propensity score, i.e., $e(a):=Pr(W_i=1|A_i=a)$. Given the complex nature of the X-learner, it does not directly yield the Conditional Expectation Function (CEF). However, by analyzing the modified treatment variable $\hat{\tau}^w(a)$, we can discern the differences between the treated and control groups. This approach allows for a deeper understanding of the treatment effects within these distinct groups. The pseudocode of the X-learner is provided in Appendix \ref{app:B}.

Choosing between S, T, and X learners hinges on the data structure and the specific objectives of causal inference analysis. The S-learner is most effective when the dataset is relatively small or when the treatment effect is hypothesized to be uniform across observations, making it a streamlined choice for simpler causal relationships. In contrast, the T-learner excels with larger datasets that can support separate models for treated and control groups, offering a more detailed comparison of treatment effects and is preferable when the data allows for a clear segregation of treatment groups. The X-learner, however, stands out in scenarios where treatment effects are expected to be highly heterogeneous or when there is an imbalance in the number of observations between treatment groups. It is specifically tailored for complex causal inquiries where the treatment effect varies significantly across individuals or subgroups. Thus, the decision to use a specific learner model should be informed by the dataset's characteristics, including size, balance, and the anticipated variability in treatment effects. This decision-making process is underscored by a simulation study we conducted to assess the comparative effectiveness of these models in estimating ACTE.

\subsection{Base Learners}
The selection of a base learner offers considerable flexibility, allowing for the use of a simple regression algorithm or the adoption of more sophisticated supervised machine learning models. This versatility ensures that the chosen learner can be tailored to the specific needs and complexities of the data at hand.
Leveraging meta-learners' versatility, we employ here two learners. The primary model, drawn from \cite{Schuckers2023}, is the spline with a fixed effect model which achieved minimal average RMSE in their simulation study. Their study evaluates estimated curves using different regression-based models by comparing their average Root Mean Squared Error (RMSE) across different ages.

To enhance the predictability of our models, we also integrate more sophisticated methodologies, namely tree-based approaches. These,  particularly for tabular data, have demonstrated superior predictive capabilities over deep neural networks, as evidenced by \cite{grinsztajn_why_2022}, \cite{borisov_deep_2022}, and \cite{shwartz-ziv_tabular_2021}. Among these, we specifically employ the \texttt{Rforesty} package \cite{kunzel_linear_2021} to implement a fast and honest random forest model. The concept of 'honesty' in this context signifies the utilization of distinct data subsets for model training and prediction phases, thereby effectively minimizing the overfitting issue commonly associated with standard random forest approaches. This methodology ensures that our model's predictions are not only swift but also reliably accurate, setting a new standard for predictive analysis in our field.

The choice between the spline and random forest models is context-dependent. While the spline model is simpler and easier to interpret, the random forest boasts higher predictive power. Selecting the appropriate model hinges on the specific data and problem at hand.

\section{Simulation Study}
\label{sec:3}
\subsection{Simulation Design}
In causal inference, simulation studies are crucial because only one of the potential outcomes, either $Y(0)$ or $Y(1)$, can be observed. To discern the most suitable meta-learner and estimation function for $\mu$, we have designed three simulation scenarios that increase in complexity. The initial simulation assumes a constant treatment effect, with the treatment and control groups exhibiting the same quadratic shape. In the second simulation, the treatment effect grows with age, introducing different shapes for the treatment and control groups. The third and most complex simulation models a sophisticated treatment effect, resulting in distinct shapes for each group.

Building on the ideas shared in \cite{Schuckers2023}, our method sets itself apart by focusing on the causal effect, $\tau(a)$, rather than just trying to map out the age curve. This shift in focus allows us to explore the underlying causal relationships. We simulate outcomes from

$$Y_i(w)=g(a,w)+f_i(x,a,w)+\epsilon_i$$ and $$g(a,w)=\omega+\beta_1(a-a_{max})^2+\beta_2(a-a_{max})^2\mathbbm{1}(a>a_{max})+\beta_3(a-a_{max})^3\mathbbm{1}(y>t_{max})+\tau(a)w$$. 

where $i=1,...N$ and $a=\{a_1,...a_{max}\}$
and for each scenario, we vary $\tau(a)$ and $f(x,a,w)$ as follows

    \begin{enumerate}
        \item \textbf{Simulation 1 (Constant Treatment Effect)}: $$\tau(a)=2$$ 
        $$f(x,a,w)=\gamma_{i}+b_i(a-a_{max})^2\mathbbm{1}(a>a_{max})$$

        \item \textbf{Simulation 2 (Linear Treatment Effect)}: $$\tau(a) = 0.1(a-a_{min})$$ 
        $$f(x,a,w)=\gamma_{i}+b_i(a-a_{max})^2\mathbbm{1}(a>a_{max})$$

        \item \textbf{Simulation 3 (Non-linear Treatment Effect)}: $$\tau(a)=2(a - 16) + 0.0005*\mathbbm{1}(a > 20)*(a - a_{max})^3 - 0.0005*\mathbbm{1}(a > a_{max})*(a - a_{max})^4$$ 
        $$f(x,a,w)=\gamma_{i}+b_i(a-a_{max})^2\mathbbm{1}(a>a_{max})+(2+3w)x_i$$
    \end{enumerate}

In the simulation, individual random effects $\gamma_i$ and slope coefficients $b_i$ are normally distributed with mean zero and standard deviations $\sigma_{\gamma}$ and $\sigma_{\beta}$, respectively. The quadratic and cubic terms are assigned specific values: $\beta_1=-\frac{1}{9}$, $\beta_2=-\frac{6}{1000}$, and $\beta_3=\frac{45}{10000}$, reflecting the empirical shape of the age-performance relationship observed in hockey data, as discussed in \cite{Schuckers2023}. The standard deviations $\sigma_{\beta}=0.02$, $\sigma_{\epsilon}=1$, and $\sigma_{\gamma}=0.4$ are set to capture the variability inherent in the data. Treatment assignment is modeled as a Bernoulli distributed variable with a probability of $0.151$, representing the proportion of back-to-back games played across ten seasons of our cleaned data (2011-2021). Additionally, $x_i$ in scenario 3 is a covariate from a uniform distribution on the interval $(-1,1)$

\begin{figure}
    \centering
    \includegraphics[scale=0.32]{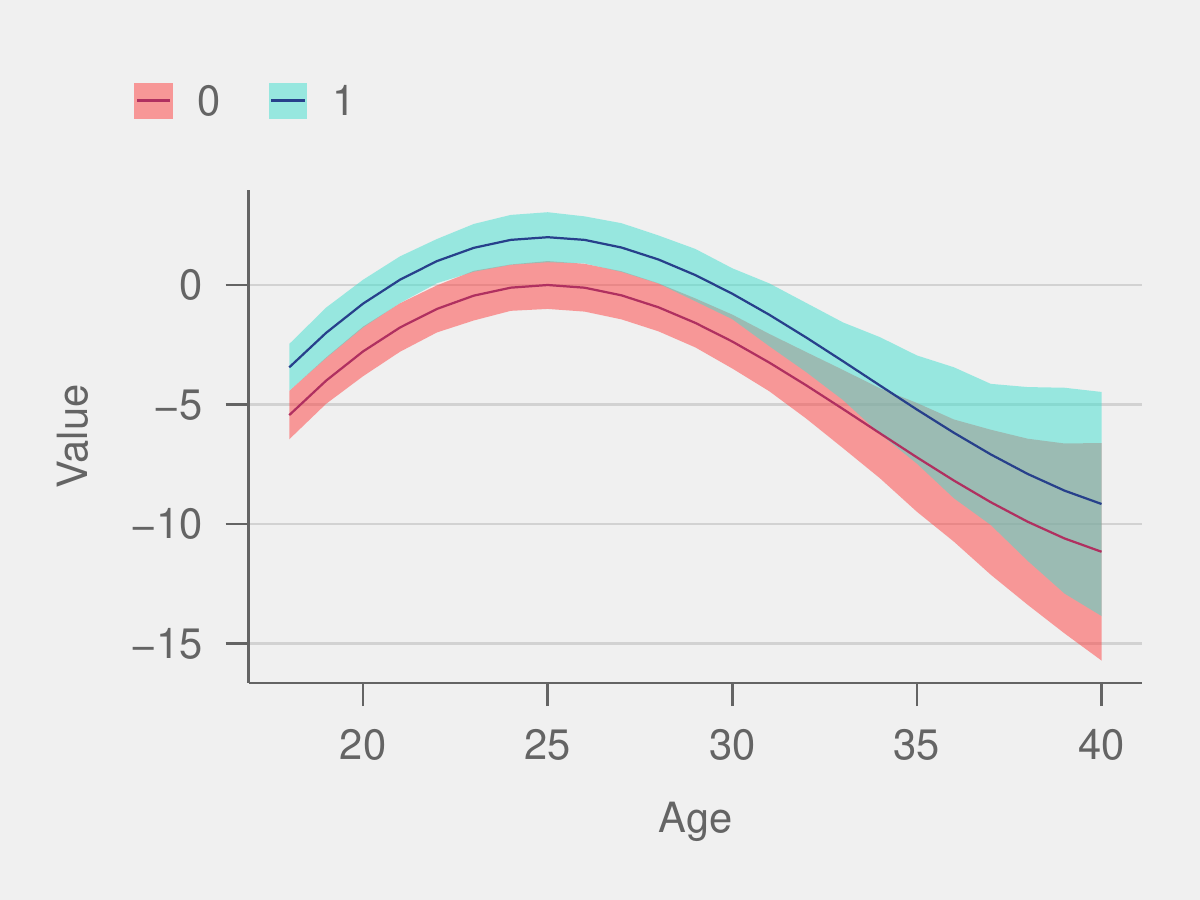}
    \includegraphics[scale=0.32]{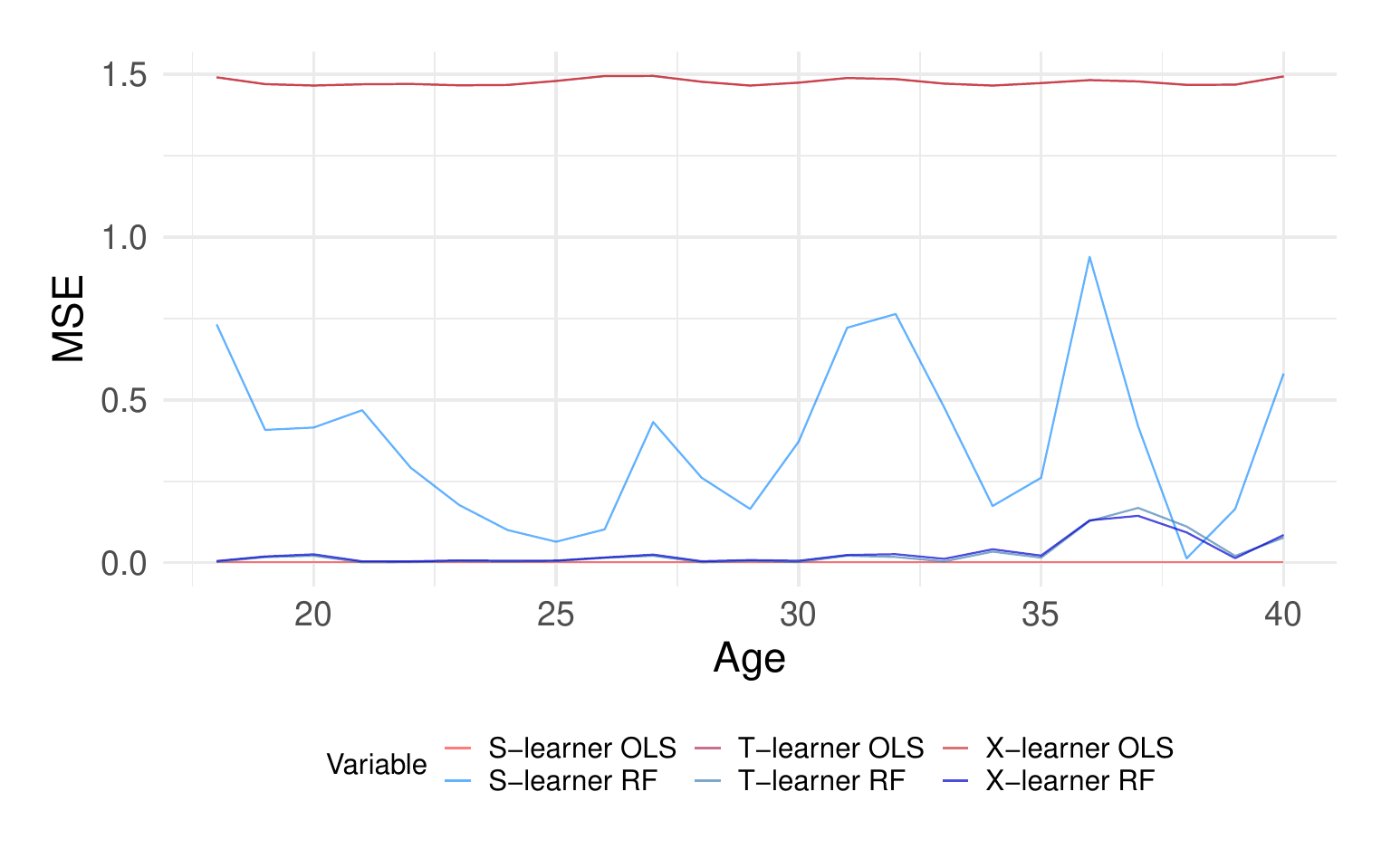}
    \caption{Simulation Scenario 1: Analysis of Constant Treatment Effect - Potential Outcome Functions for Control and Treated Groups (Left), MSE in ACTE Estimation (Right).}
    \label{fig:sim1}
\end{figure}

\subsection{Simulation Results}
We compute the $\tau(a)$ for a total of six combinations. Three from meta-learning algorithm, S, T, and X learner and two models ols-spline and random forest. Then we calculated the MSE to assess the performance of each combination.

Figures \ref{fig:sim1}, \ref{fig:sim2}, \ref{fig:sim3} present the MSE by age, highlighting that in Simulation 1, S-learner with OLS yields a minimal error, aligning with model assumptions due to both groups' constant shape derived from a spline family. In Simulation 2, the T-learner with RF emerges as the top performer, capturing the distinct age-curve shapes. For Simulation 3, the intricate X-learner with RF attains the lowest MSE. Notably, S-learner with OLS is less effective for younger and older ages, as it assumes the same shape for controlled and treated age curves. While S-learner with RF improves upon this by differentiating treatment in its tree splits especially for younger players, it still falls short in encapsulating the full complexity of the group shapes.

In Table \ref{tab:simulatio_result}, it is evident that the optimal method varies by scenario. For complex treatment effects, the X-learner paired with a random forest is superior, whereas, for simpler effects with parallel trends in treatment and control groups, the S-learner with ols-spline is preferable. Researchers must establish well-informed assumptions and choose their models judiciously. In practice, due to the fundamental problem of causal inference, the unobservability of counterfactuals, actual MSE cannot be computed.

\begin{figure}
    \centering
    \includegraphics[scale=0.32]{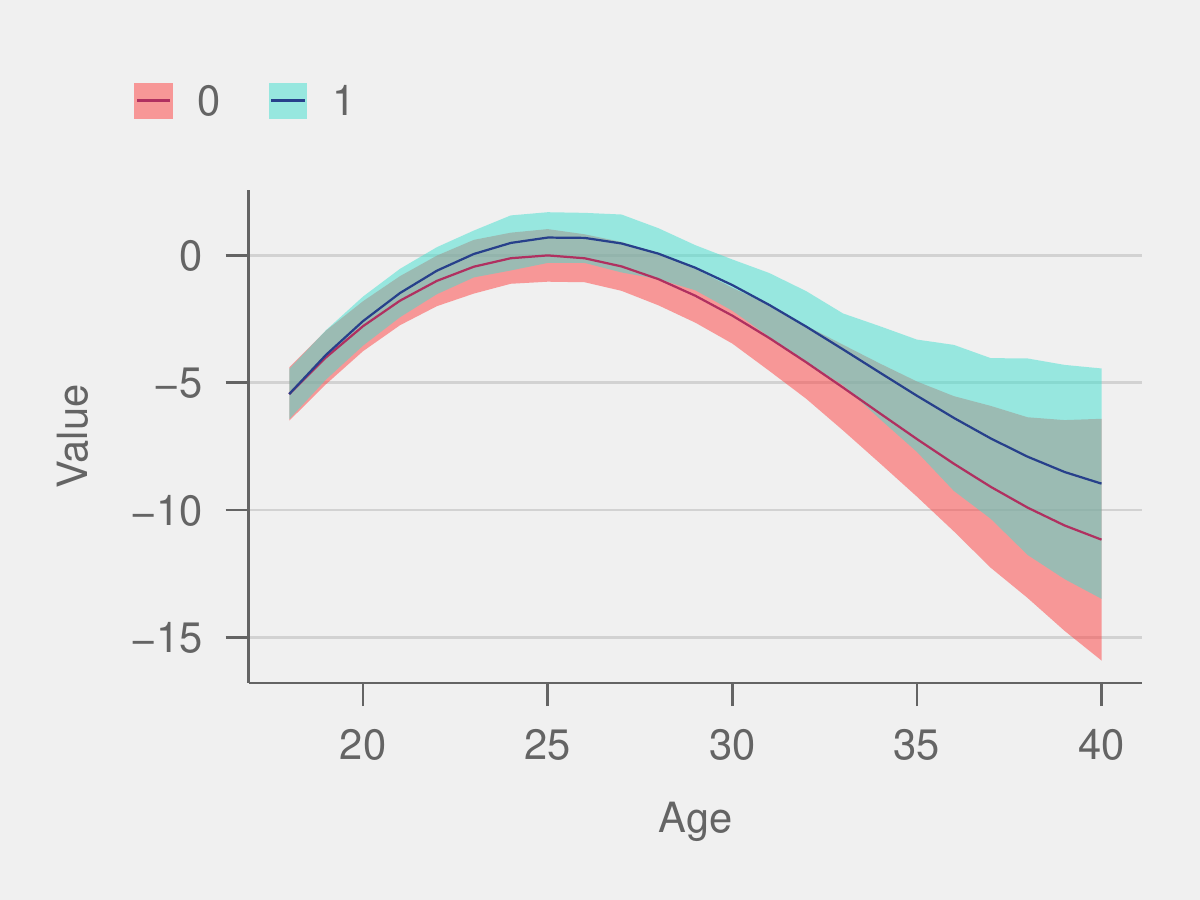}
    \includegraphics[scale=0.32]{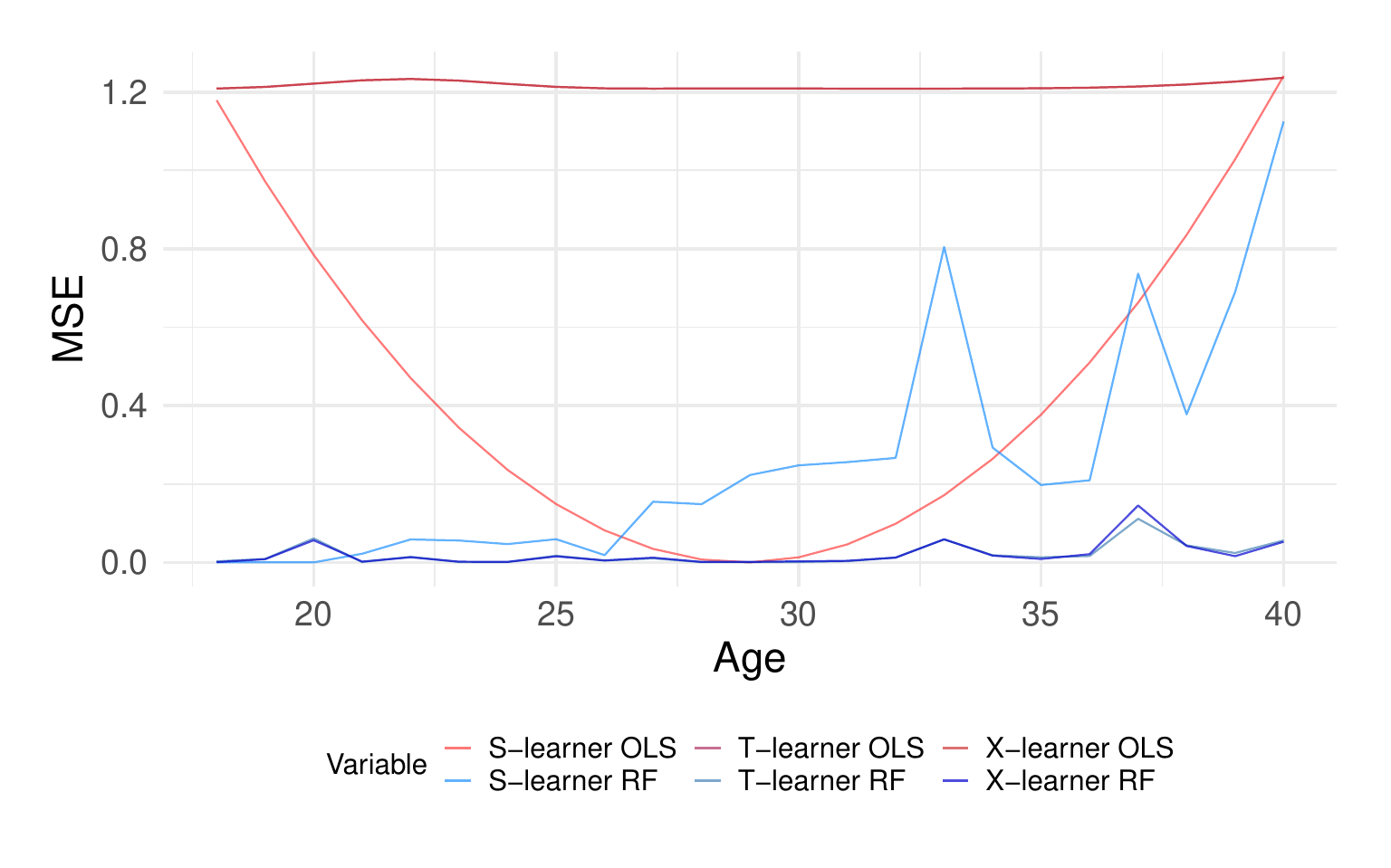}
    \caption{Simulation Scenario 2: Linear Treatment Effect - Potential Outcome Functions for Control and Treated Groups (Left), MSE in ACTE Estimation (Right).}
    \label{fig:sim2}
\end{figure}

\begin{figure}
    \centering
    \includegraphics[scale=0.32]{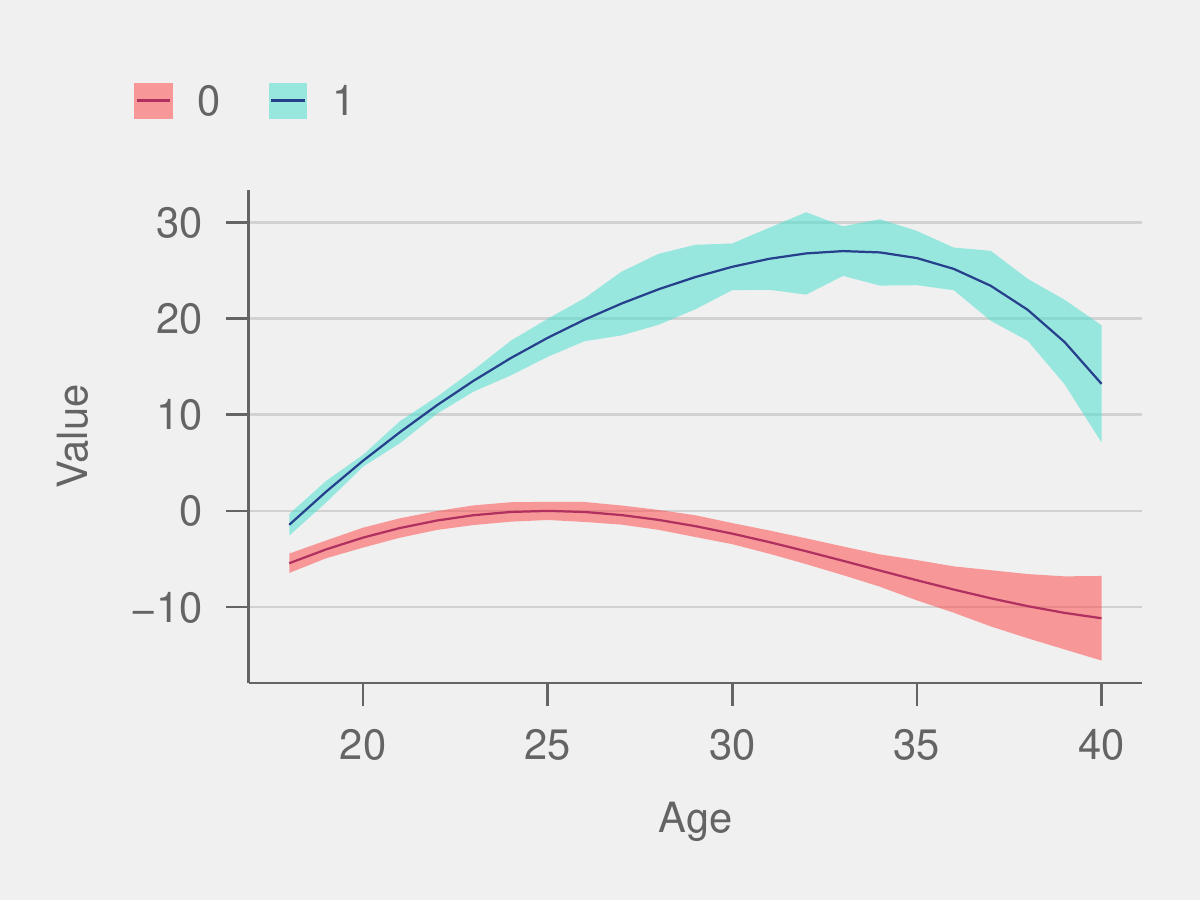}
    \includegraphics[scale=0.32]{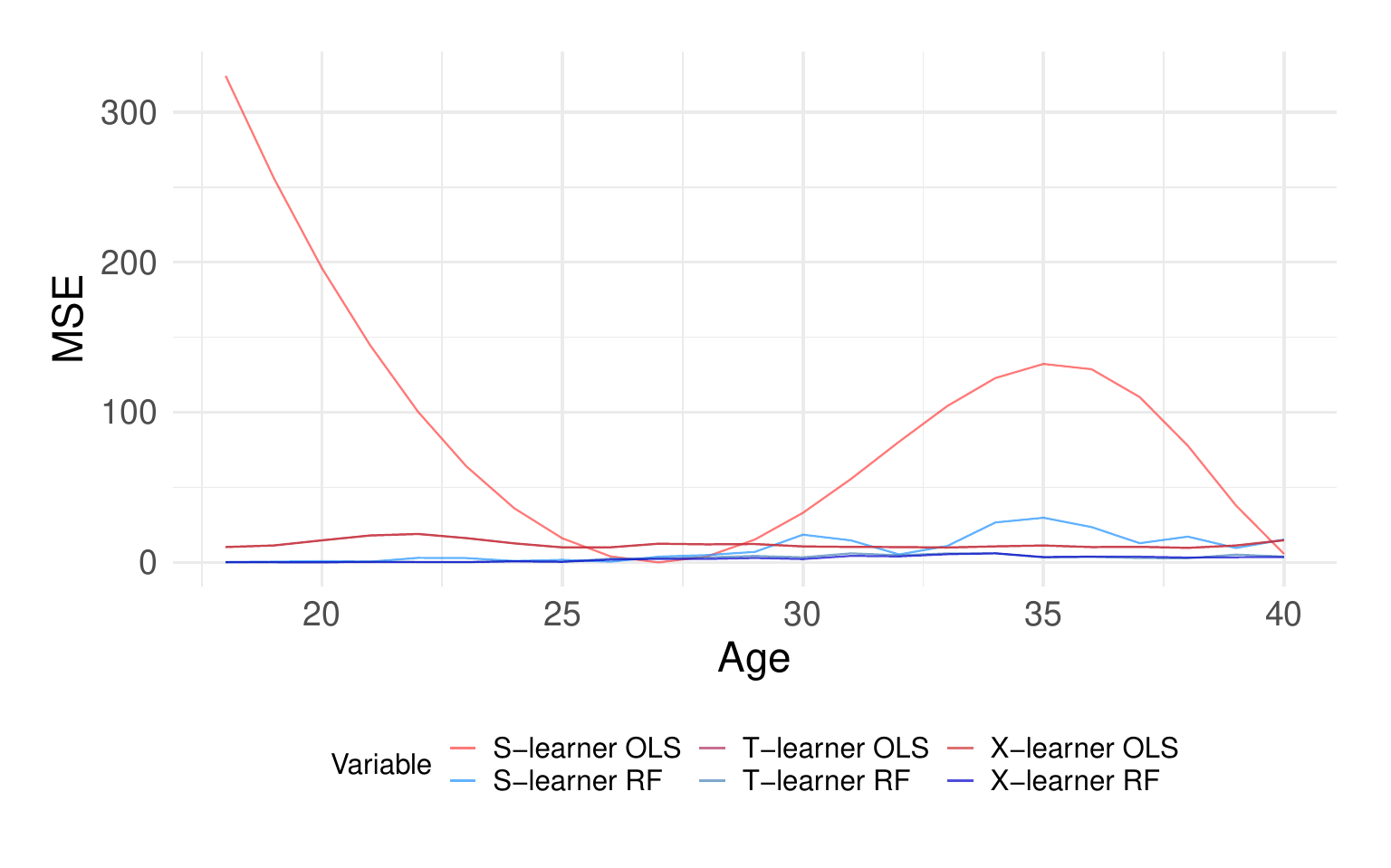}
    \caption{Simulation Scenario 3: Non-linear Treatment Effect - Potential Outcome Functions for Control and Treated Groups (Left), MSE in ACTE Estimation (Right).}
    \label{fig:sim3}
\end{figure}

\begin{table}[ht]
\centering
\begin{tabular}{lrrr}
  \toprule
    Model & simulation1 & simulation2 & simulation3 \\ 
  \midrule
  s.ols & \textbf{0.00} & 0.44 & 89.23 \\ 
  t.ols & 1.53 & 1.44 & 9.40 \\ 
  x.ols & 1.53 & 1.44 & 9.40 \\ 
  s.rf & 0.35 & 0.29 & 9.90 \\ 
  t.rf & 0.05 & \textbf{0.07} & 4.17 \\ 
  x.rf & 0.07 & 0.08 & \textbf{3.74}\\ 
   \bottomrule
\end{tabular}
\caption{MSE for ACTE estimation} 
\label{tab:simulatio_result}
\end{table}

\section{Application to National Basketball Association data}
\label{sec:4}
Load management in the NBA, a strategy increasingly adopted for injury prevention and career longevity, also serves to optimize player performance. Kawhi Leonard's 2018-2019 season with the Toronto Raptors is a prime example. By participating in only 60 of the 82 regular-season games, Leonard was in top form for the playoffs, leading the Raptors to their inaugural championship. Similarly, Gregg Popovich of the San Antonio Spurs masterfully applied this approach with veteran Tim Duncan, preserving Duncan's performance levels and contributing to the Spurs' repeated championship successes. This strategic balance between player health and team achievement exemplifies the dual benefits of load management.

Critiques of load management point to its impact on the regular season's significance, suggesting that teams might not prioritize these games, thereby diminishing fan engagement. This critique led to the NBA's implementation of the Player Participation Policy for the 2023-2024 season, requiring star players to participate more frequently, with restrictions on resting multiple star players simultaneously. This highlights the need for quantifying the impact of rest on both player performance and the overall competitiveness of the season.

Despite its growing acceptance, quantitative analysis of load management's precise impact remains scarce, particularly in comparison to the effects of playing in back-to-back games. In this study, we apply the ACTE framework to measure the impact of rest on players, segmented by age. While a randomized controlled trial (RCT) would be the ideal method for such analysis, the practical challenges and high costs make it infeasible in the context of an NBA season. Therefore, our approach utilizes observational studies within this framework. While acknowledging the inherent limitations and assumptions of observational studies, we believe this approach offers a valuable starting point for applying causal machine learning to age-related performance curves, providing insightful results into the effects of load management.

\subsection{Data}
Leveraging ten seasons (2011-2022) of NBA data obtained through the \texttt{hoopR} package \cite{hoopR}, we utilize the discussed framework to analyze how days of rest and age influence various box score statistics for NBA players. To ensure the relevance of our data, we only consider instances where players were active for at least 25 minutes in the previous game, thereby excluding minimal play times. After the data processing, we have total of 827 players. In our analysis, the treatment variable $w=0$ signifies games played on consecutive days (back-to-back, b2b), while $w=1$ represents having at least one day of rest. The potential outcome under treatment $w$ is denoted as $Y(w)$, with $X$ representing the covariate matrix. We define $\mu_w(a,x)$ as the conditional expectation of the potential outcome $Y(w)$, given the player's age $a$. This calculation facilitates the construction of the age curve. Through $\tau(a, x)$, we ascertain the ACTE, quantifying the impact of rest versus participating in b2b games at specific ages. Given the observational nature of our study, completely eliminating confounders might not be feasible. However, we aim to mitigate their influence by adjusting for variables such as player, team, opponent team, home or away games, and season-fixed effects. Our analysis primarily focuses on players aged between 18 and 39, as the treatment group's sample size outside this range is significantly limited, as shown in table \ref{table:age_table}.

\begin{table}[tbh]
    \centering
    \caption{Number of games played across age and treatment status}
    \label{table:age_table}
    \begin{tabular}{lll}
  \toprule
   Age & b2b & non-b2b\\
   \midrule
18 & 0 (0.00\%) & 5 (100.00\%) \\ 
  19 & 266 (14.22\%) & 1604 (85.78\%) \\ 
  20 & 679 (14.66\%) & 3952 (85.34\%) \\ 
  21 & 917 (14.49\%) & 5412 (85.51\%) \\ 
  22 & 1312 (15.47\%) & 7168 (84.53\%) \\ 
  23 & 1734 (16.32\%) & 8892 (83.68\%) \\ 
  24 & 1795 (15.97\%) & 9446 (84.03\%) \\ 
  25 & 1715 (15.14\%) & 9616 (84.86\%) \\ 
  26 & 1850 (15.66\%) & 9964 (84.34\%) \\ 
  27 & 1785 (16.04\%) & 9344 (83.96\%) \\ 
  28 & 1529 (14.94\%) & 8706 (85.06\%) \\ 
  29 & 1221 (14.98\%) & 6930 (85.02\%) \\ 
  30 & 984 (14.25\%) & 5921 (85.75\%) \\ 
  31 & 867 (14.46\%) & 5129 (85.54\%) \\ 
  32 & 633 (13.71\%) & 3984 (86.29\%) \\ 
  33 & 464 (13.98\%) & 2855 (86.02\%) \\ 
  34 & 316 (13.98\%) & 1944 (86.02\%) \\ 
  35 & 218 (13.53\%) & 1393 (86.47\%) \\ 
  36 & 128 (12.04\%) & 935 (87.96\%) \\ 
  37 & 69 (11.54\%) & 529 (88.46\%) \\ 
  38 & 39 (13.18\%) & 257 (86.82\%) \\ 
  39 & 22 (10.38\%) & 190 (89.62\%) \\ 
  40 & 3 (5.66\%) & 50 (94.34\%) \\ 
  41 & 0 (0.00\%) & 5 (100.00\%) \\ 
  42 & 0 (0.00\%) & 10 (100.00\%) \\ 
   \bottomrule
    \end{tabular}
    \label{tabel:age_table}
\end{table}

\subsection{Average Conditional Expectation Function (ACEF)}
In this section, we will analyze the ACEF for S and T learners using ols-spline and random forest.

In all ACEF plots, a finer line denotes the $\hat{\mu}_w(a)$, which represents the estimated expectation given age. We also employ a smoothing spline function, fitted through a generalized additive model (GAM) with a smoothness degree of six, as suggested by \cite{Schuckers2023}. This approach helps us understand the general age-curve trend for both the treatment and control groups, represented by a solid line. The shaded regions show the 90\% bootstrap confidence intervals for each age group. Non-overlapping shaded areas provide statistical evidence of the impact of rest on subsequent game performance. For a detailed discussion on the confidence interval calculation, refer to Appendix \ref{app:B}.

First, we analyze the net, offensive, and defensive ratings. Net, Offensive, and Defensive Ratings are crucial statistics in basketball analytics. Offensive Rating measures a player's or team's efficiency in scoring, calculated as the number of points produced per 100 possessions; a higher offensive rating indicates better offensive performance. Conversely, a defensive rating assesses the ability to prevent opponents from scoring, with points allowed per 100 opponent possessions; a lower defensive rating means stronger defense. Net Rating is the difference between offensive and defensive rating, offering a holistic view of overall impact. A positive net rating implies a player or team contributes more offensively than they concede defensively, while a negative value suggests a need for improvement in scoring efficiency, defensive effectiveness, or both.

In Figures \ref{fig:sols_rating}, \ref{fig:srf_rating}, \ref{fig:tols_rating}, and \ref{fig:trf_rating}, the Conditional Expectation Function (CEF) for net, offensive, and defensive ratings is displayed. Across all three metrics, the effects are significantly pronounced with age except for older players where we lack sample size.

With the S-learner using OLS-Spline, most effects show significance. Yet, this model is the most constrained, potentially hiding strong intrinsic bias. It operates under the assumption of uniform treatment effects across ages, which can be mitigated only by incorporating two-way interaction terms between covariates and treatment. The S-learner with Random Forest (RF) addresses this limitation, as trees adeptly capture nonlinear relationships between treatment and covariates. However, this model reveals a non-significant effect on offensive ratings for several ages. One hypothesis for this observation could be that older players, in back-to-back games, emphasize offense over defense. Defensively, rotating younger players to cover key offensive opponents might mitigate physical fatigue. Another contributing factor could be the limited data on older and younger players. As Table \ref{table:age_table} indicates, maintaining a career in the competitive NBA is challenging, leading to less precise inferences for this demographic.

The T-learner results offer more insights, particularly regarding older players, due to its complexity. Being more intricate than the S-learner, the T-learner demands a larger sample size. While the general conclusions for both OLS-Spline and RF mirror those from the S-learner with RF, it becomes more apparent that rest impacts the defensive aspect more than the offensive end.

\begin{figure}
    \centering
    \includegraphics[scale=0.2]{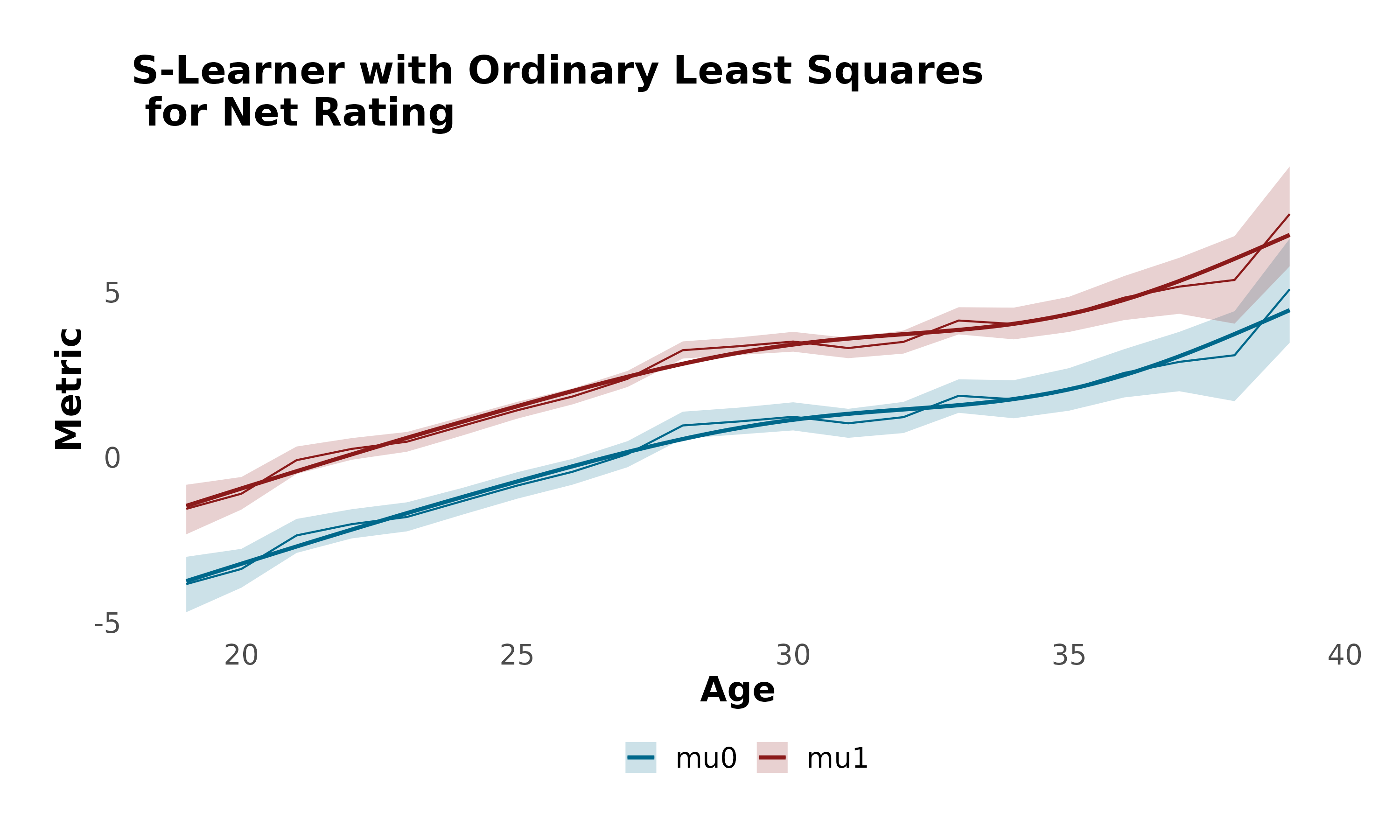} 
    \includegraphics[scale=0.2]{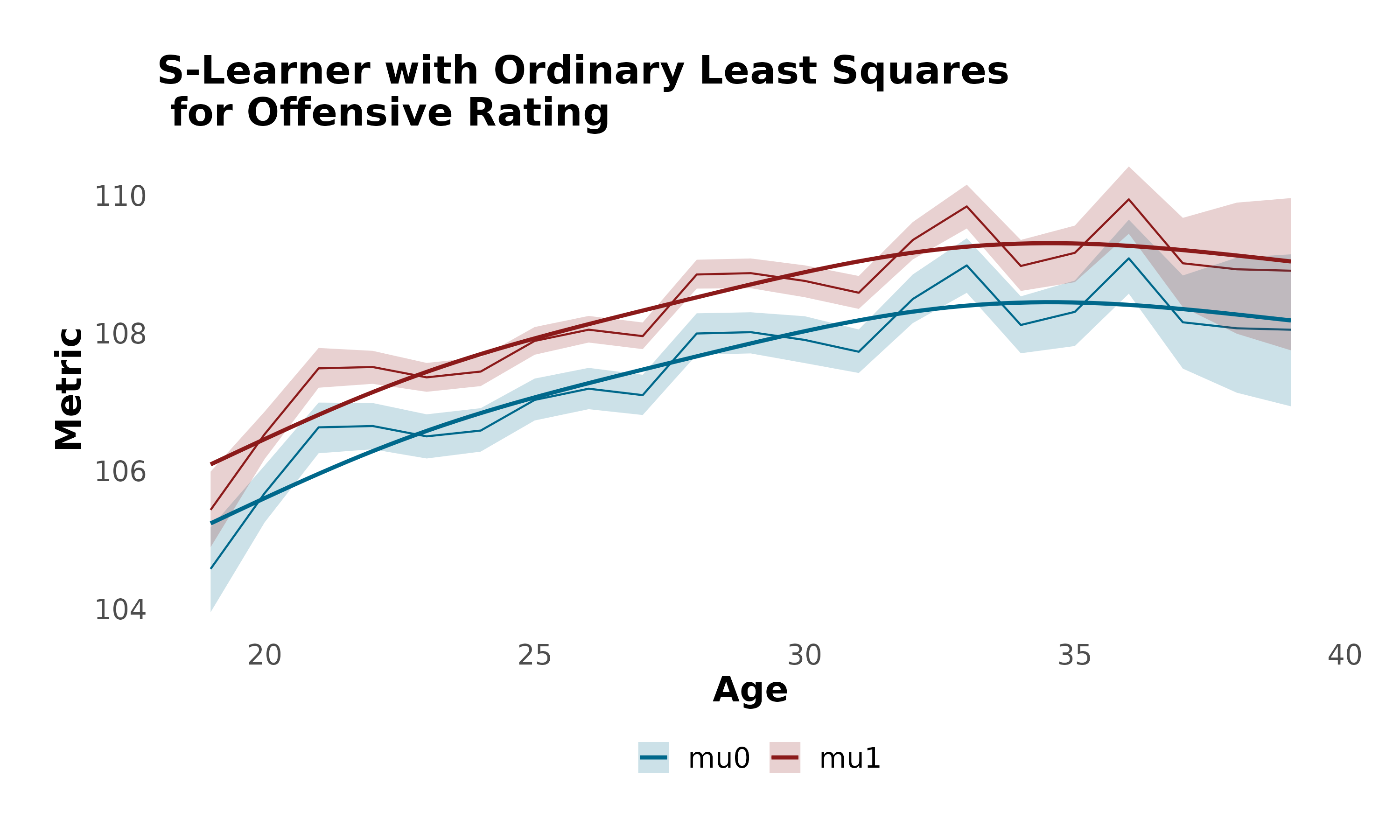}
    \includegraphics[scale=0.2]{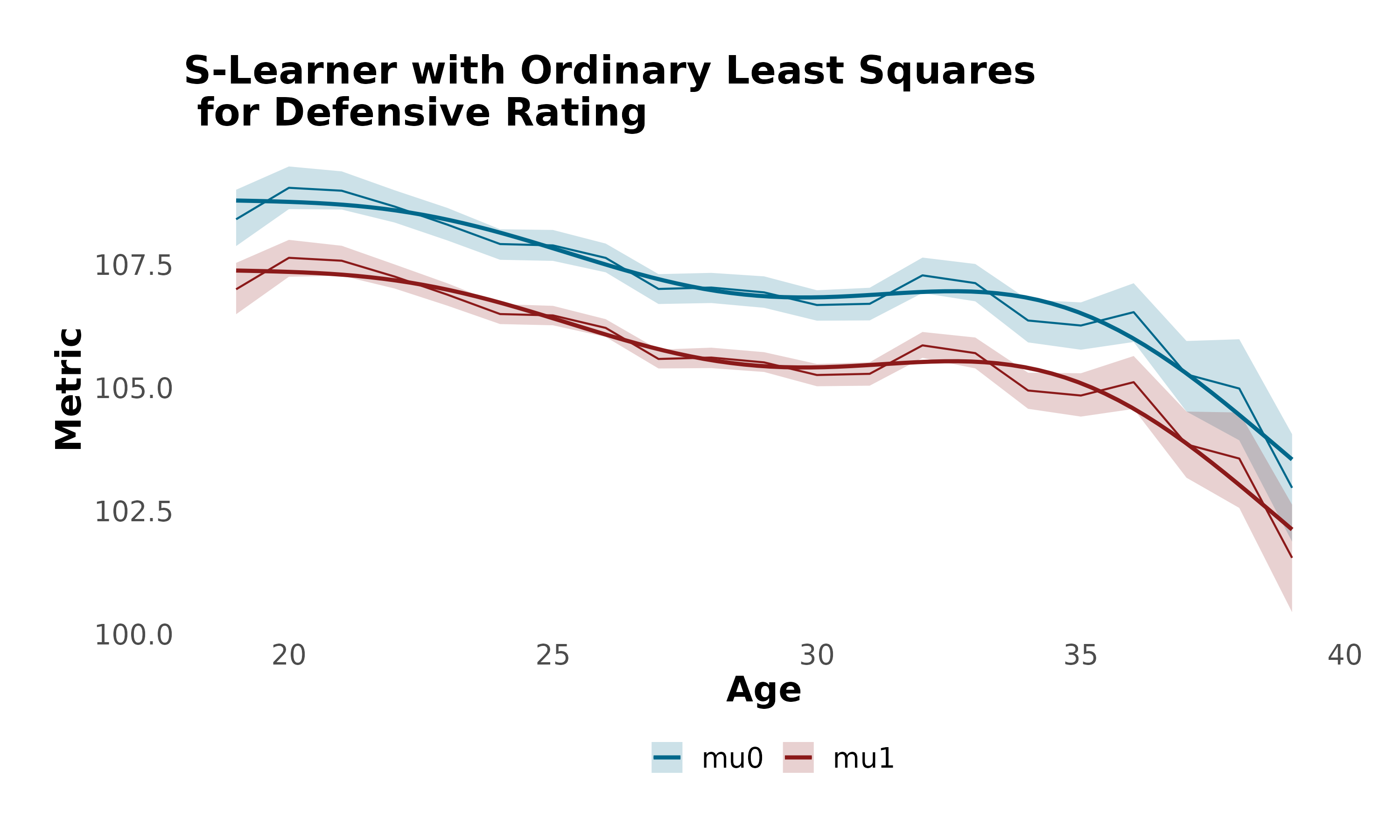}
    \caption{CEF using S-learner and OLS-Spline for net, offensive, and defensive rating}
    \label{fig:sols_rating}
\end{figure}

\begin{figure}
    \centering
    \includegraphics[scale=0.2]{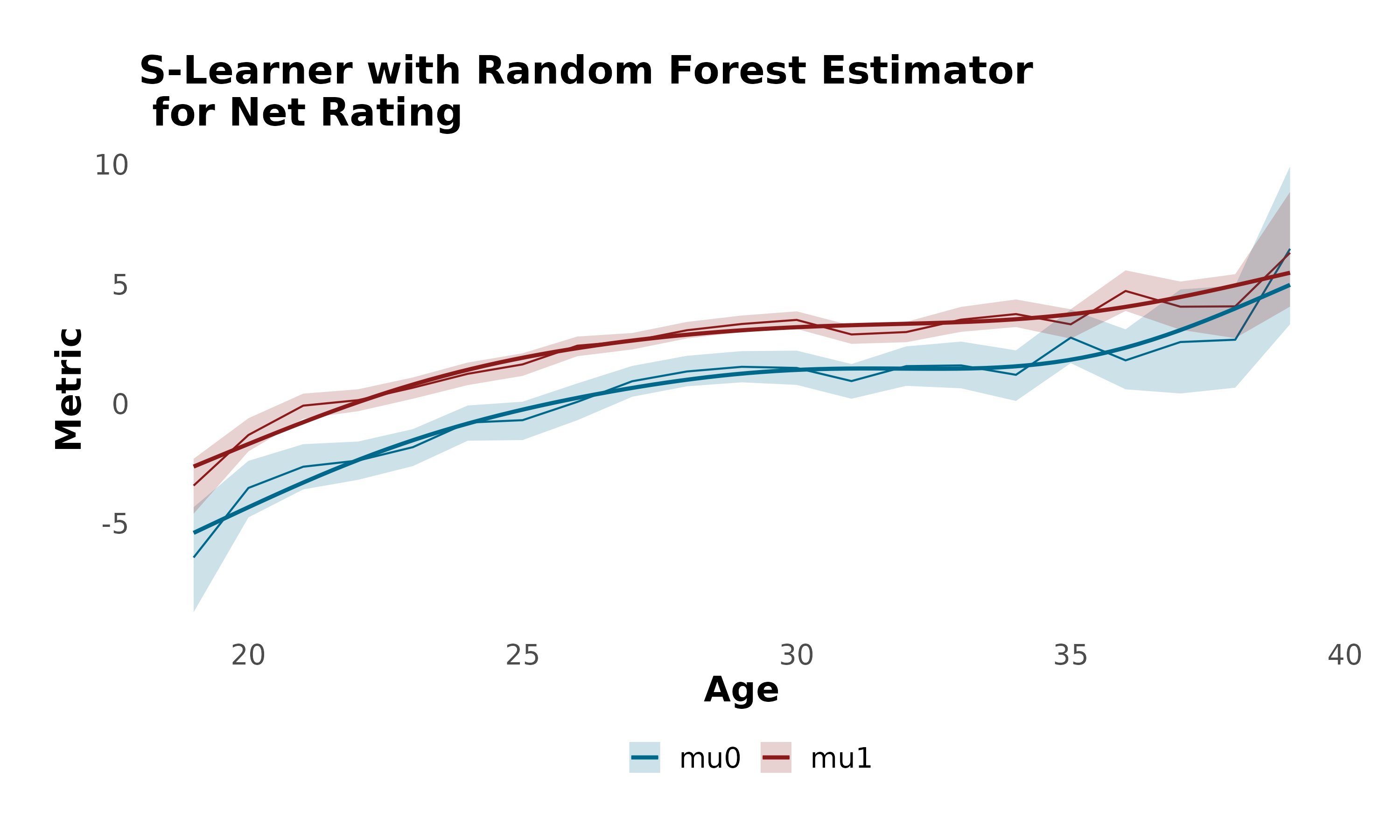} 
    \includegraphics[scale=0.2]{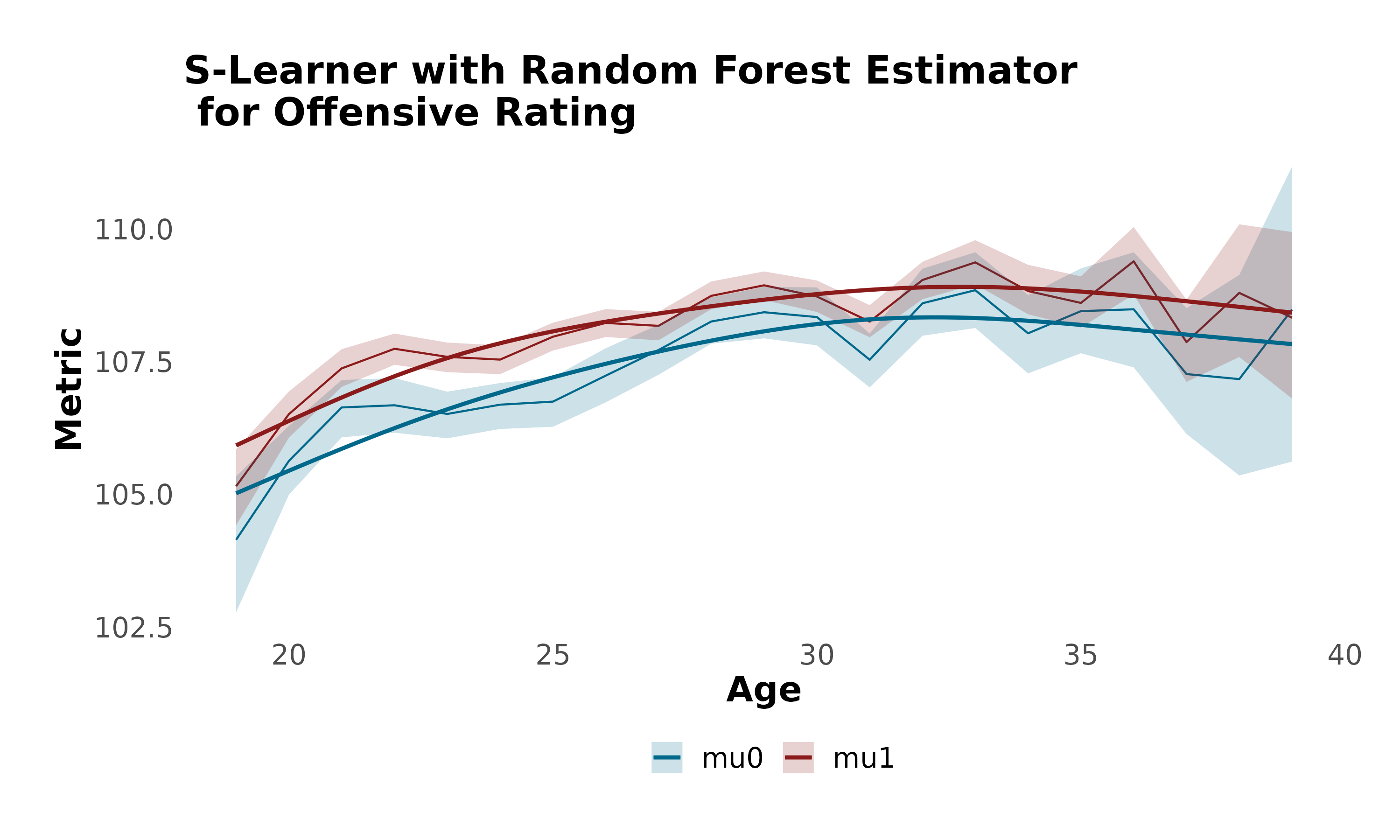}
    \includegraphics[scale=0.2]{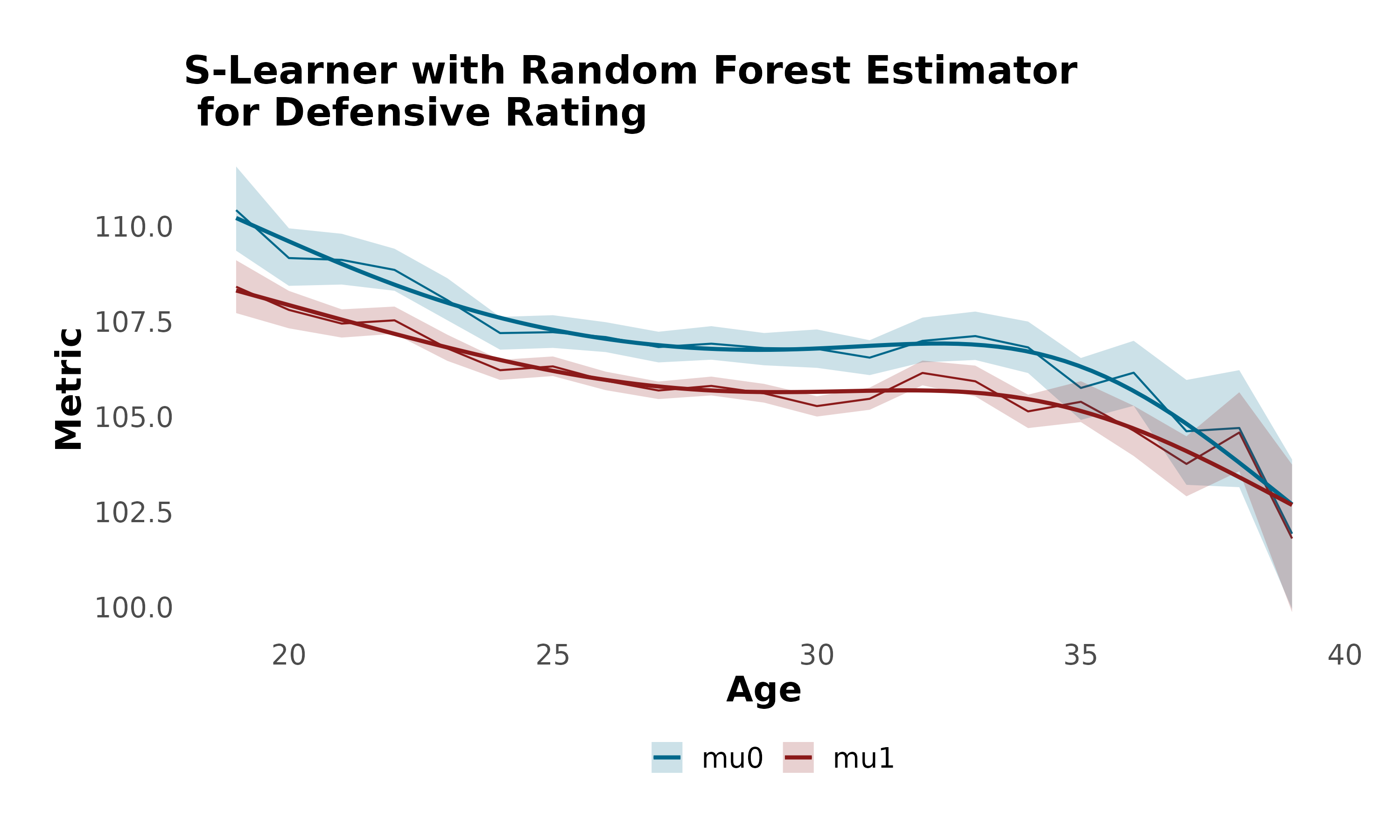}
    \caption{CEF using S-learner and RF for net, offensive, and defensive rating}
    \label{fig:srf_rating}
\end{figure}

\begin{figure}
    \centering
    \includegraphics[scale=0.2]{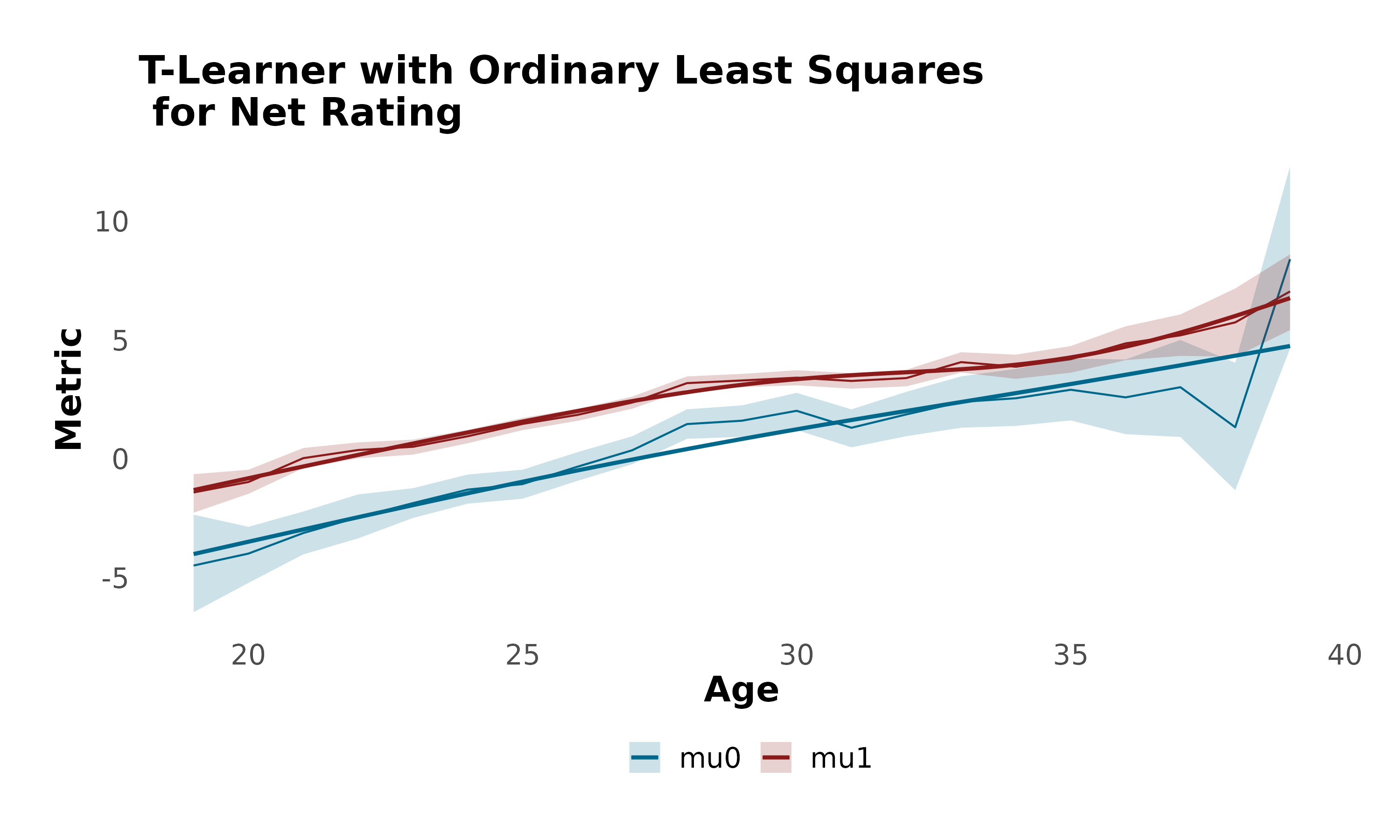} 
    \includegraphics[scale=0.2]{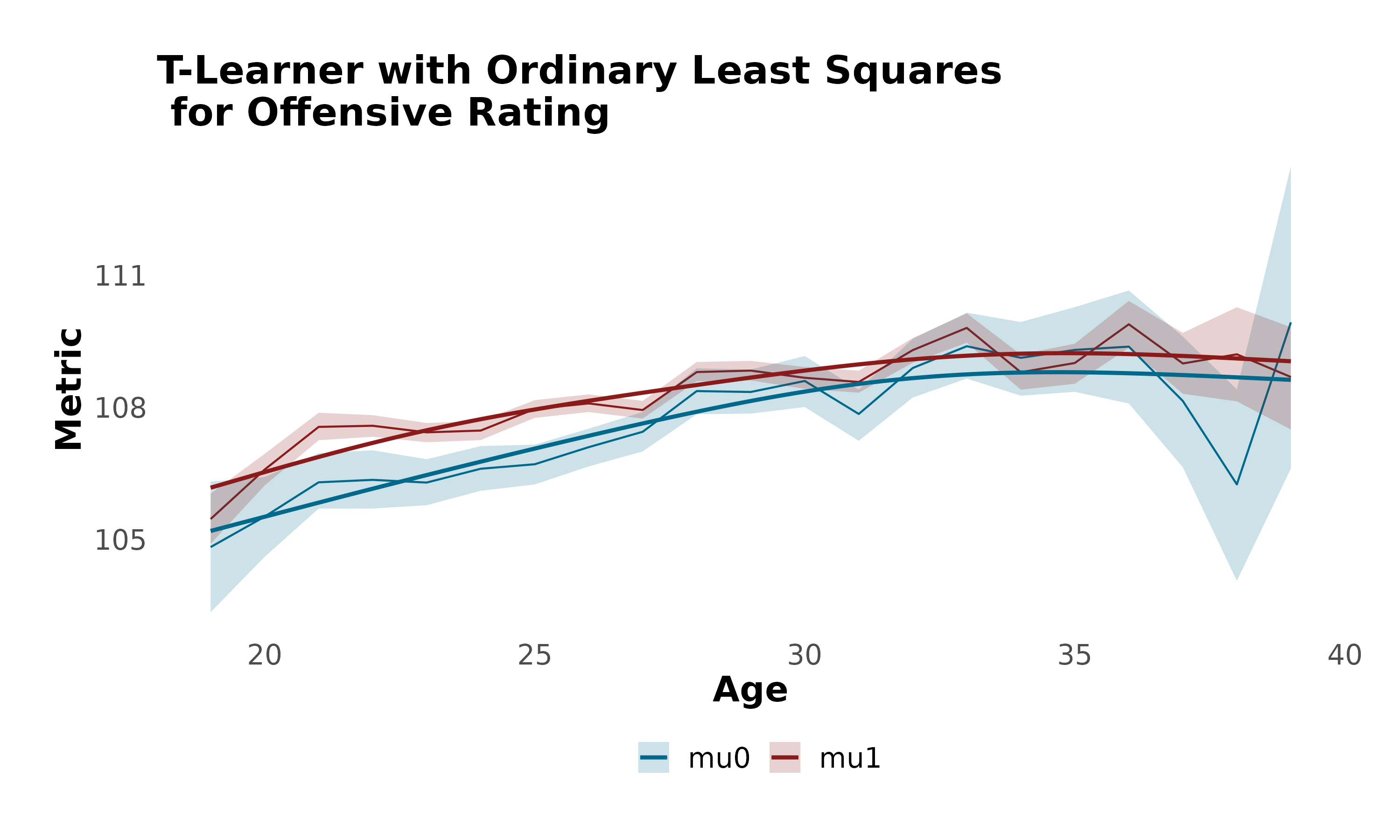}
    \includegraphics[scale=0.2]{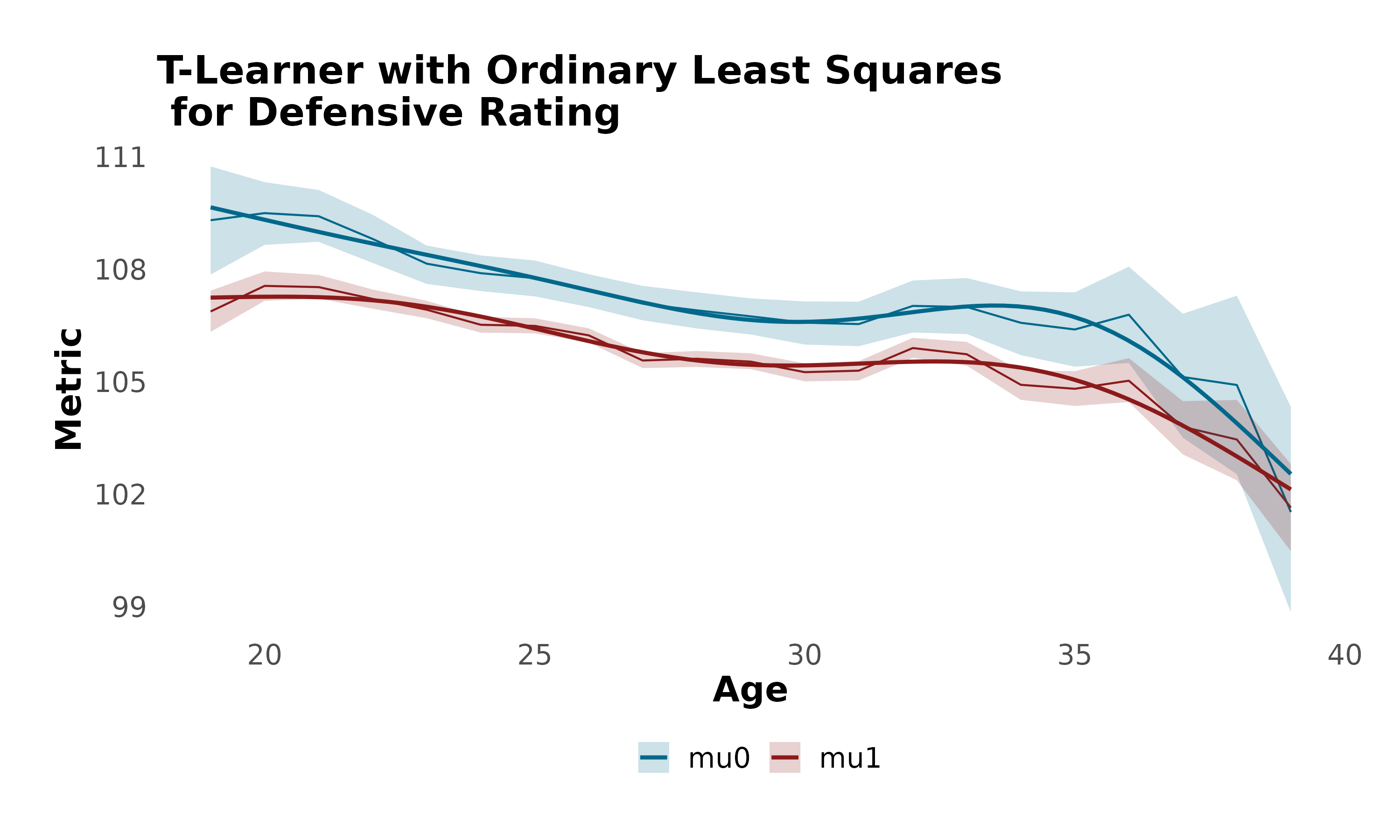}
    \caption{CEF using T-learner and OLS-Spline for net, offensive, and defensive rating}
    \label{fig:tols_rating}
\end{figure}

\begin{figure}
    \centering
    \includegraphics[scale=0.2]{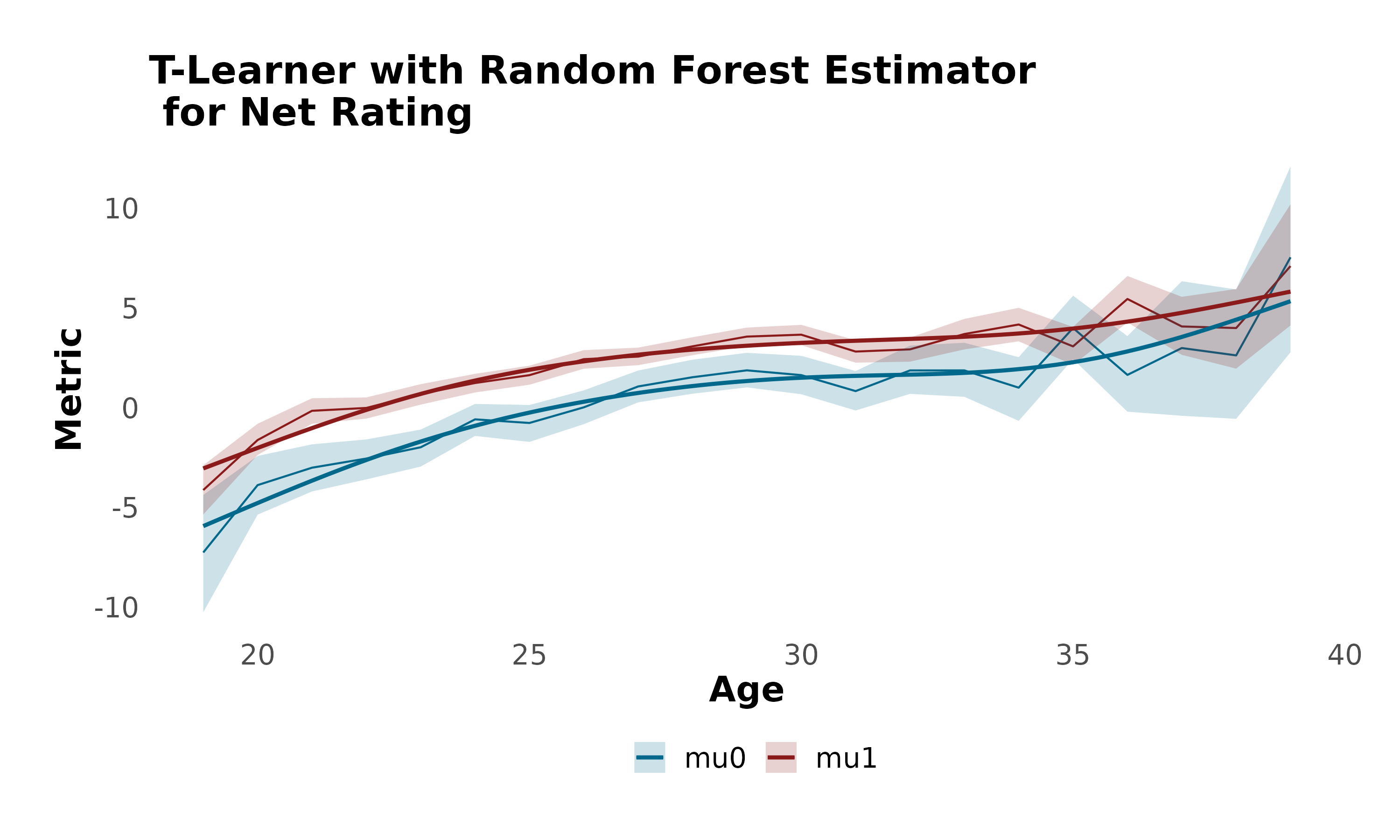} 
    \includegraphics[scale=0.2]{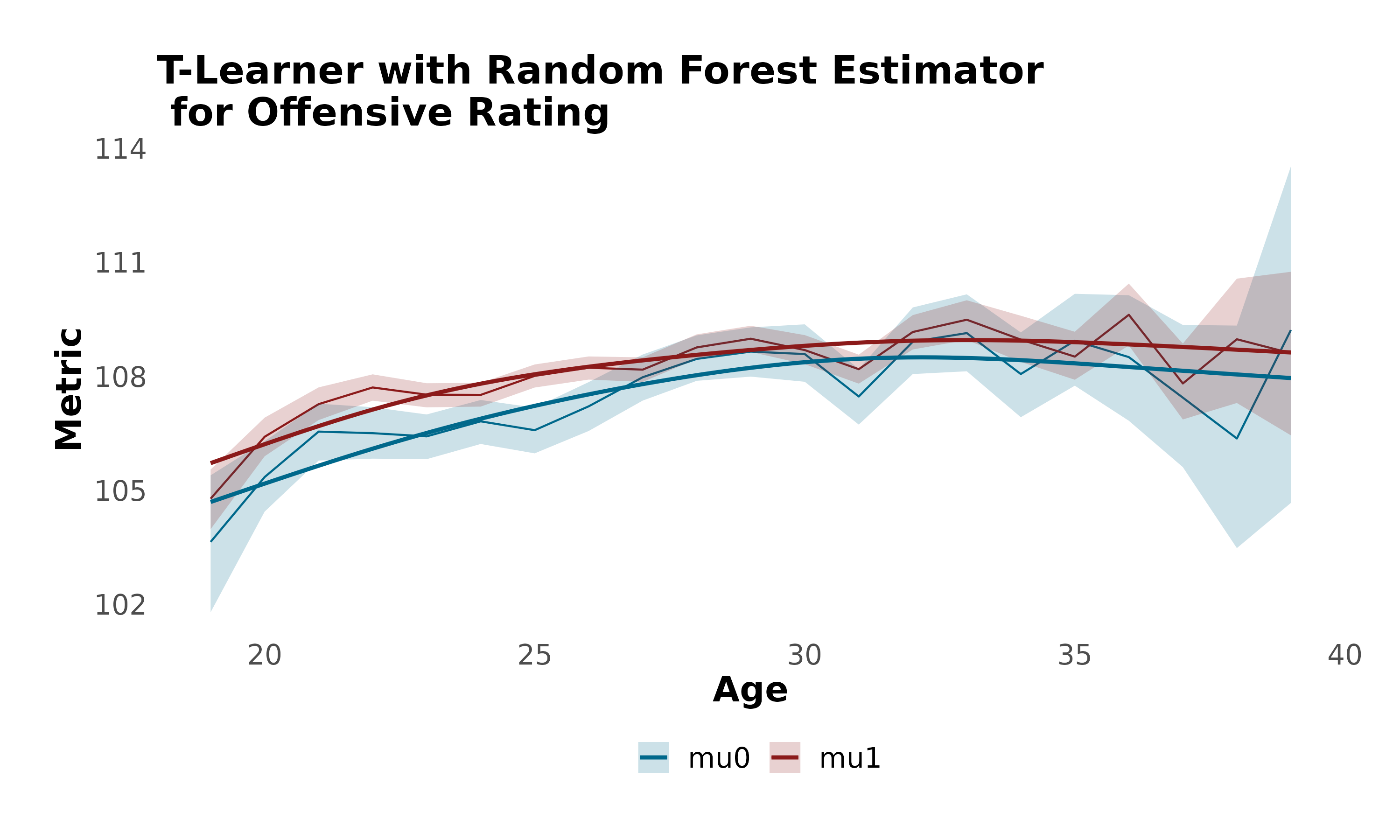}
    \includegraphics[scale=0.2]{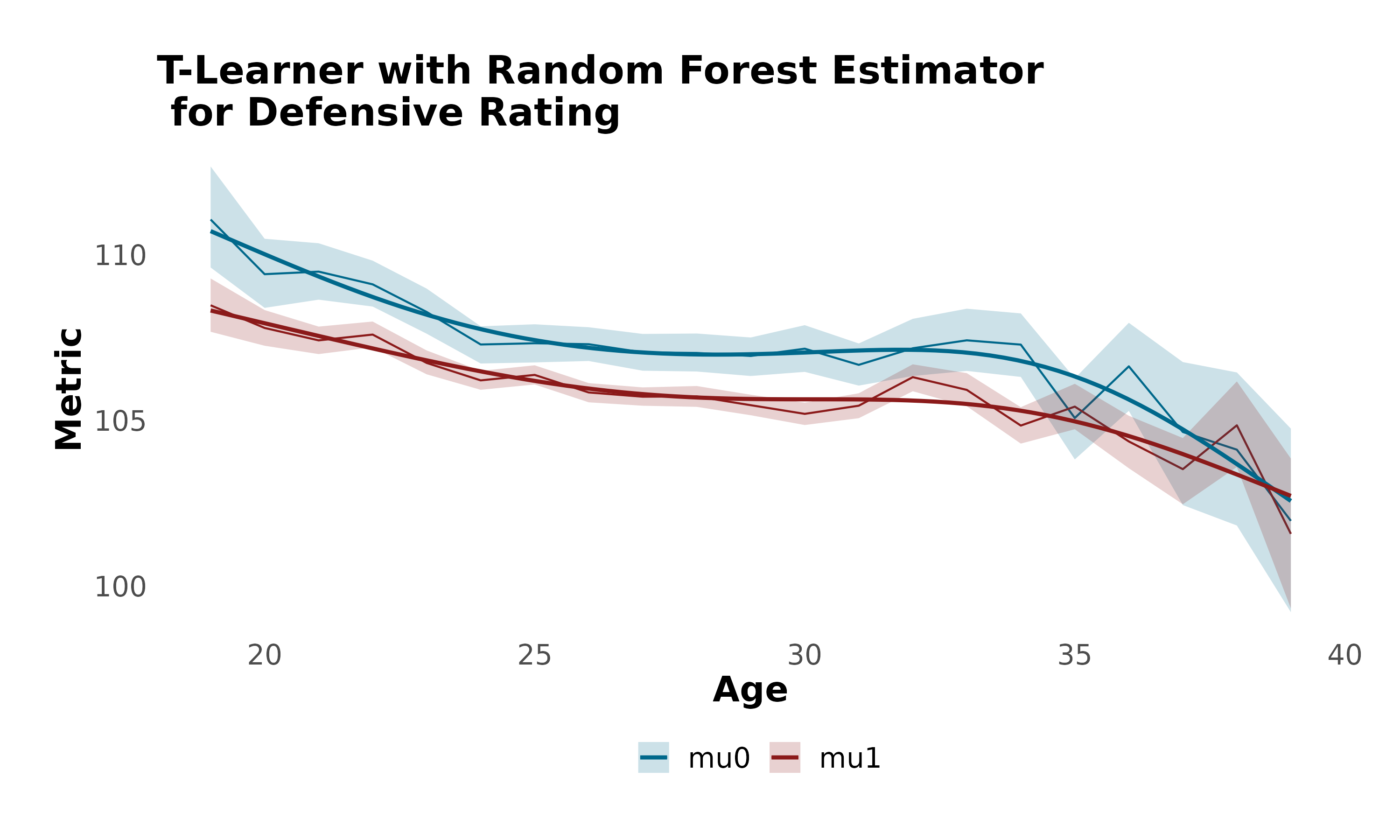}
    \caption{CEF using T-learner and RF for net, offensive, and defensive rating}
    \label{fig:trf_rating}
\end{figure}

We now turn our attention to analyzing box score statistics in the NBA, a detailed set of metrics that measure a player's game performance. These include points, assists, rebounds (divided into offensive and defensive), steals, blocks, turnovers, and fouls. Assists tally a player's contribution to teammates' scoring, while steals and blocks are key defensive indicators of disrupting the opponent's play. Turnovers and fouls, conversely, denote player errors and infractions. Such statistics are vital for assessing both individual and team play, offering a quantitative recap of crucial game elements. Our analysis particularly focuses on points, assists, total rebounds, steals, blocks, and turnovers.

To ensure equitable comparisons among players with varying game minutes, we examine these statistics normalized per 100 possessions. This approach aids in more accurately interpreting the respective trends and patterns in player performance across different game scenarios.

In Figures \ref{fig:sols_box}, \ref{fig:srf_box}, \ref{fig:tols_box}, and \ref{fig:trf_box}, we display the Conditional Expectation Function (CEF) for the previously mentioned box score statistics. Notably, steals are the only statistic showing a stronger effect than other metrics.

In the analysis using the S-learner, it becomes clear that steals are the only statistic consistently demonstrating a positive effect. This pattern suggests that rest could significantly influence steals, likely due to the physical nature required for successful execution. Although the effect on turnovers is not marked by statistical significance, the notable U-shaped curve indicates that prime-aged players generally incur fewer turnovers compared to their counterparts. Blocks exhibit a similar U-shaped trend. However, this might be skewed by Tim Duncan's extraordinary blocking in his later career years. As shown in Table \ref{table:block_table}, Duncan's performance predominantly drives this trend, with the smaller sample size in his age bracket possibly leading to the unusual pattern in blocks. For assists, the older the players get they improve their assist skills. One hypothesis for this is that when the players get old, instead of being physical and attacking the rim, veteran players prefer to dish the ball taking advantage of their experience consistent players rendering high assists are mostly Steve Nash and Chris Paul as we can observe in the table \ref{table:assist_table}.
Moreover, the assist category is an area where treatment has a minimal impact, primarily because it involves less physical contact compared to other statistics.

The T-learner largely aligns with the conclusions drawn from other models, but it uniquely captures the rest effect on blocks when we use random forest, particularly among younger players. This model highlights that rest has a significant impact on this demographic. Across all four combinations analyzed, the conclusions are similar, indicating consistent trends. This consistency is especially notable in most box score statistics, where the T-learner reinforces the patterns observed in other estimators, emphasizing the influence of rest on certain aspects of player performance.

\begin{table}[tbh]
    \centering
    \caption{Top 10 blocks per 100 possessions by players above 35 years old}
    \begin{tabular}{llll}
  \toprule
   Player & age & b100 & date\\
   \midrule
    Tim Duncan & 36.00 & 12.07 & 2013-01-13 \\ 
  Tim Duncan & 39.00 & 10.91 & 2015-11-18 \\ 
  Tim Duncan & 37.00 & 10.53 & 2014-02-03 \\ 
  Marcus Camby & 37.00 & 10.20 & 2012-02-01 \\ 
  Tim Duncan & 36.00 & 9.80 & 2012-11-28 \\ 
  Tim Duncan & 36.00 & 9.43 & 2013-03-03 \\ 
  Tim Duncan & 36.00 & 9.09 & 2013-02-13 \\ 
  Tim Duncan & 36.00 & 8.96 & 2012-12-12 \\ 
  Marcus Camby & 37.00 & 8.96 & 2012-03-30 \\ 
  Tim Duncan & 36.00 & 8.93 & 2013-01-16 \\ 
   \bottomrule
    \end{tabular}
    \label{table:block_table}
\end{table}


\begin{table}[tbh]
    \centering
    \caption{Top 10 assists per 100 possessions by players above 35 years old}
    \begin{tabular}{llll}
  \toprule
   Player & age & a100 & date\\
   \midrule
    Steve Nash & 37.00 & 30.91 & 2012-01-24 \\ 
  Chris Paul & 36.00 & 27.59 & 2012-02-04 \\ 
  Kobe Bryant & 36.00 & 27.42 & 2012-03-05 \\ 
  Steve Nash & 38.00 & 26.15 & 2012-12-12 \\ 
   Steve Nash & 38.00 & 25.86 & 2015-11-28 \\ 
   Chris Paul & 36.00 & 25.42 & 2013-12-02 \\ 
   Steve Nash & 37.00 & 25.42 & 2014-01-24 \\ 
   Chris Paul & 36.00 & 25.00 & 2022-02-25 \\ 
   Steve Nash & 38.00 & 25.00 & 2012-04-13 \\ 
   Steve Nash & 38.00 & 25.00 & 2013-11-15 \\ 
   \bottomrule
    \end{tabular}
    \label{table:assist_table}
\end{table}

\begin{figure}
    \centering
    \includegraphics[scale=0.2]{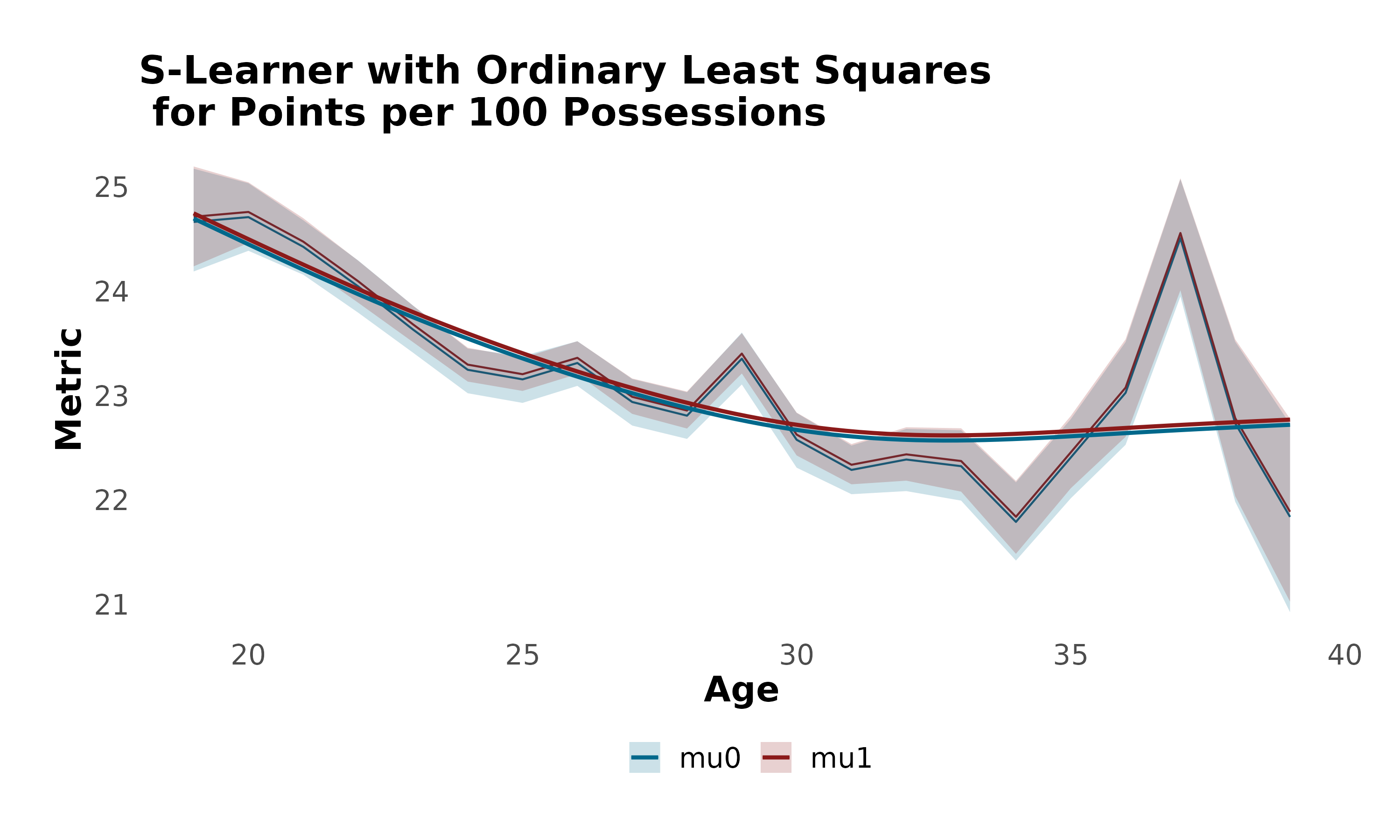} \includegraphics[scale=0.2]{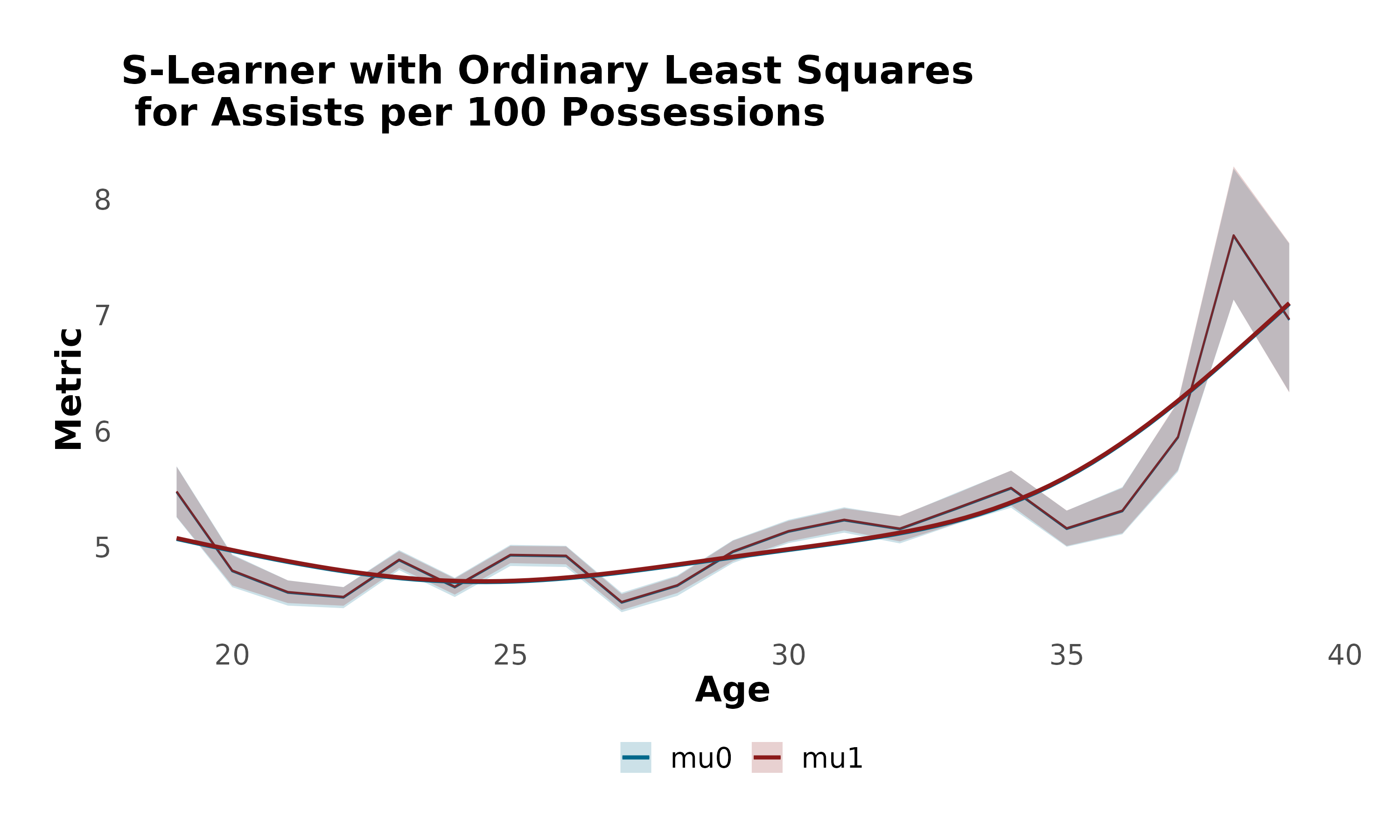}
    \includegraphics[scale=0.2]{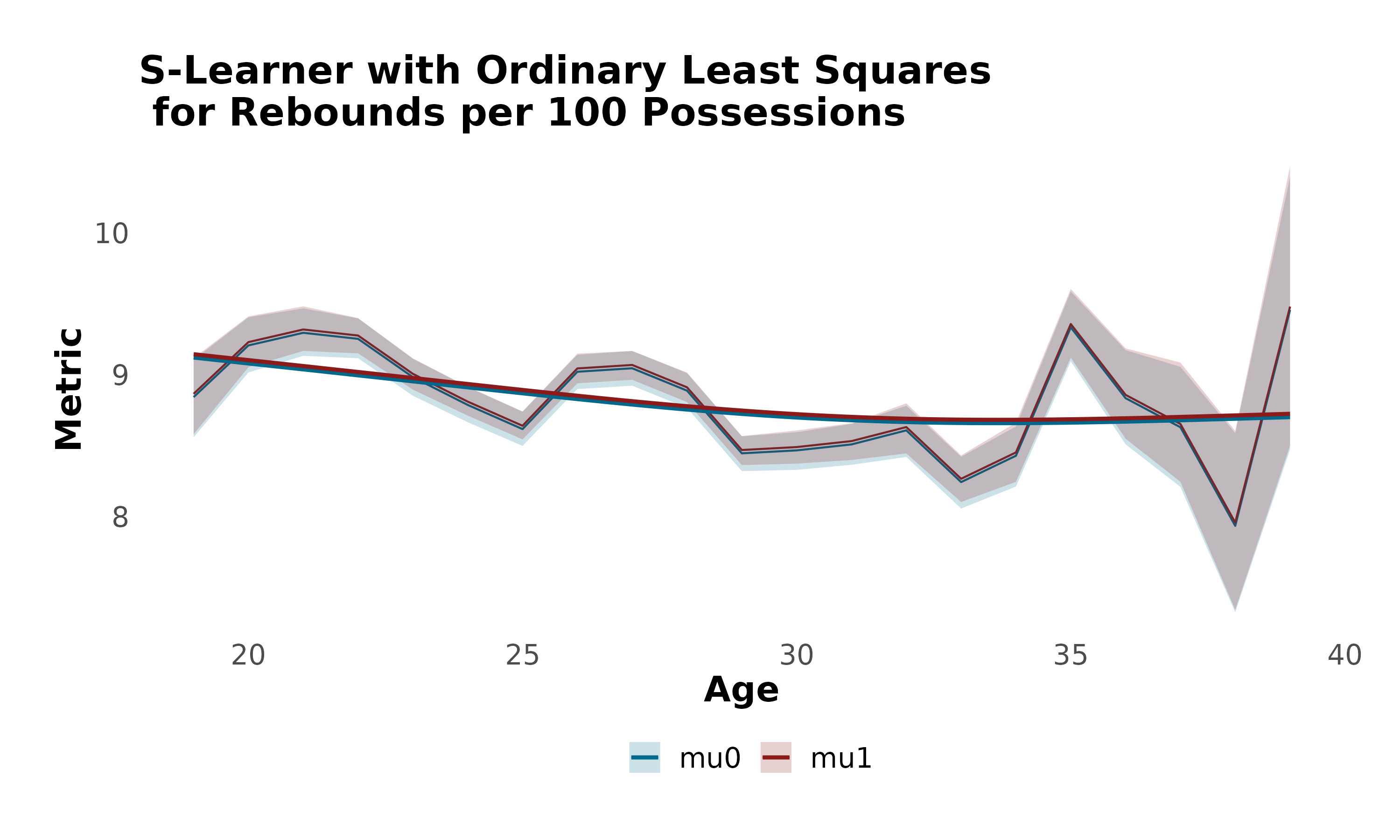}
    \includegraphics[scale=0.2]{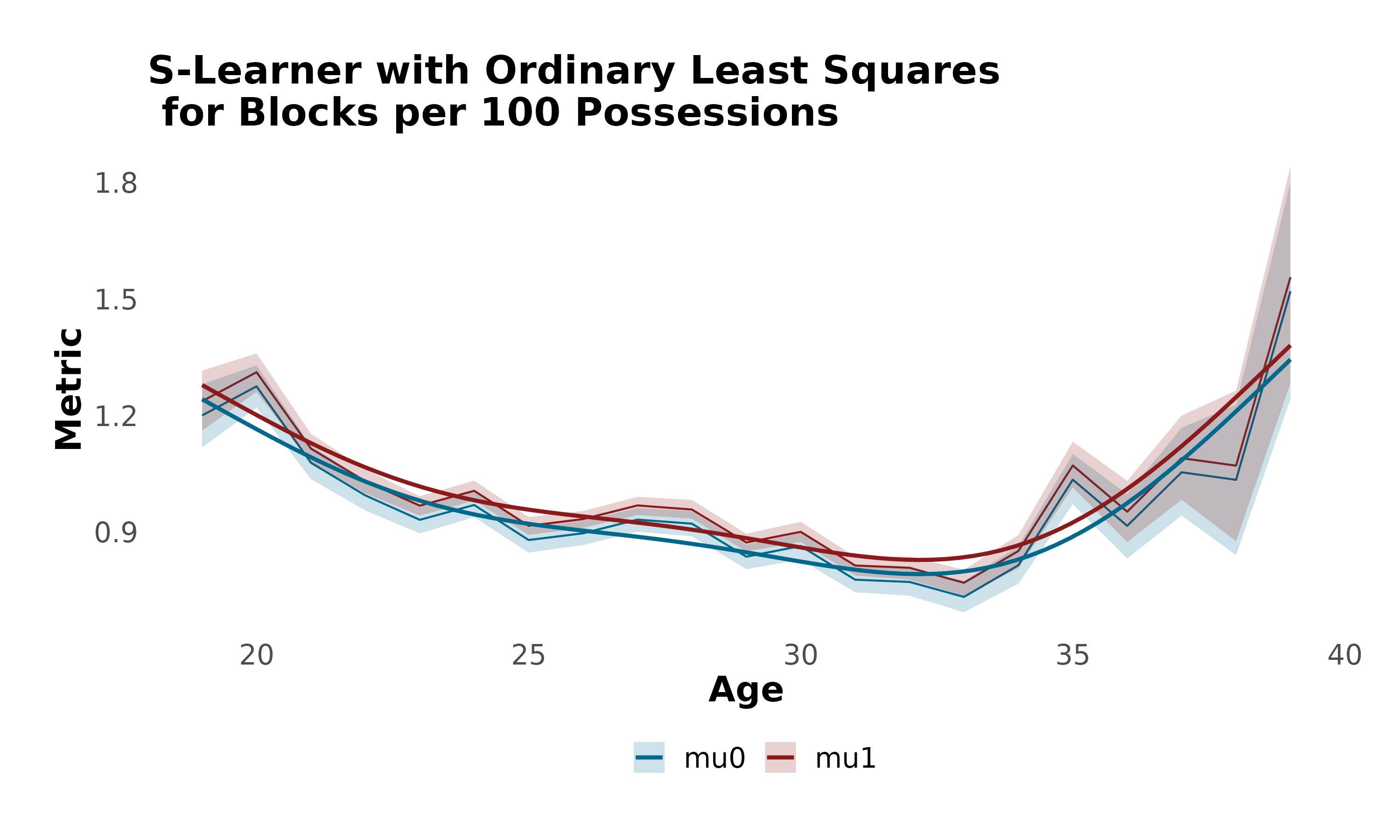}
    \includegraphics[scale=0.2]{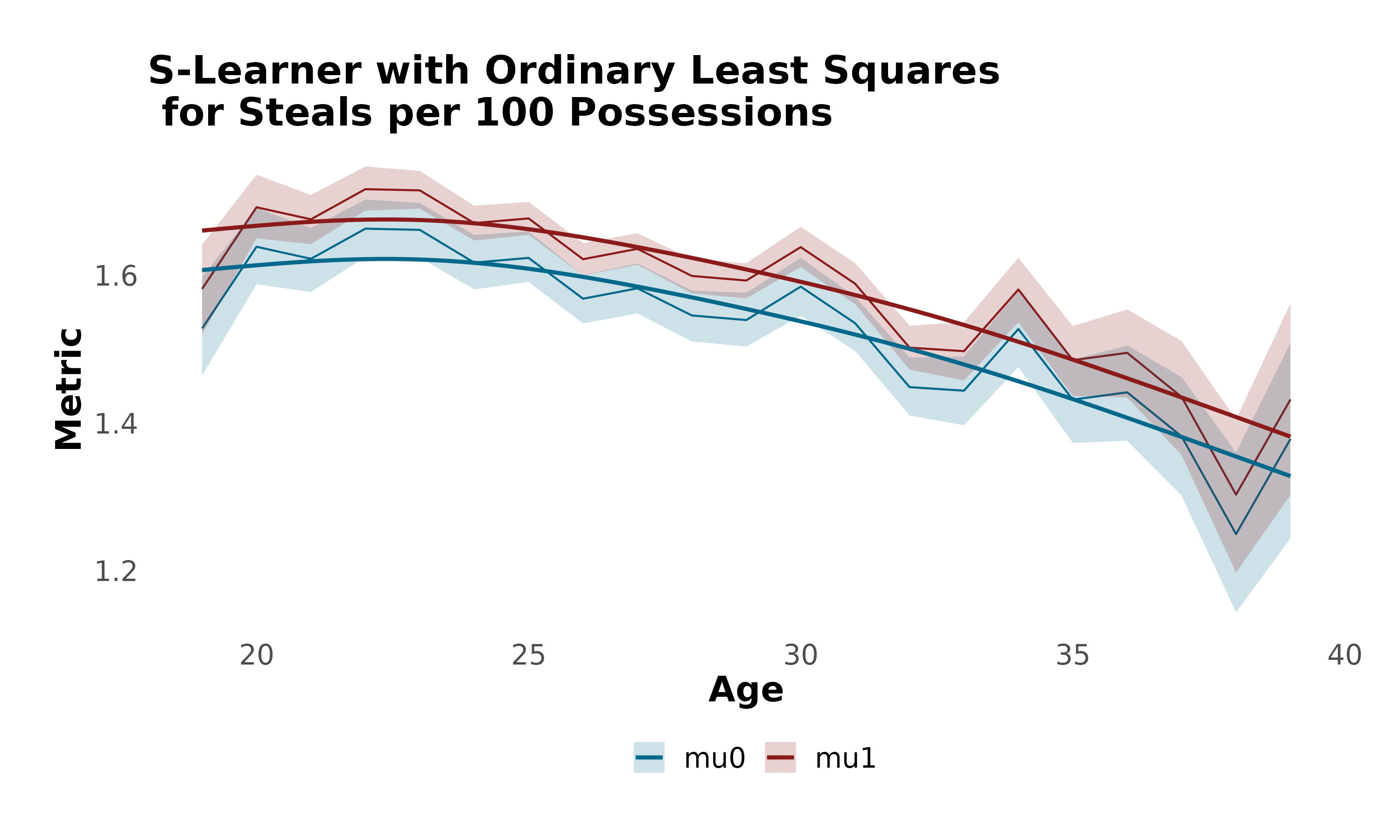}
    \includegraphics[scale=0.2]{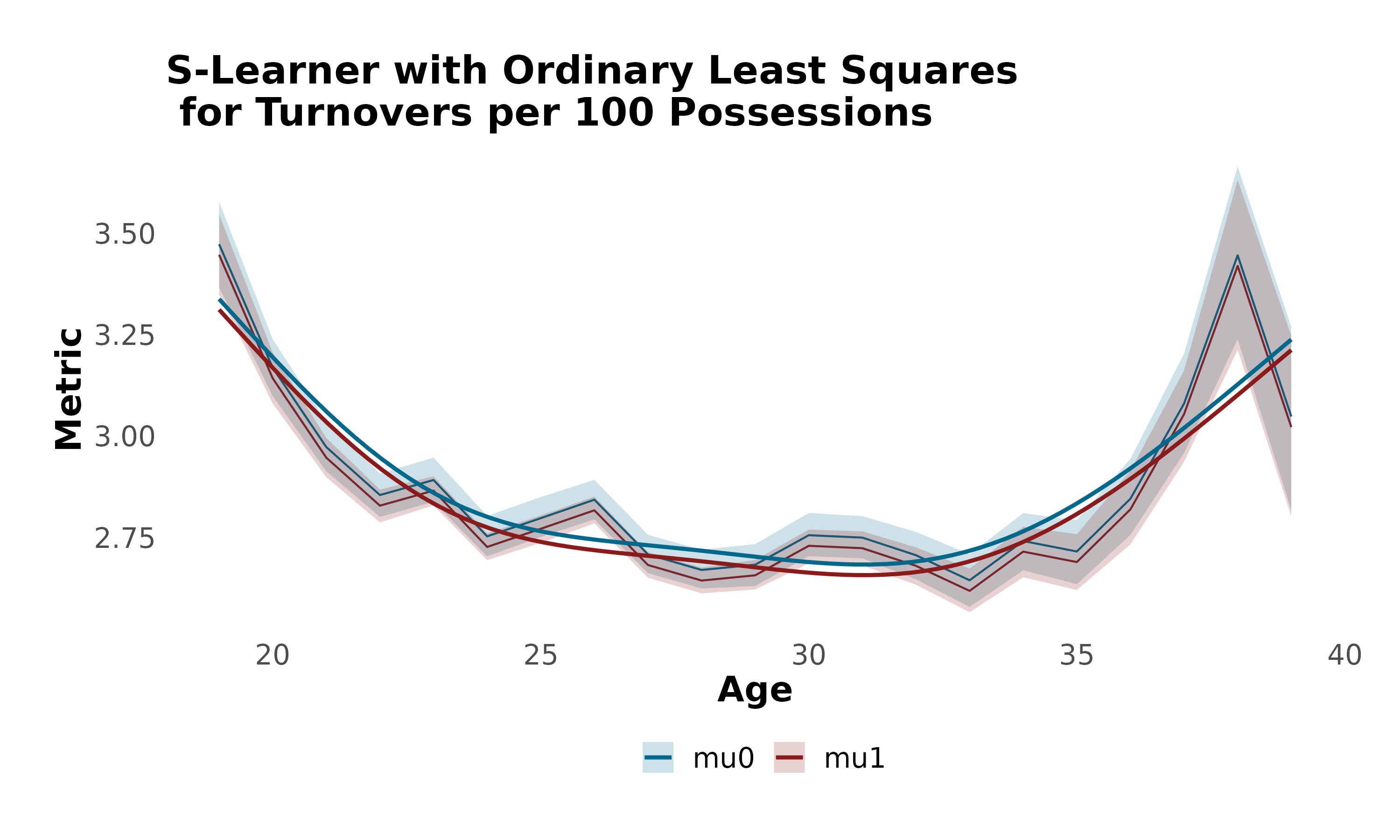}
    \caption{CEF using S-learner and OLS-Spline for box score statistics per 100 possessions}
    \label{fig:sols_box}
\end{figure}

\begin{figure}
    \centering
    \includegraphics[scale=0.2]{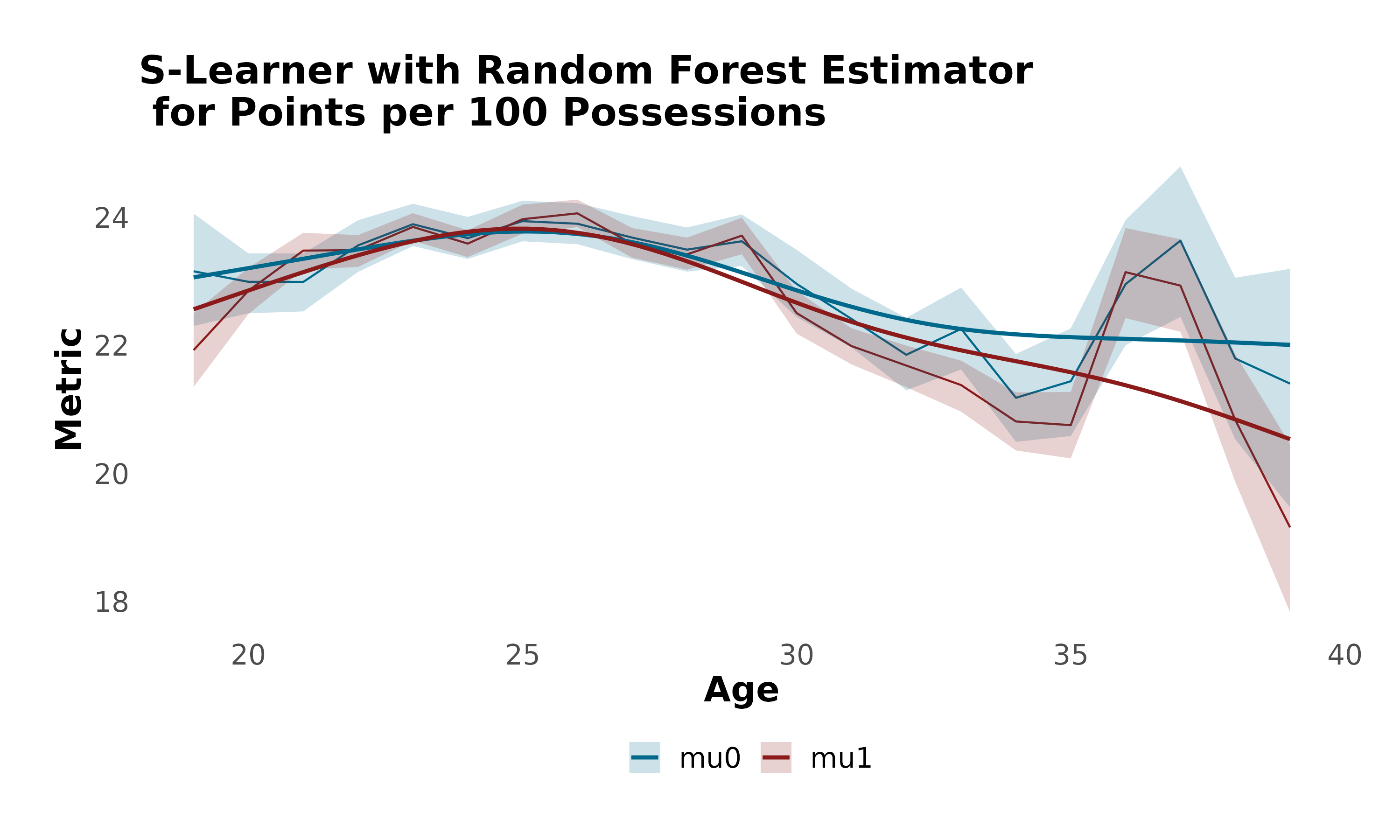} \includegraphics[scale=0.2]{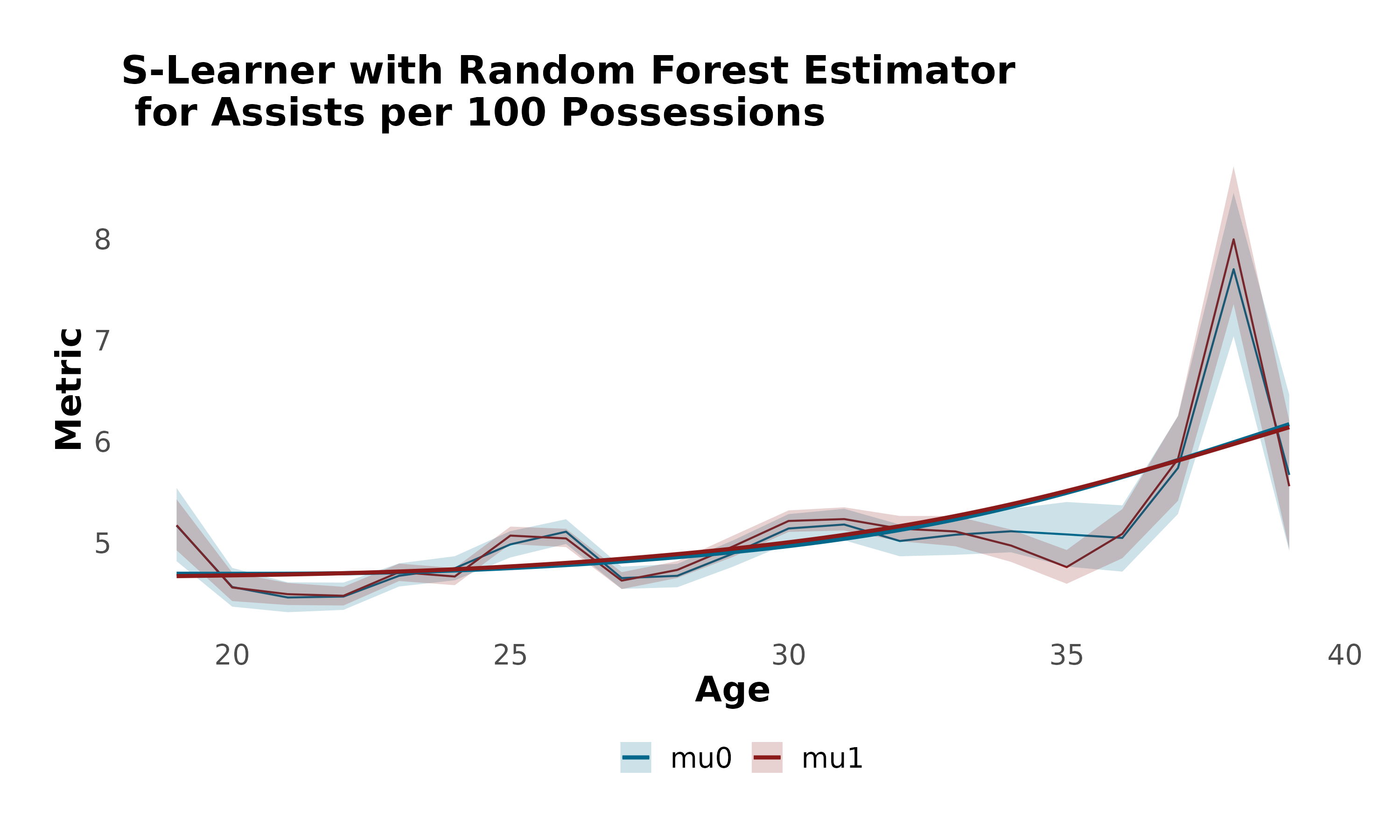}
    \includegraphics[scale=0.2]{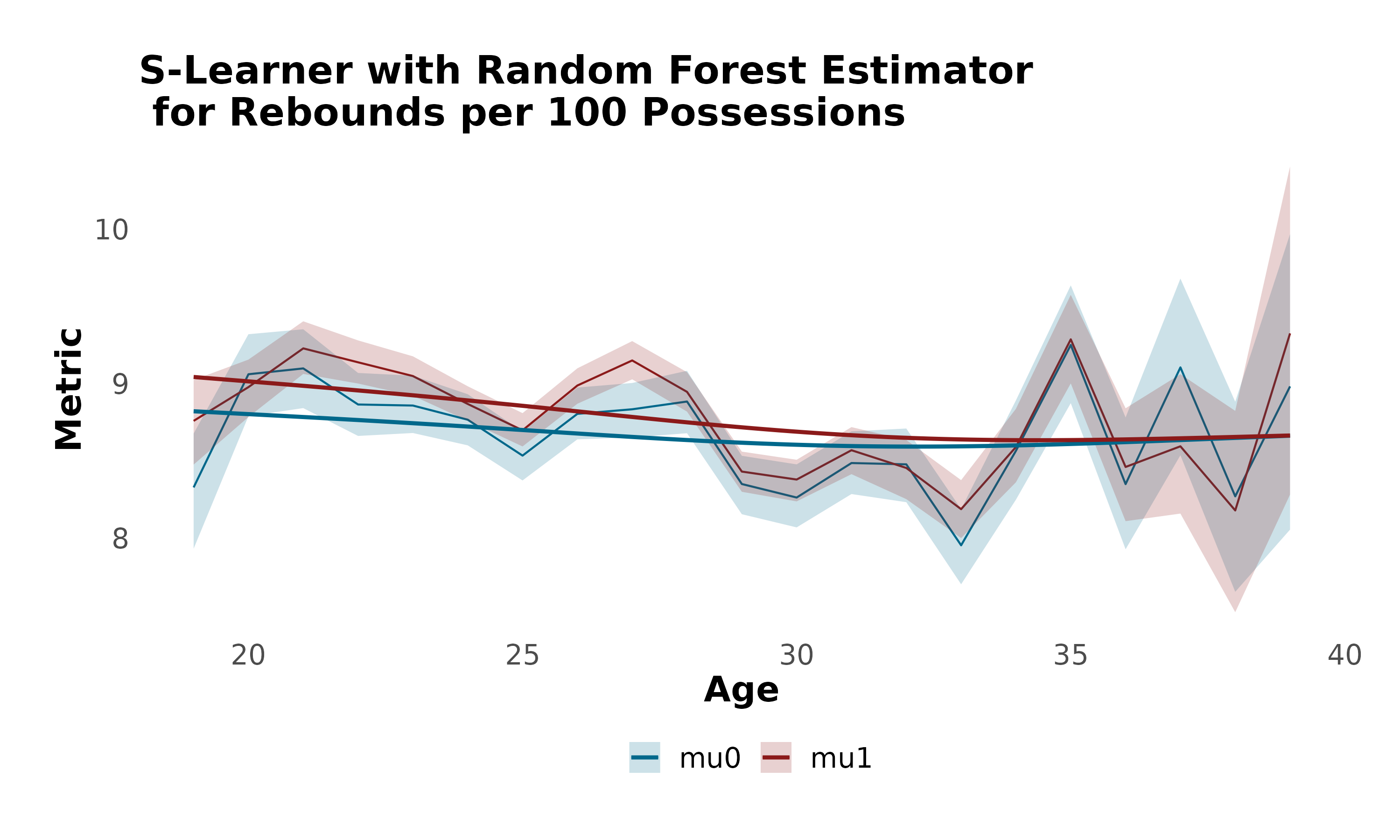}
    \includegraphics[scale=0.2]{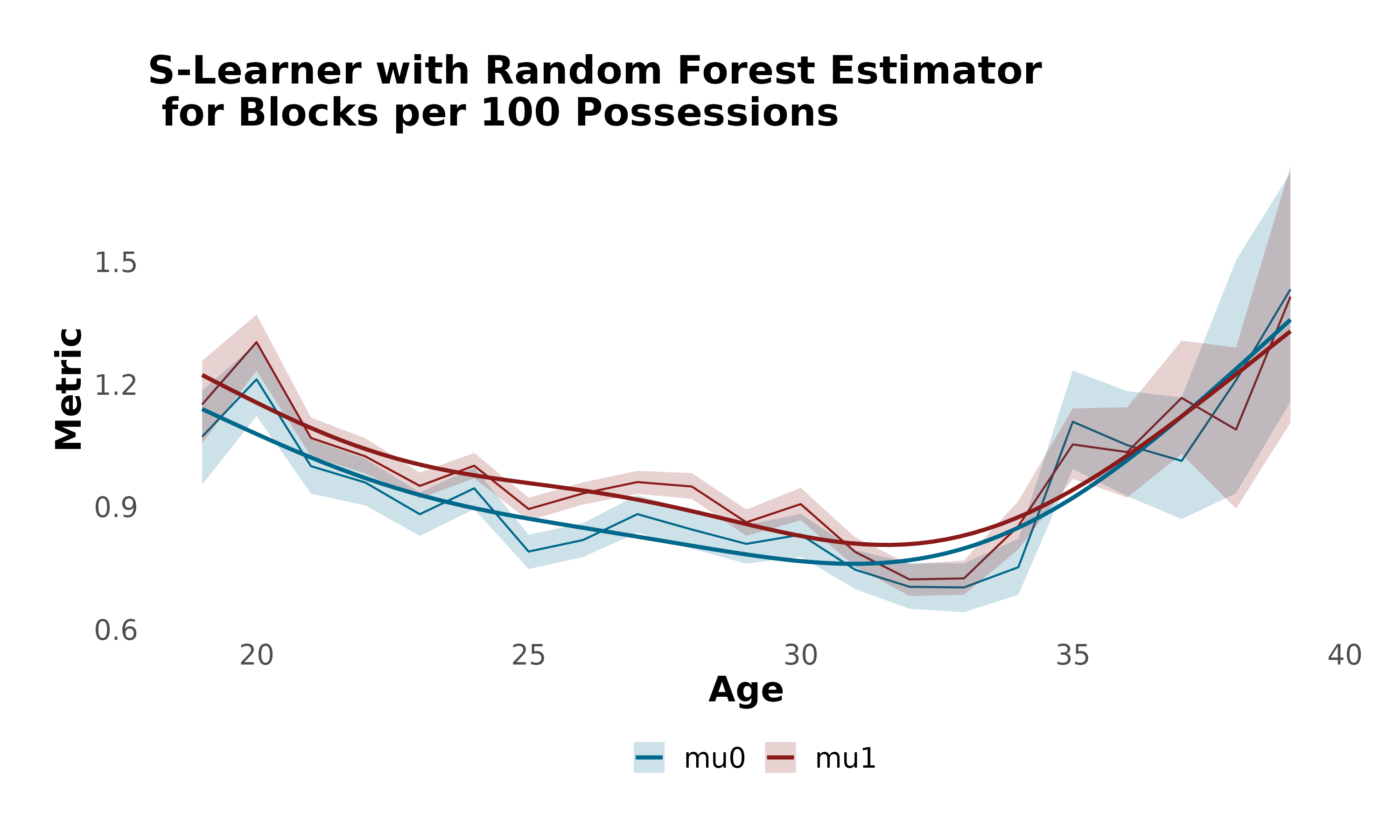}
    \includegraphics[scale=0.2]{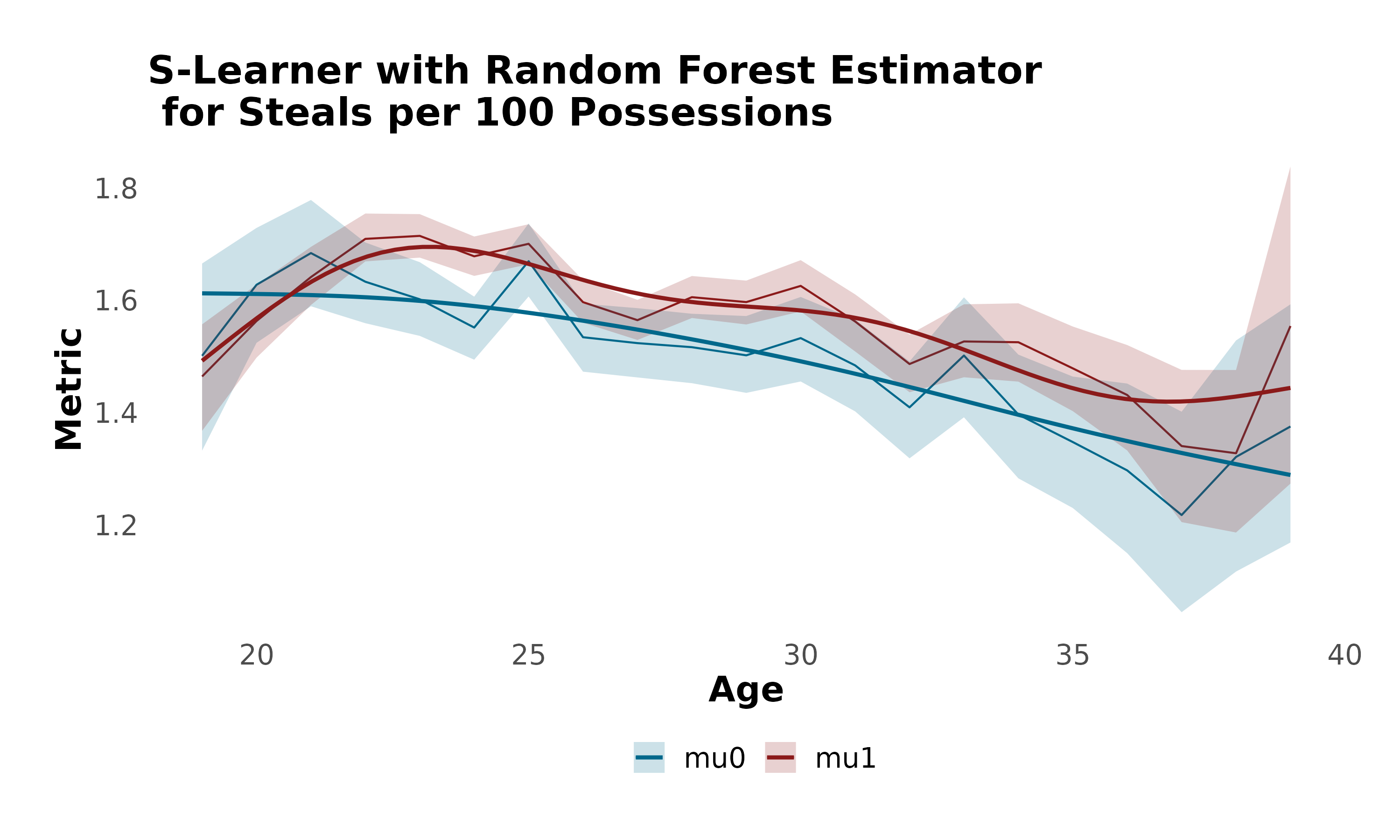}
    \includegraphics[scale=0.2]{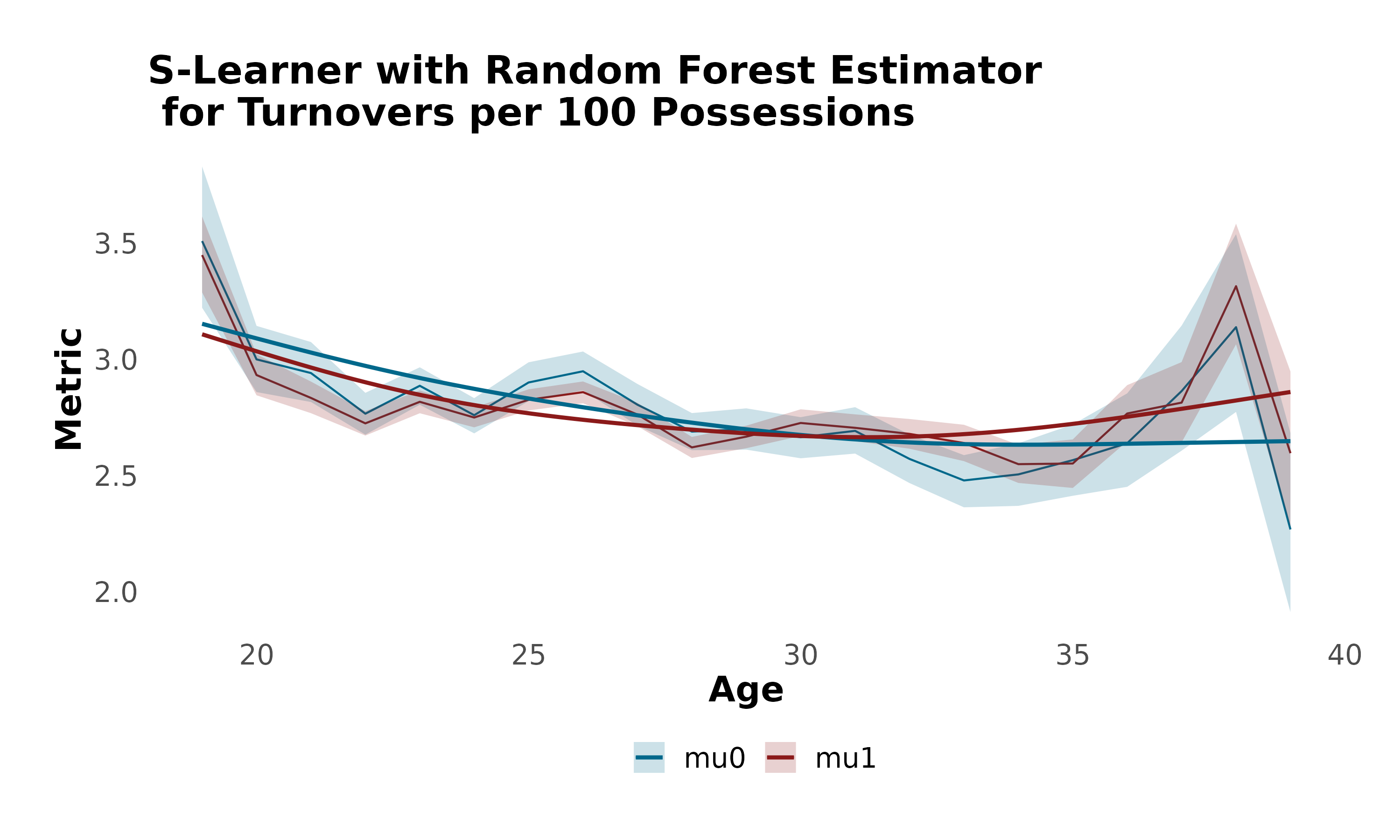}
    \caption{CEF using S-learner and RF for box score statistics per 100 possessions}
    \label{fig:srf_box}
\end{figure}

\begin{figure}
    \centering
    \includegraphics[scale=0.2]{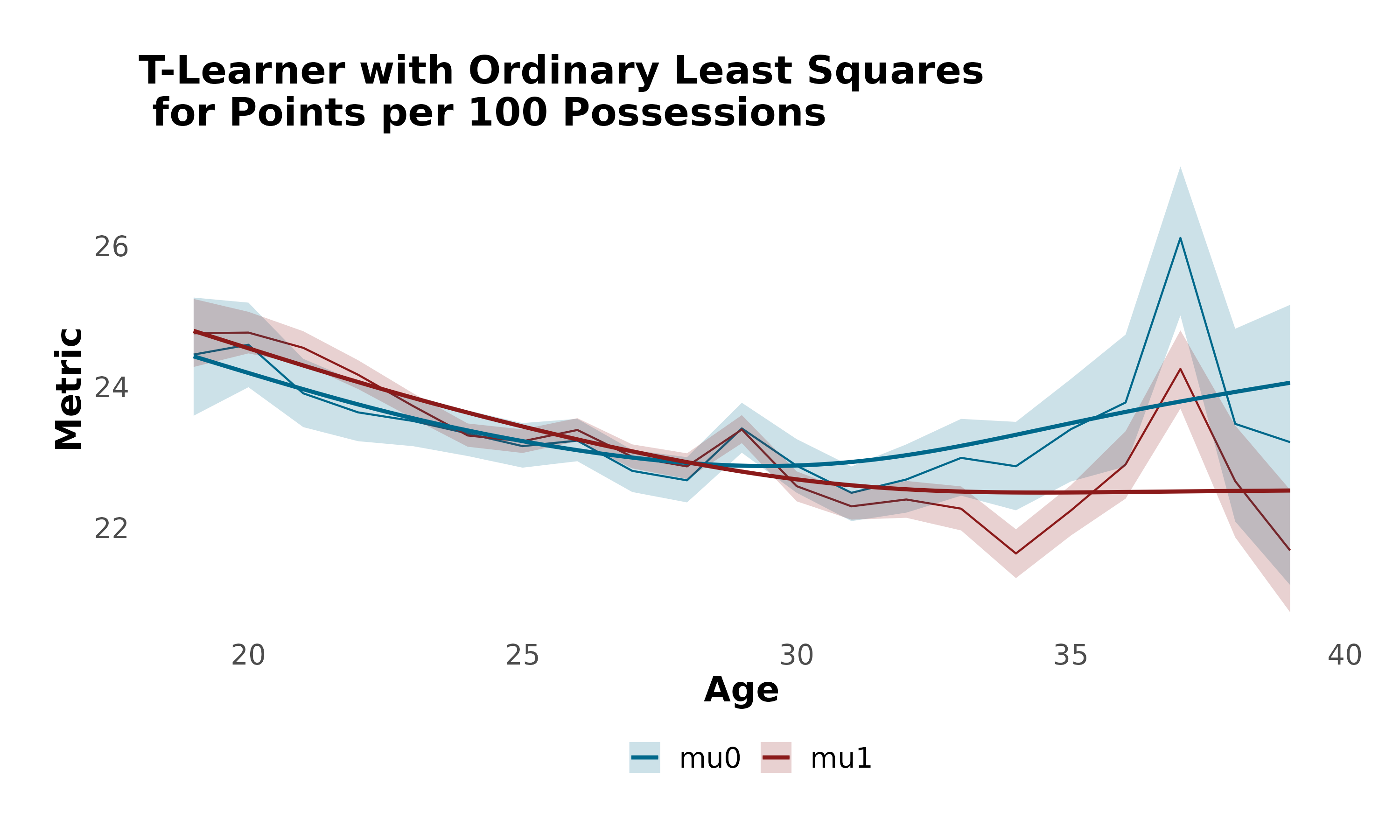} \includegraphics[scale=0.2]{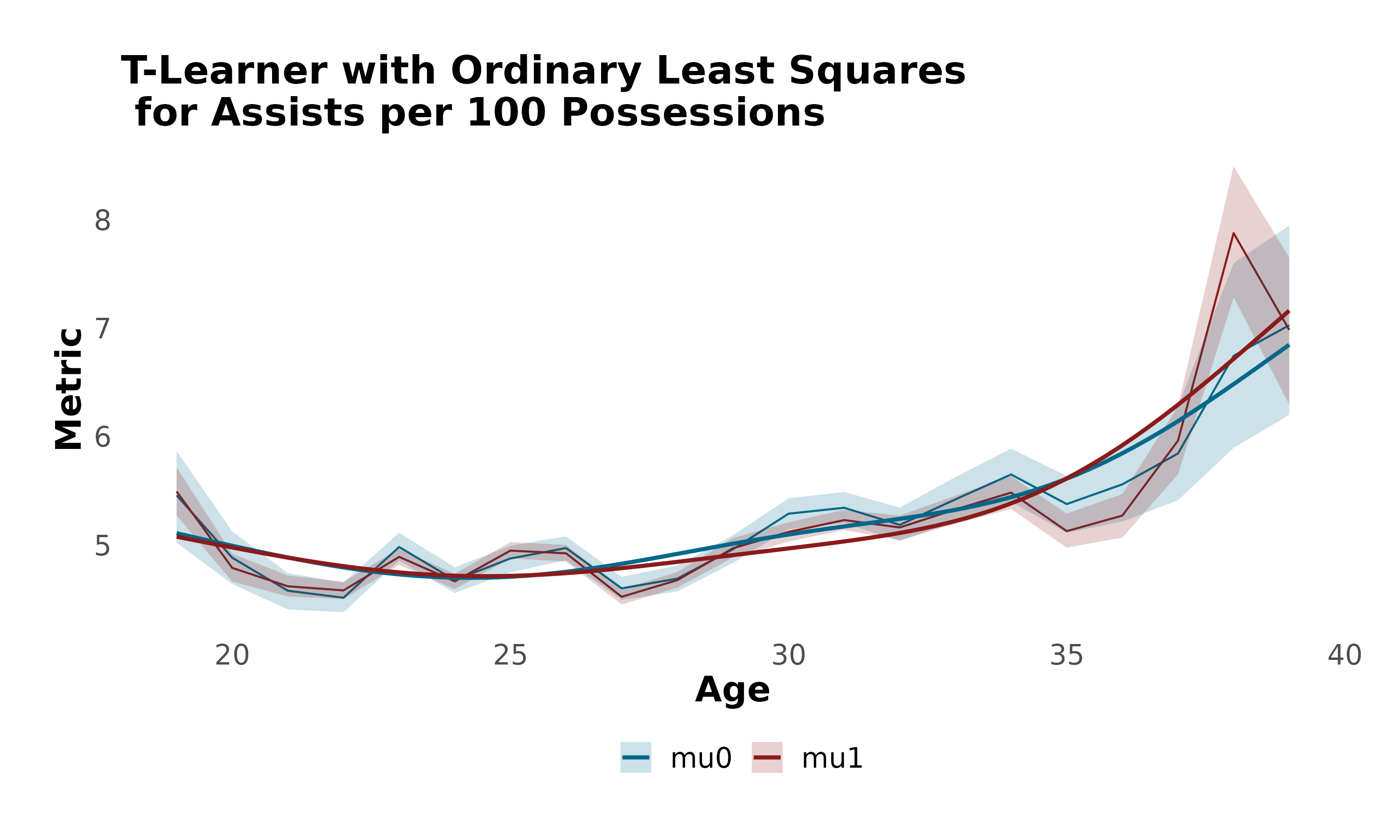}
    \includegraphics[scale=0.2]{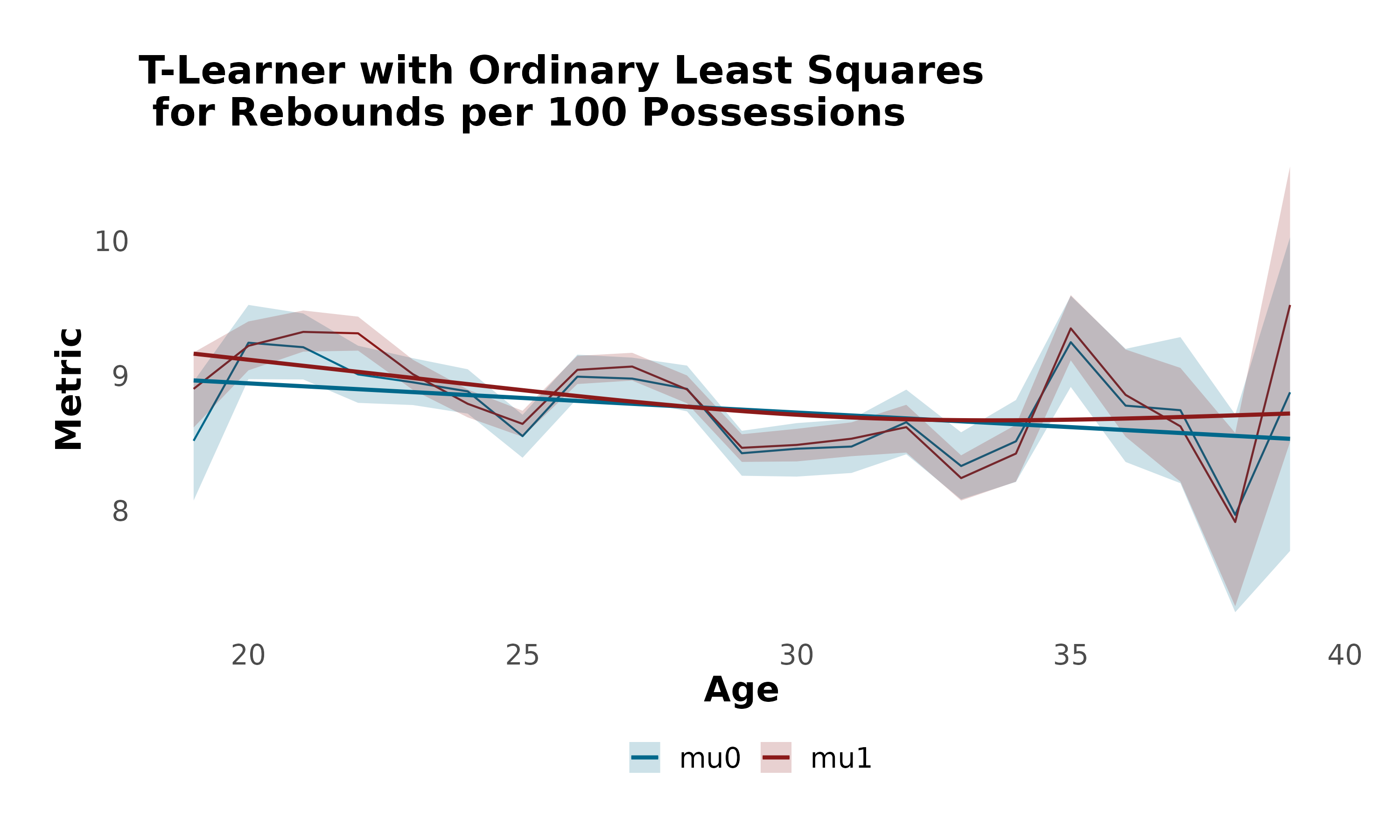}
    \includegraphics[scale=0.2]{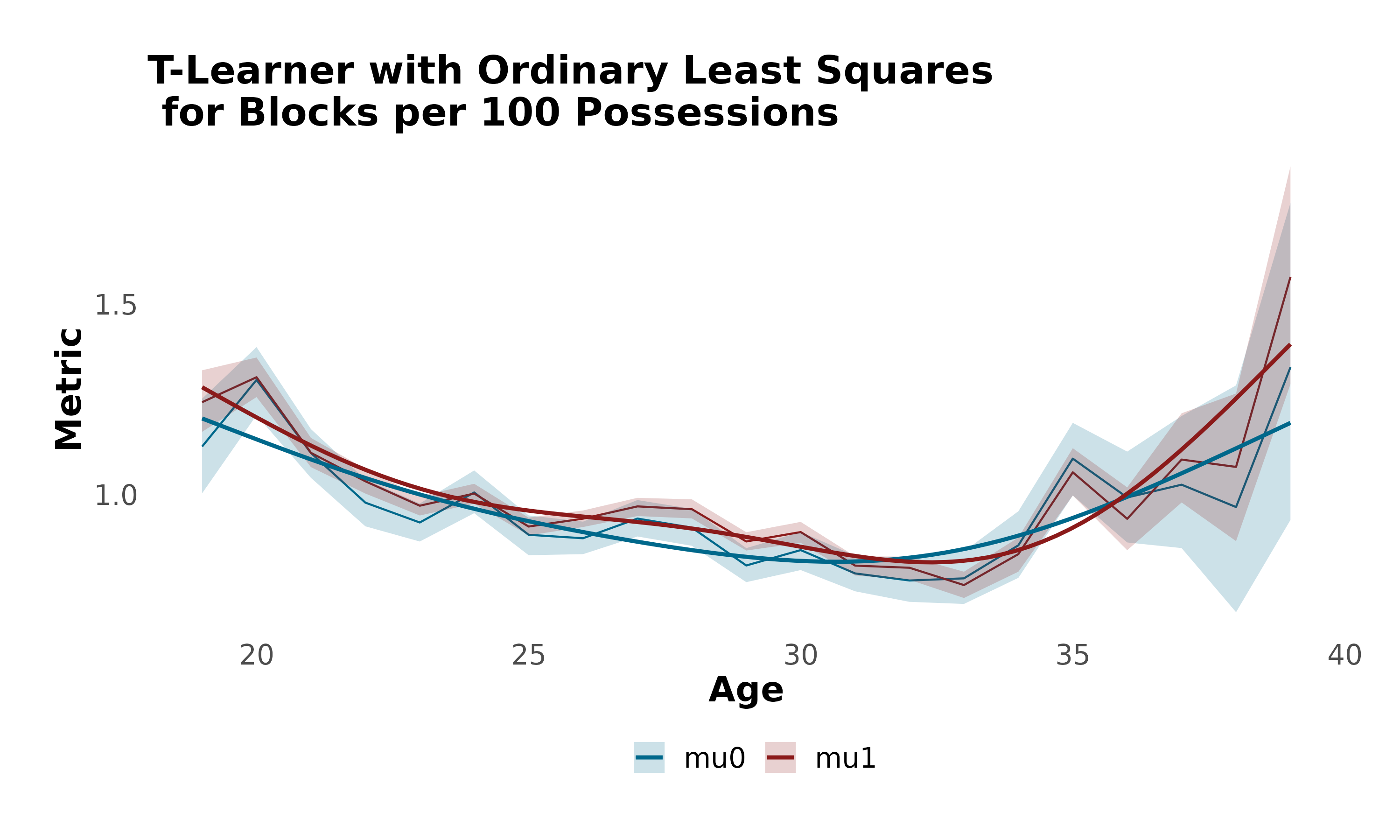}
    \includegraphics[scale=0.2]{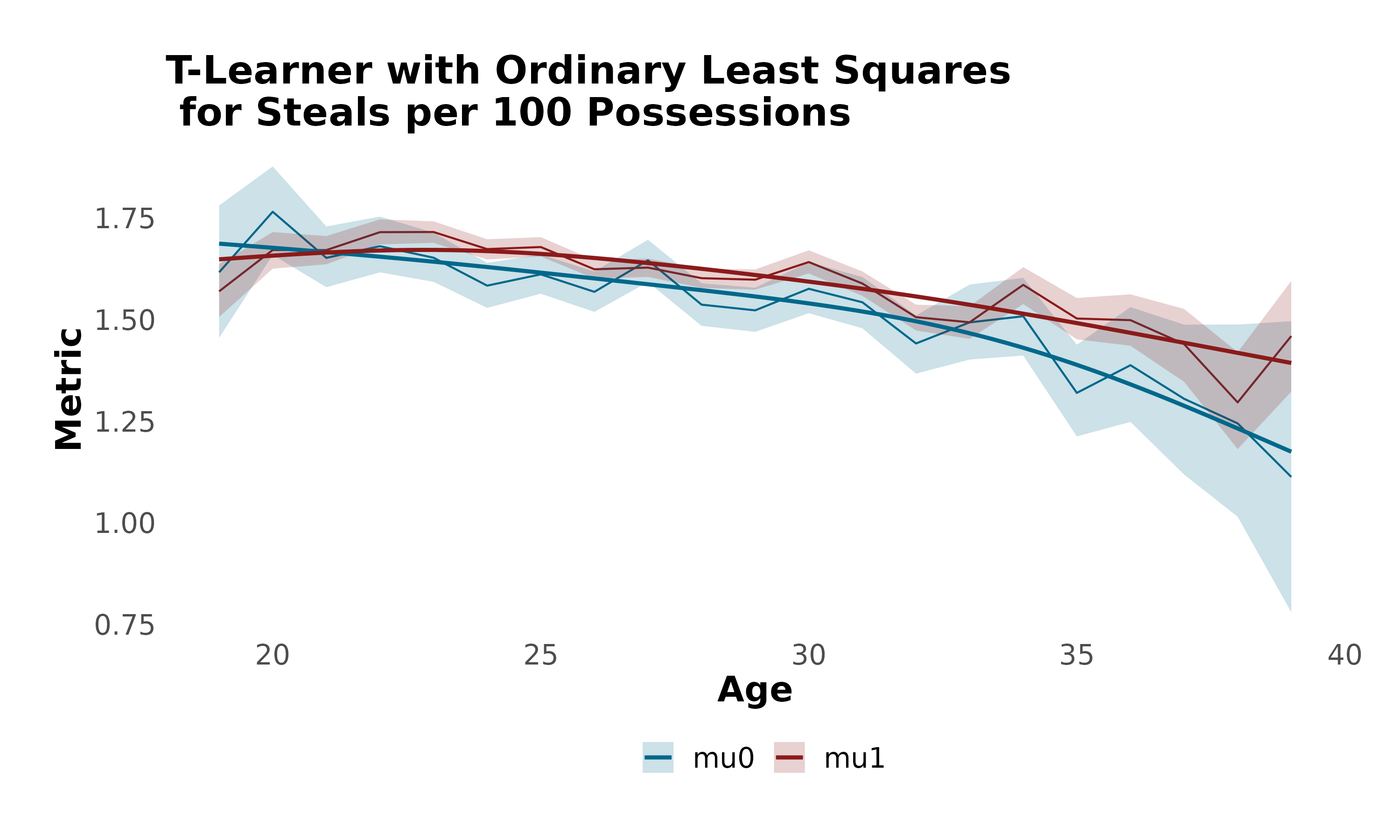}
    \includegraphics[scale=0.2]{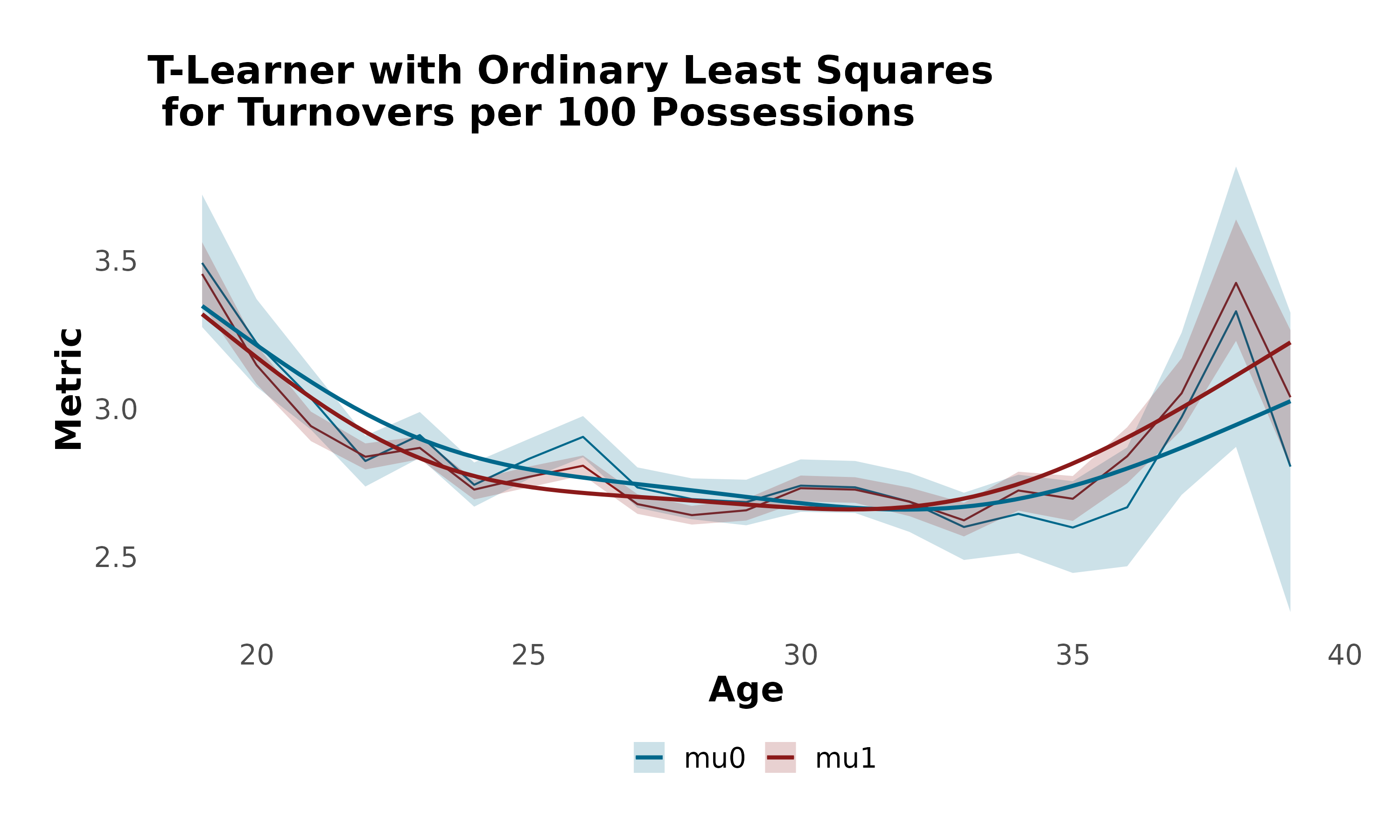}
    \caption{CEF using T-learner and OLS for box score statistics per 100 possessions}
    \label{fig:tols_box}
\end{figure}

\begin{figure}
    \centering
    \includegraphics[scale=0.2]{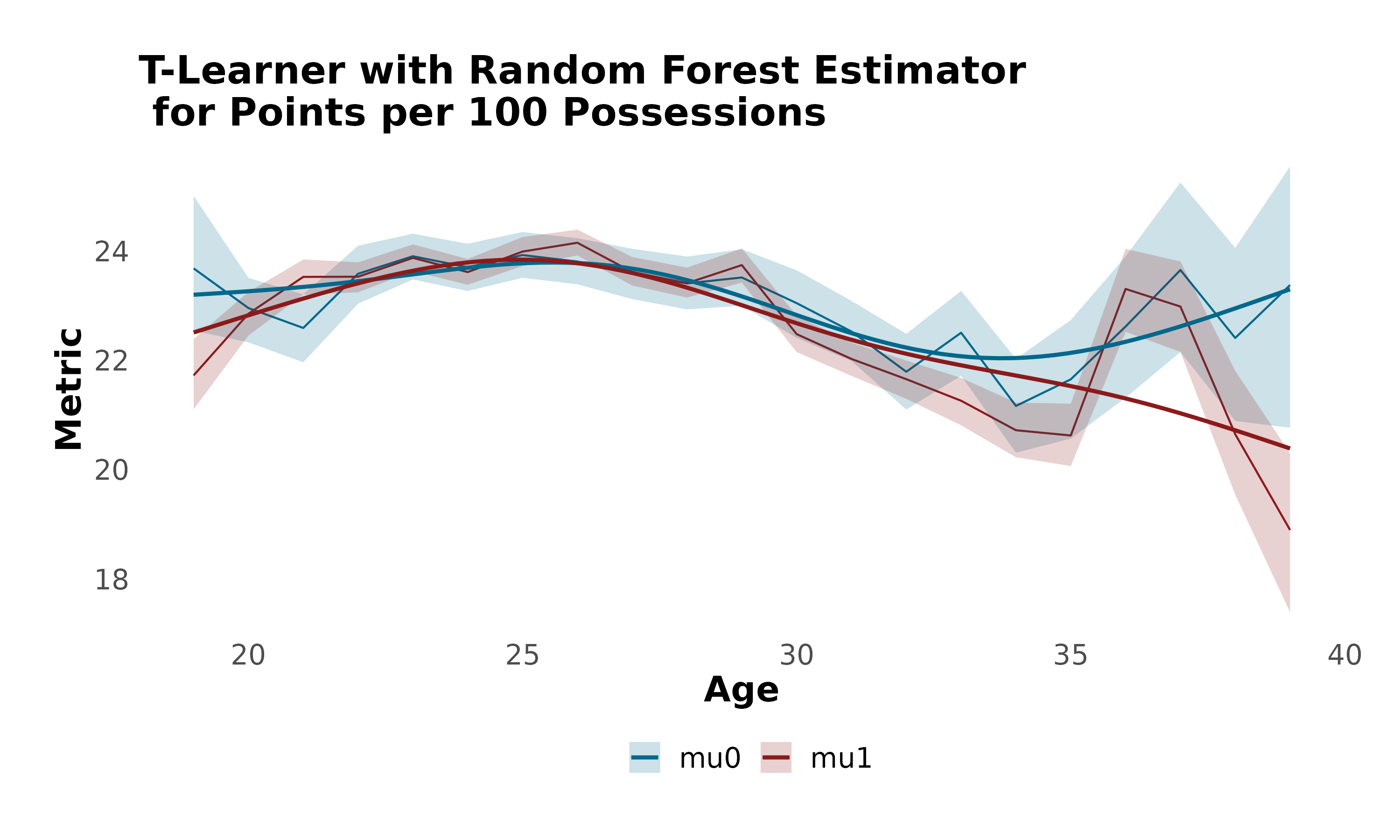} \includegraphics[scale=0.2]{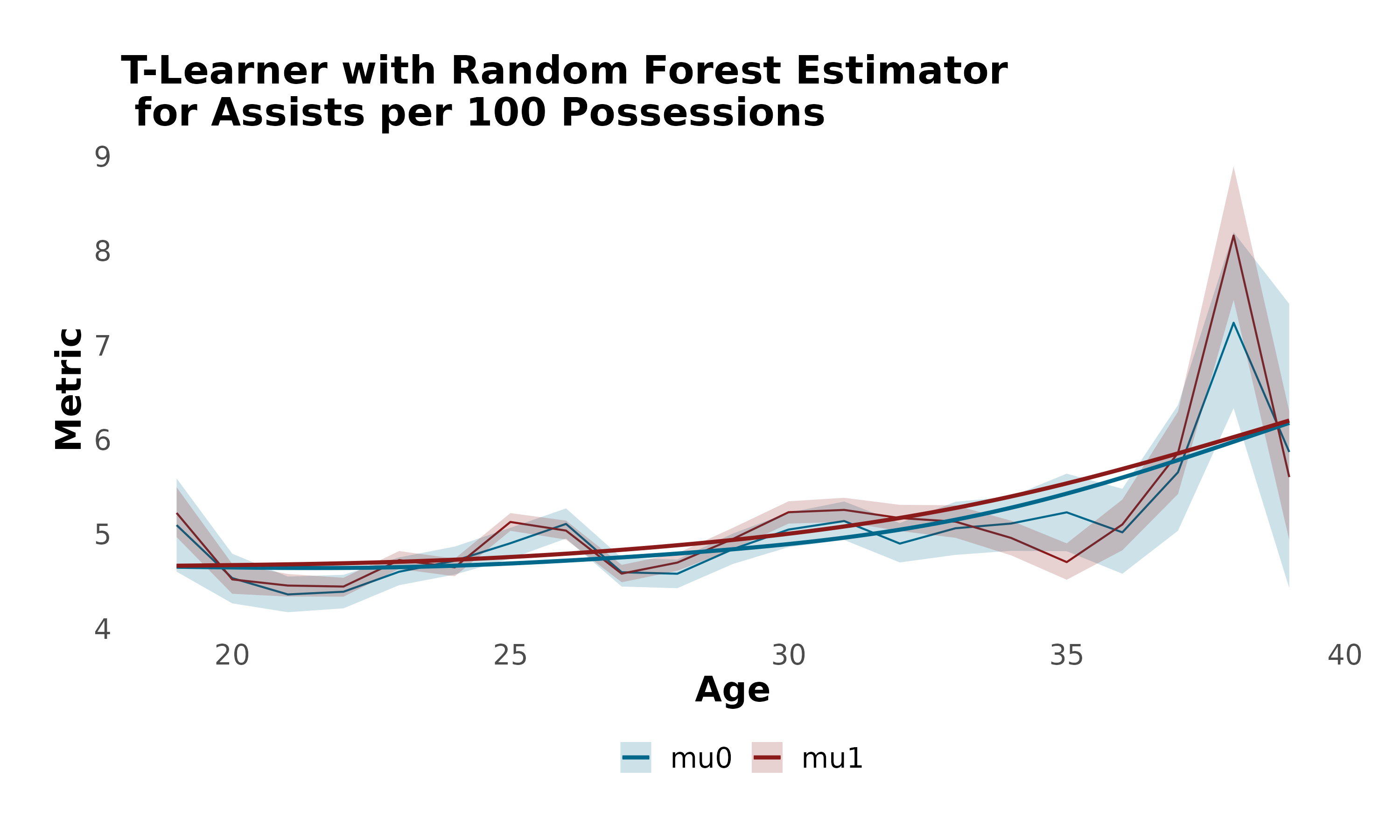}
    \includegraphics[scale=0.2]{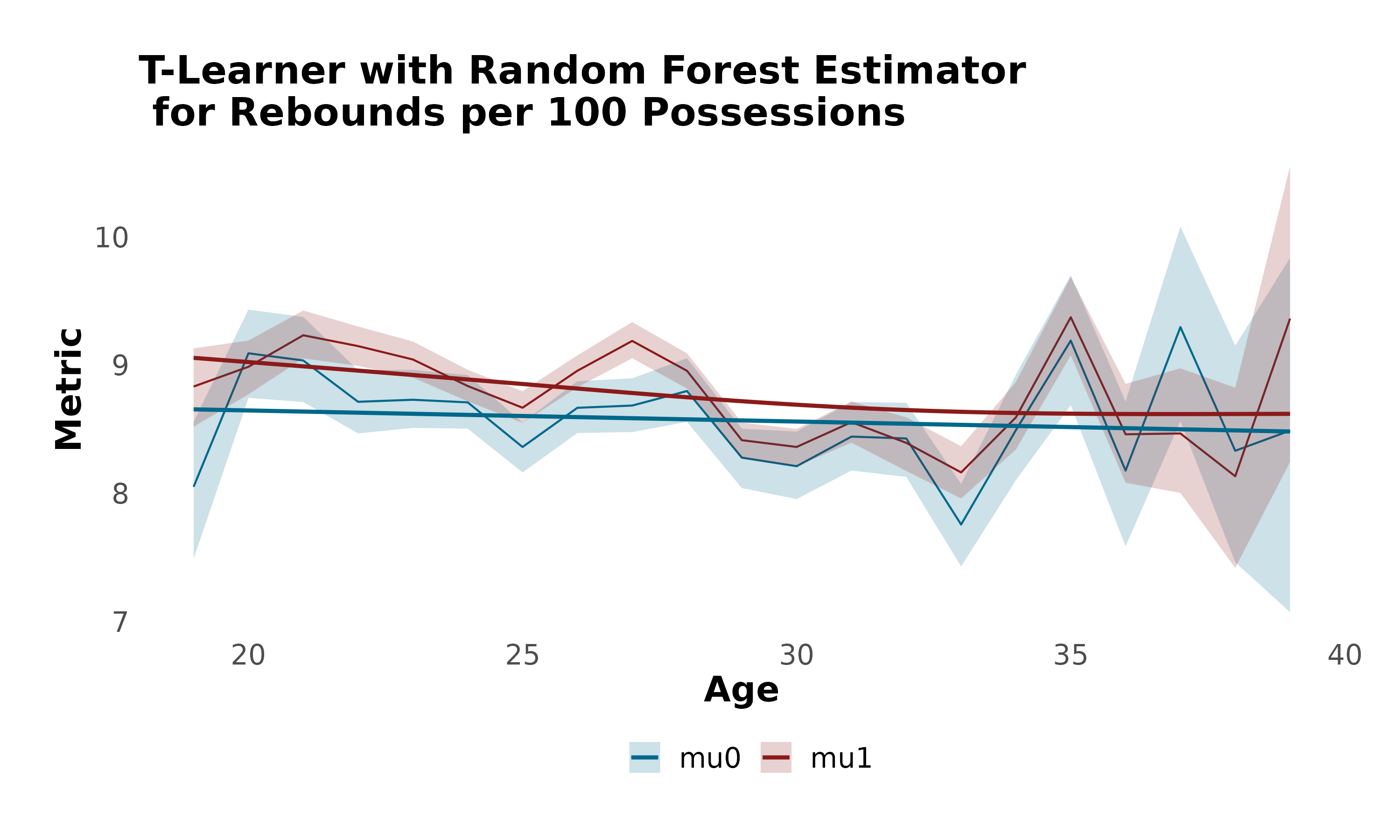}
    \includegraphics[scale=0.2]{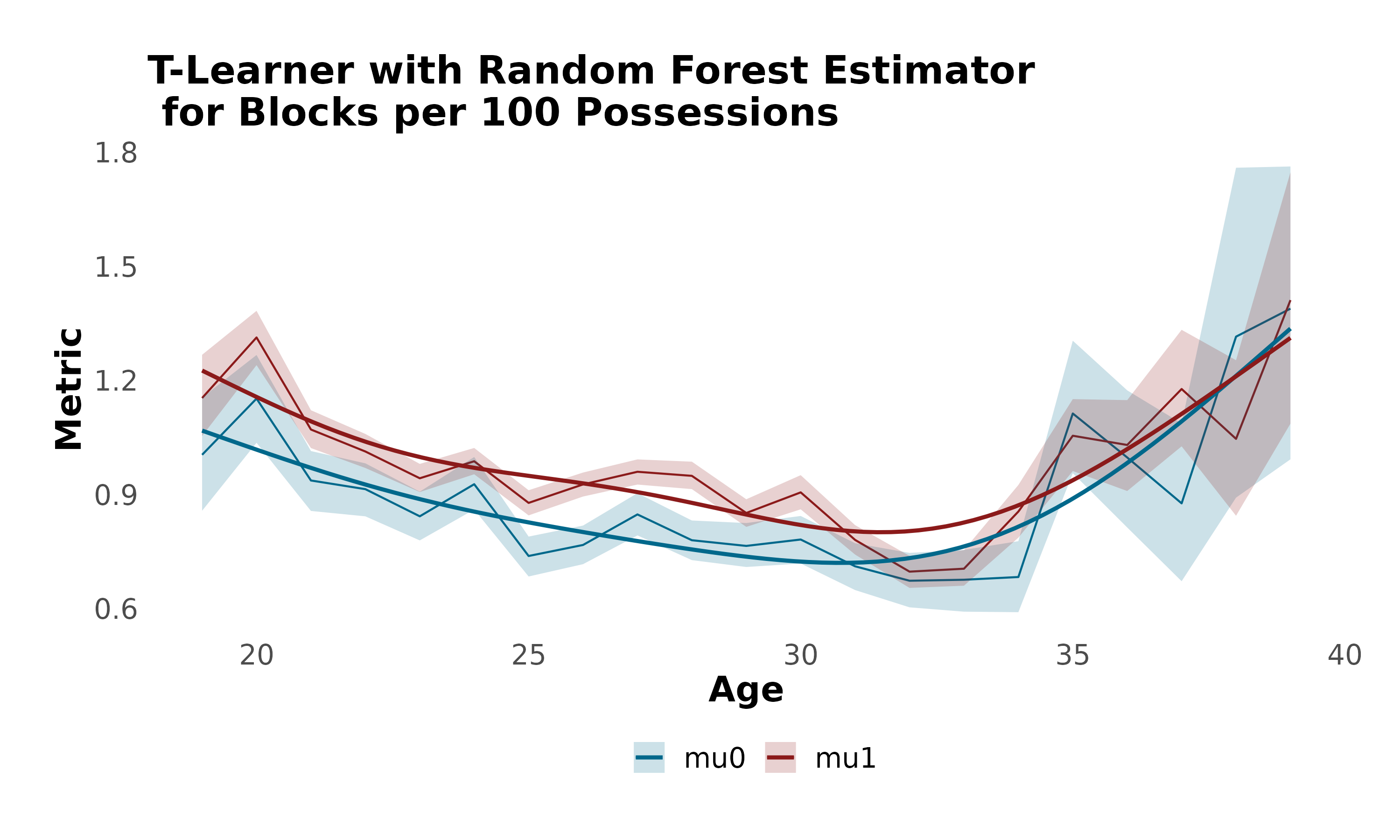}
    \includegraphics[scale=0.2]{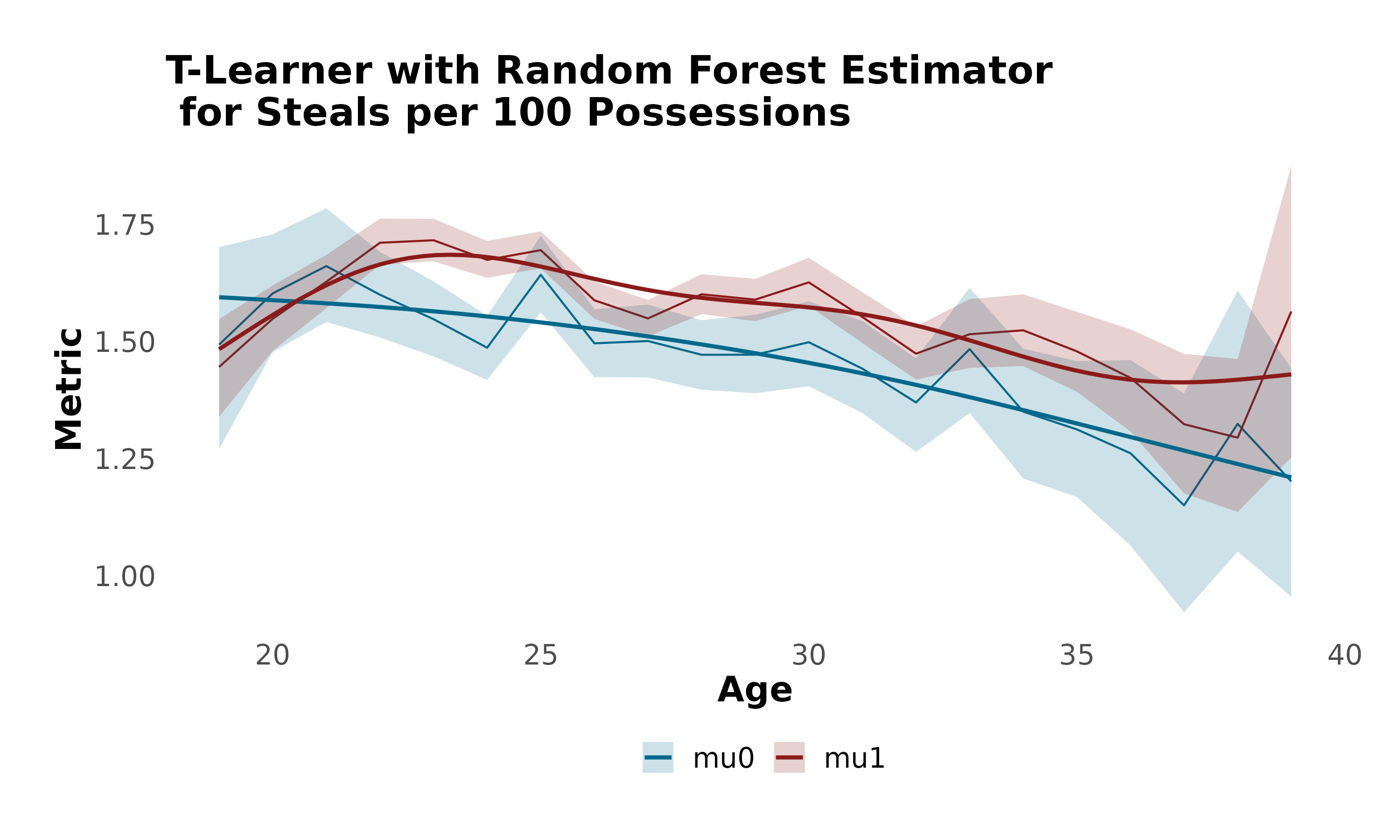}
    \includegraphics[scale=0.2]{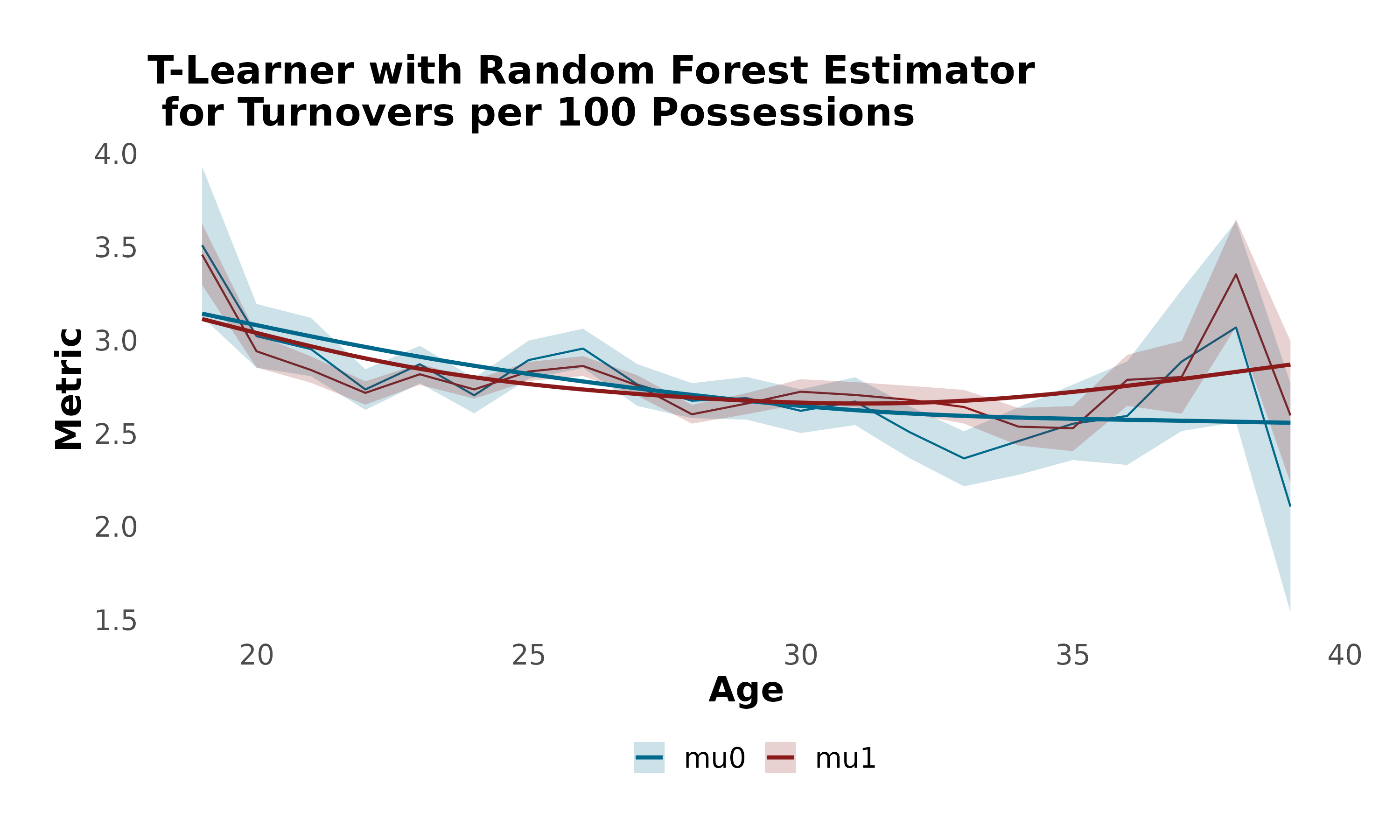}
    \caption{CEF using T-learner and RF for box score statistics per 100 possessions}
    \label{fig:trf_box}
\end{figure}

Finally, we analyze the shooting statistics. Field Goal Percentage, Three-Point Field Goal Percentage, and True Shooting Percentage are essential metrics in basketball for assessing shooting efficiency. Field Goal Percentage is calculated by dividing the number of successful field goals (any shot made from within the three-point line) by the total number of field goal attempts. Three-point field Goal Percentage focuses on shots made from beyond the three-point line, calculated similarly by dividing successful three-point field goals by three-point attempts. True Shooting Percentage is a more comprehensive measure, accounting for field goals, three-point field goals, and free throws, offering a complete picture of a player's shooting efficiency by considering the value of each shot type. It's calculated using a formula that adjusts for the fact that three-point field goals are worth more than other field goals and include free throws.

In Figures \ref{fig:sols_shooting}, \ref{fig:srf_shooting}, \ref{fig:tols_shooting}, and \ref{fig:trf_shooting},
we analyze the impact of rest on various shooting metrics, represented by four distinct curves. Generally, rest influences shooting percentages; however, for three-point attempts, this effect is less pronounced compared to the other two metrics. Specifically, for regular Field Goal Percentage and True Shooting Percentage, the influence of rest appears to be significant predominantly among younger players. However, for older players, the data presents a degree of uncertainty and erratic patterns. This inconsistency can be attributed to the limited sample size for this age group, leading to imprecise curve estimations that are vulnerable to skewing by outliers. The observed trend suggests that scoring two-point field goals, which often involve more physical exertion due to closer proximity to the basket and defensive pressure, is more affected by rest than shooting three-pointers, which are typically less physically demanding. This distinction could explain why rest has a more pronounced effect on FG and TS percentage compared to three-point shooting.

\begin{figure}
    \centering
    \includegraphics[scale=0.2]{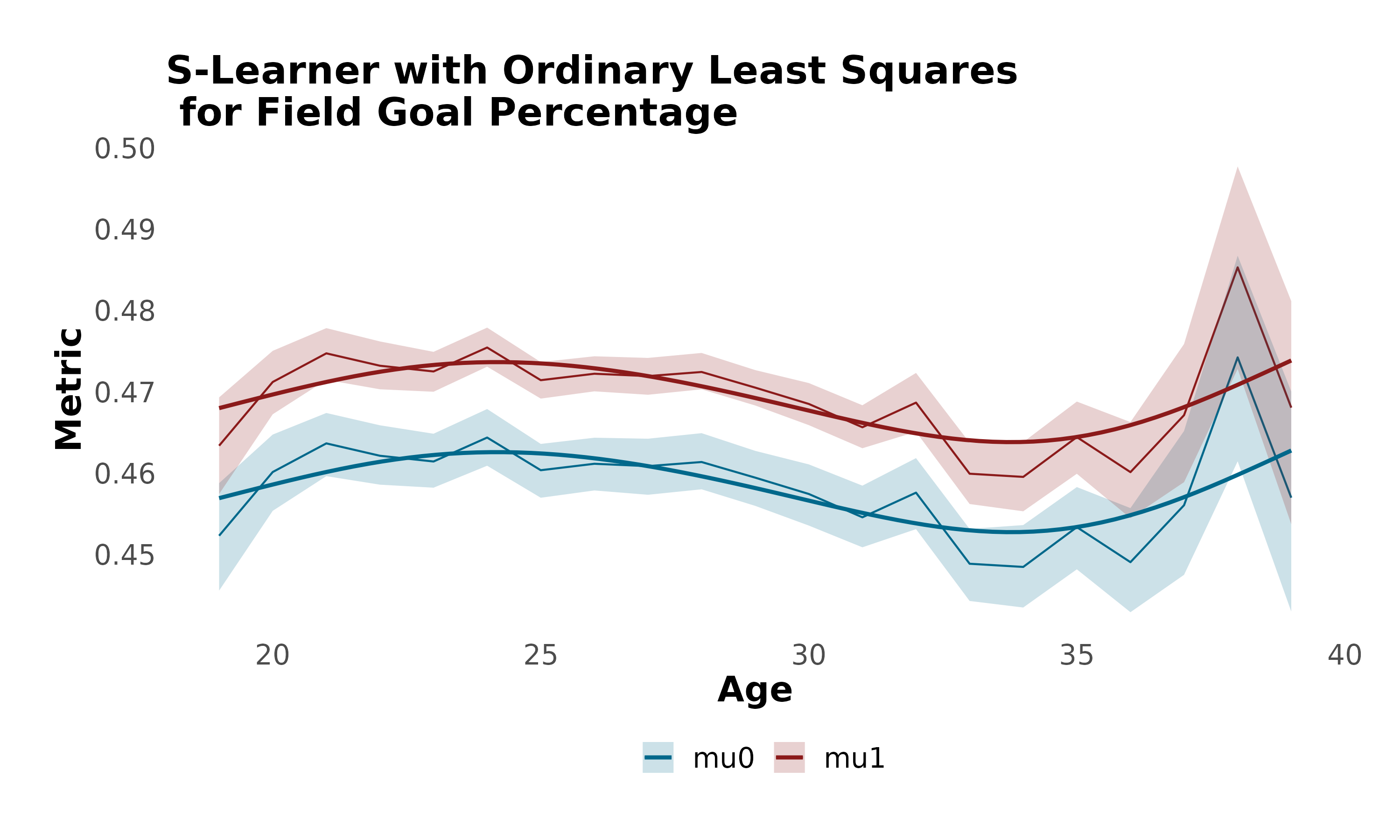} \includegraphics[scale=0.2]{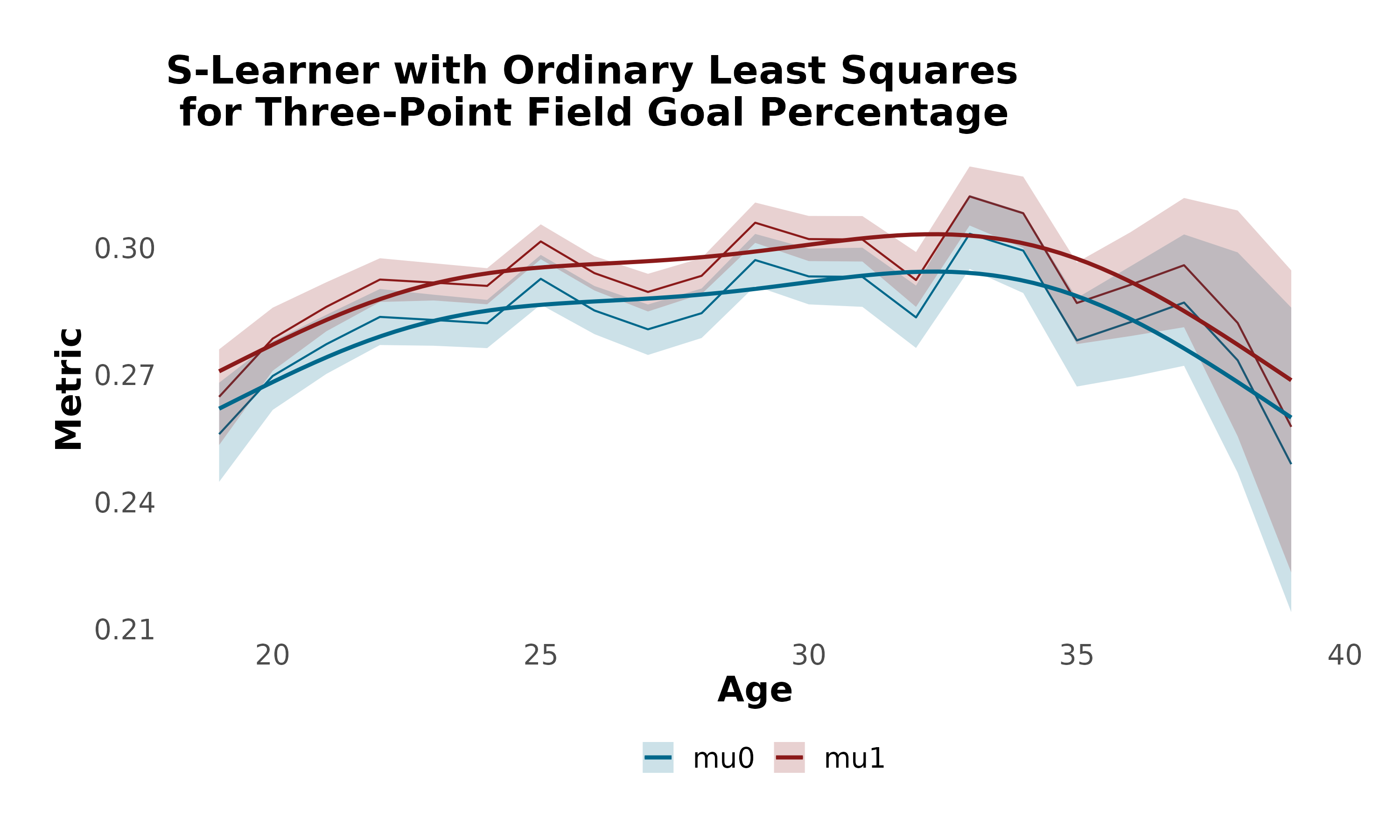}
    \includegraphics[scale=0.2]{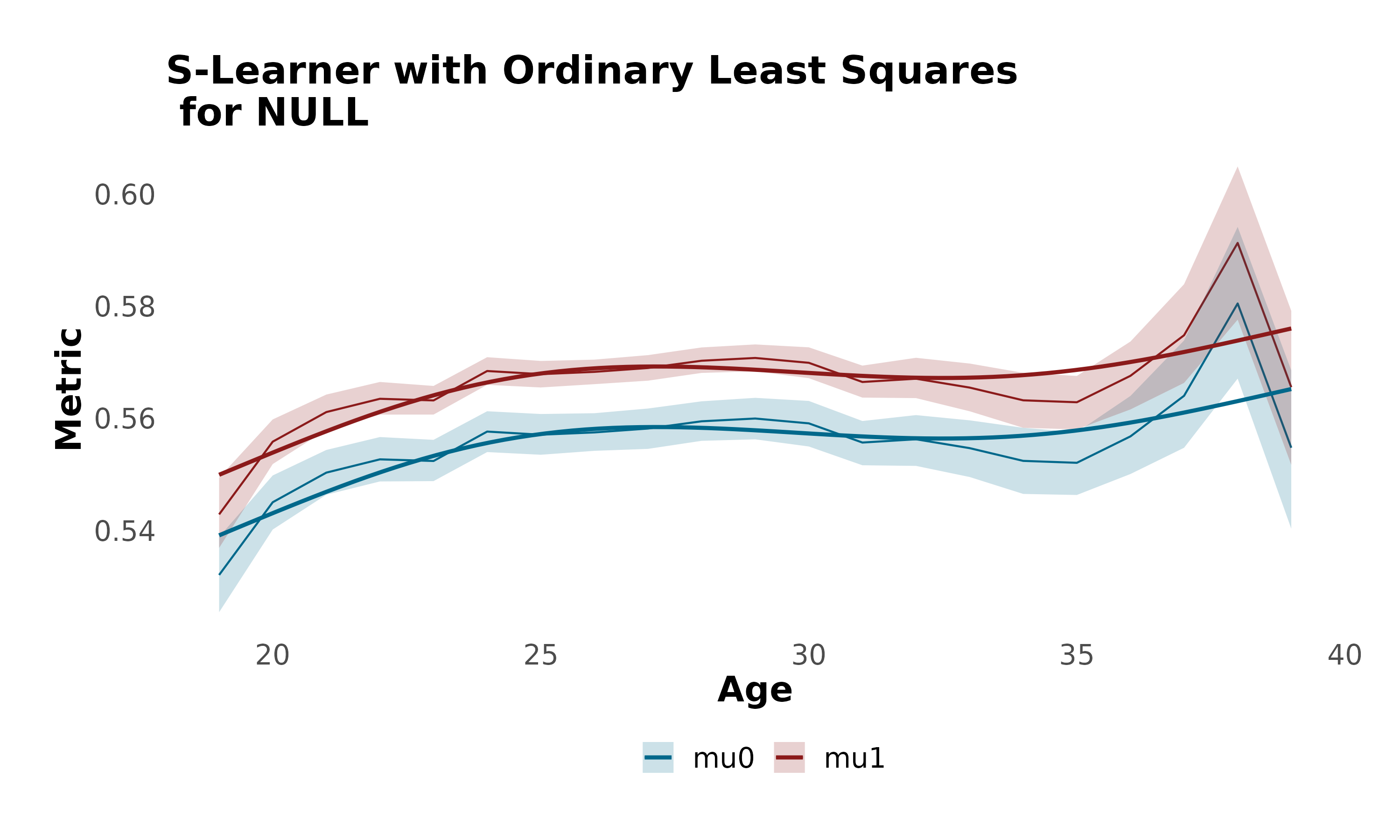}
    \caption{CEF using S-learner and OLS-Spline for the field goal, three-point, and true shooting percentage}
    \label{fig:sols_shooting}
\end{figure}

\begin{figure}
    \centering
    \includegraphics[scale=0.2]{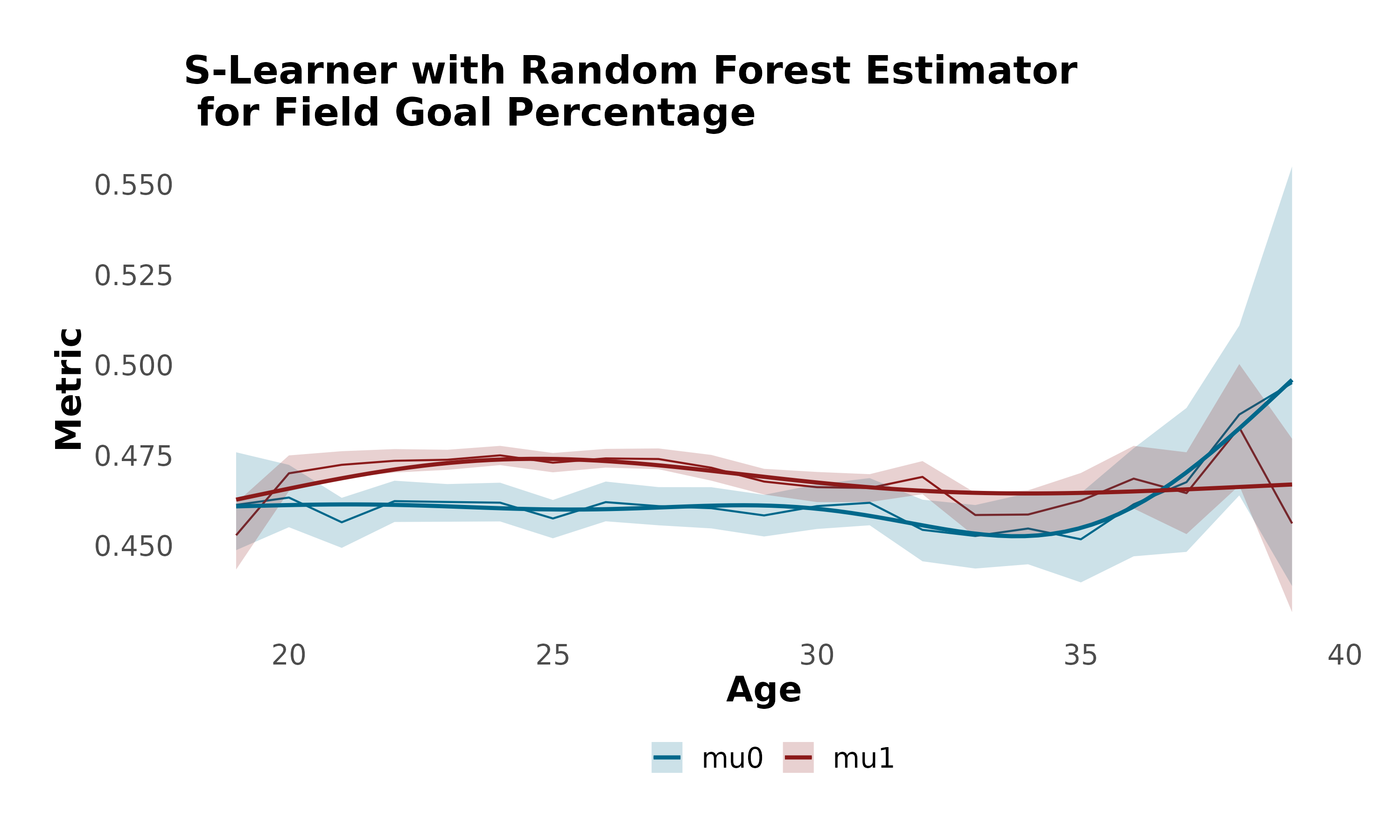} \includegraphics[scale=0.2]{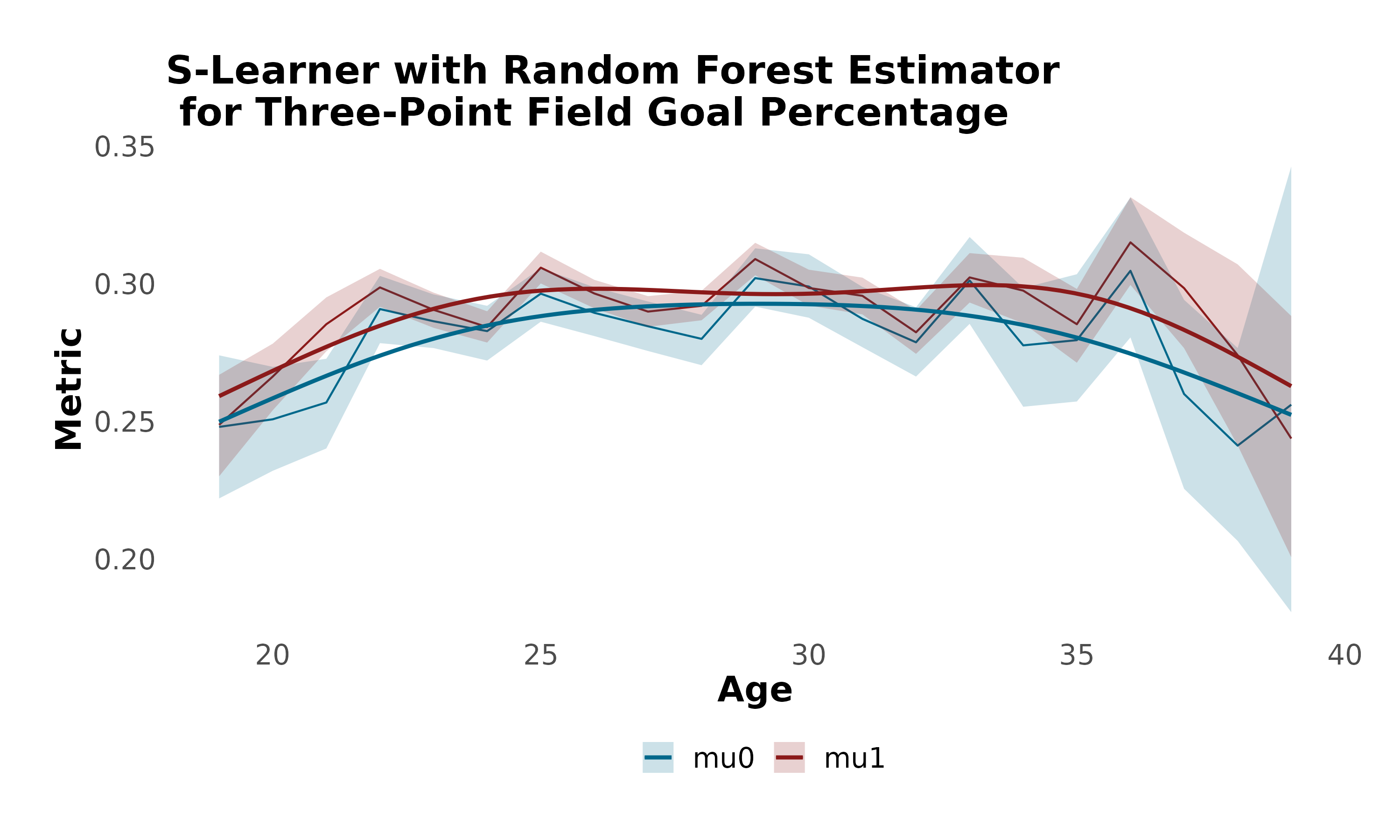}
    \includegraphics[scale=0.2]{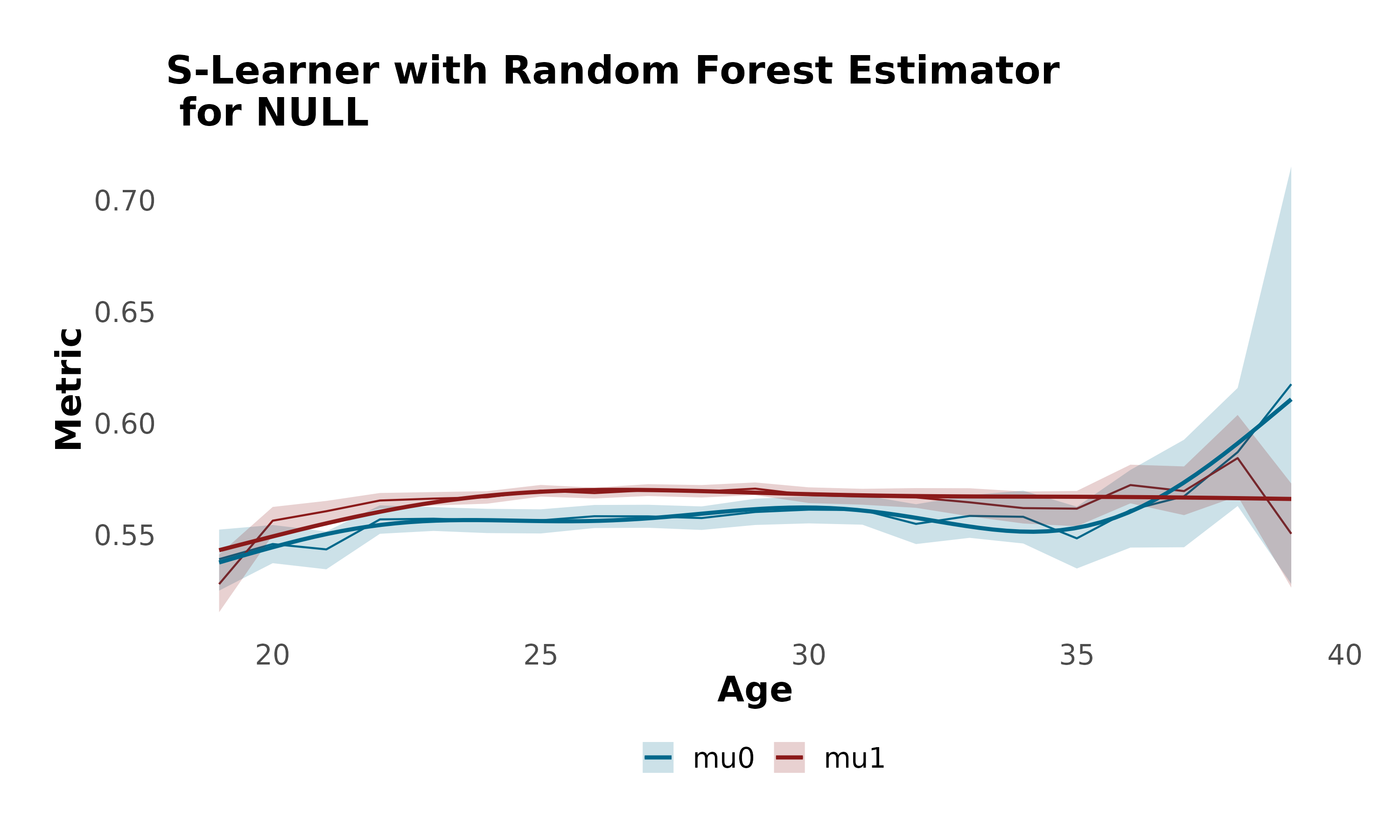}
    \caption{CEF using S-learner and RF for the field goal, three-point, and true shooting percentage}
    \label{fig:srf_shooting}
\end{figure}

\begin{figure}
    \centering
    \includegraphics[scale=0.2]{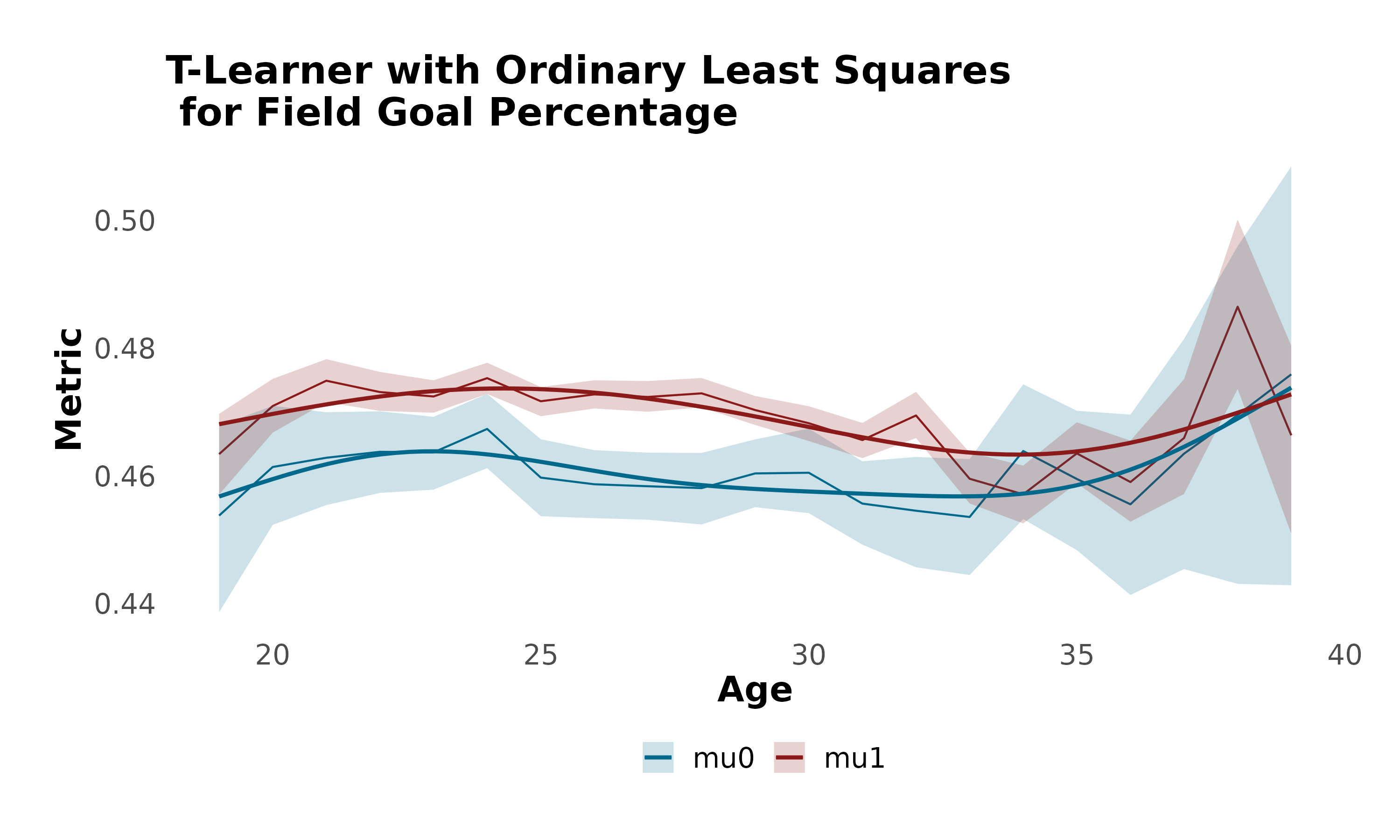} \includegraphics[scale=0.2]{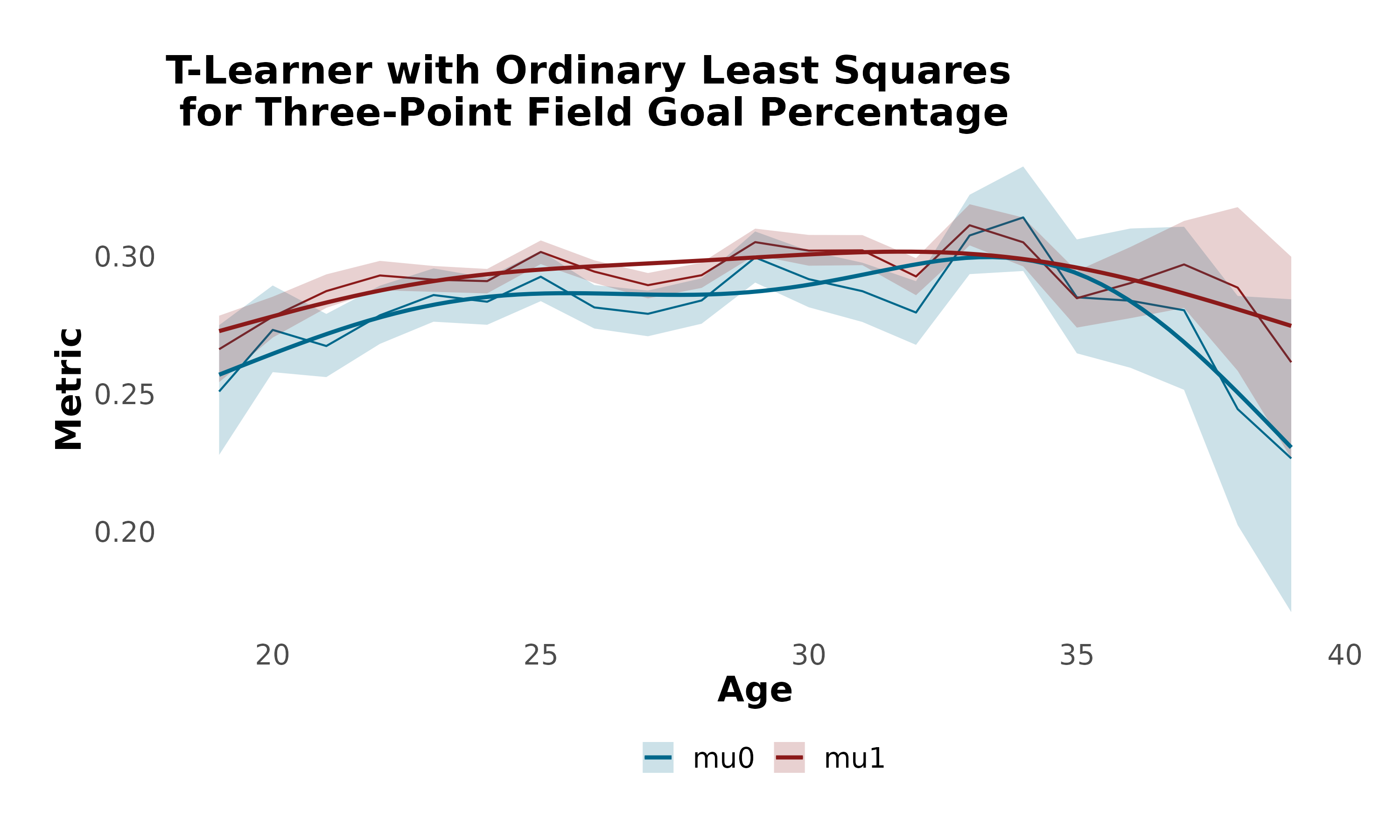}
    \includegraphics[scale=0.2]{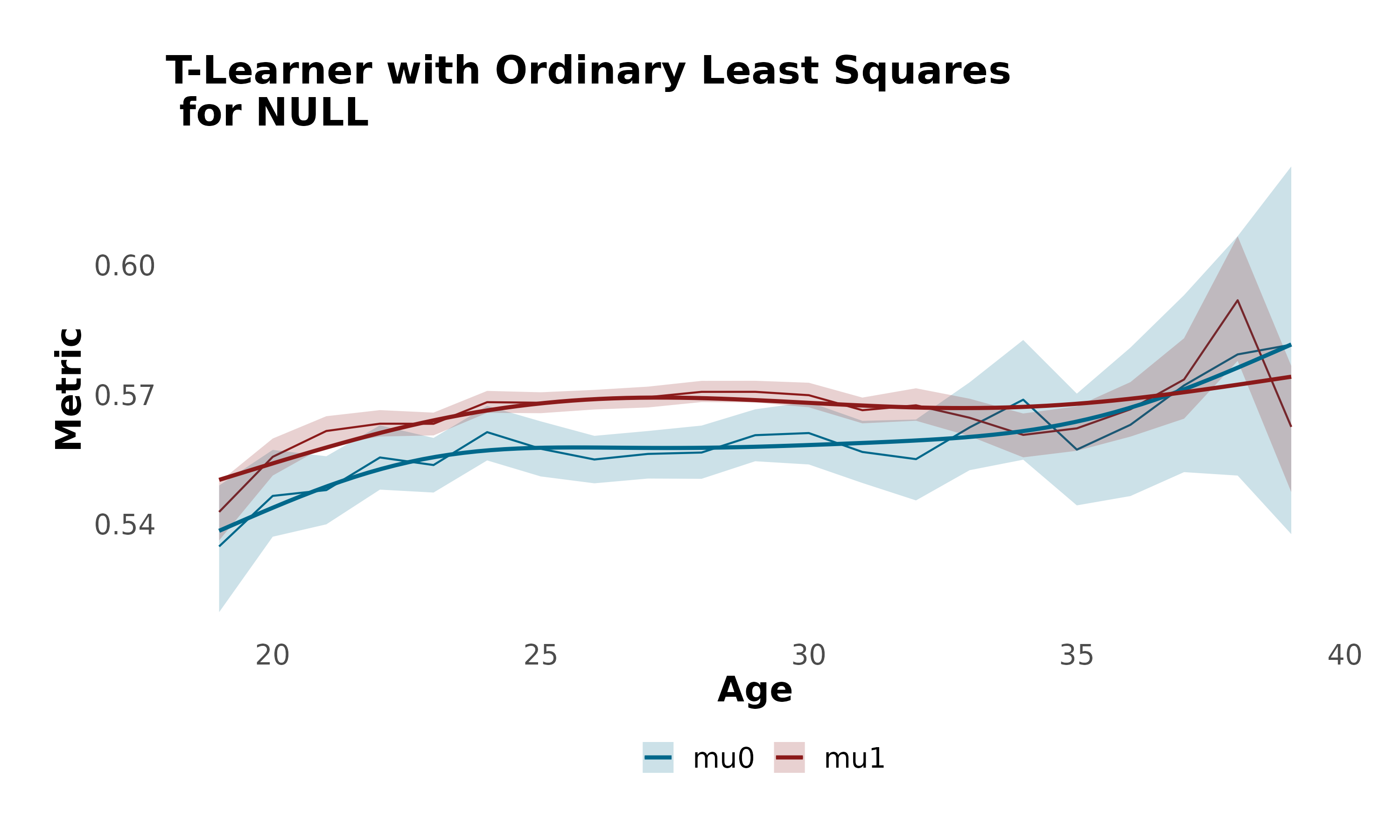}
    \caption{CEF using T-learner and OLS-Spline for the field goal, three-point, and true shooting percentage}
    \label{fig:tols_shooting}
\end{figure}

\begin{figure}
    \centering
    \includegraphics[scale=0.2]{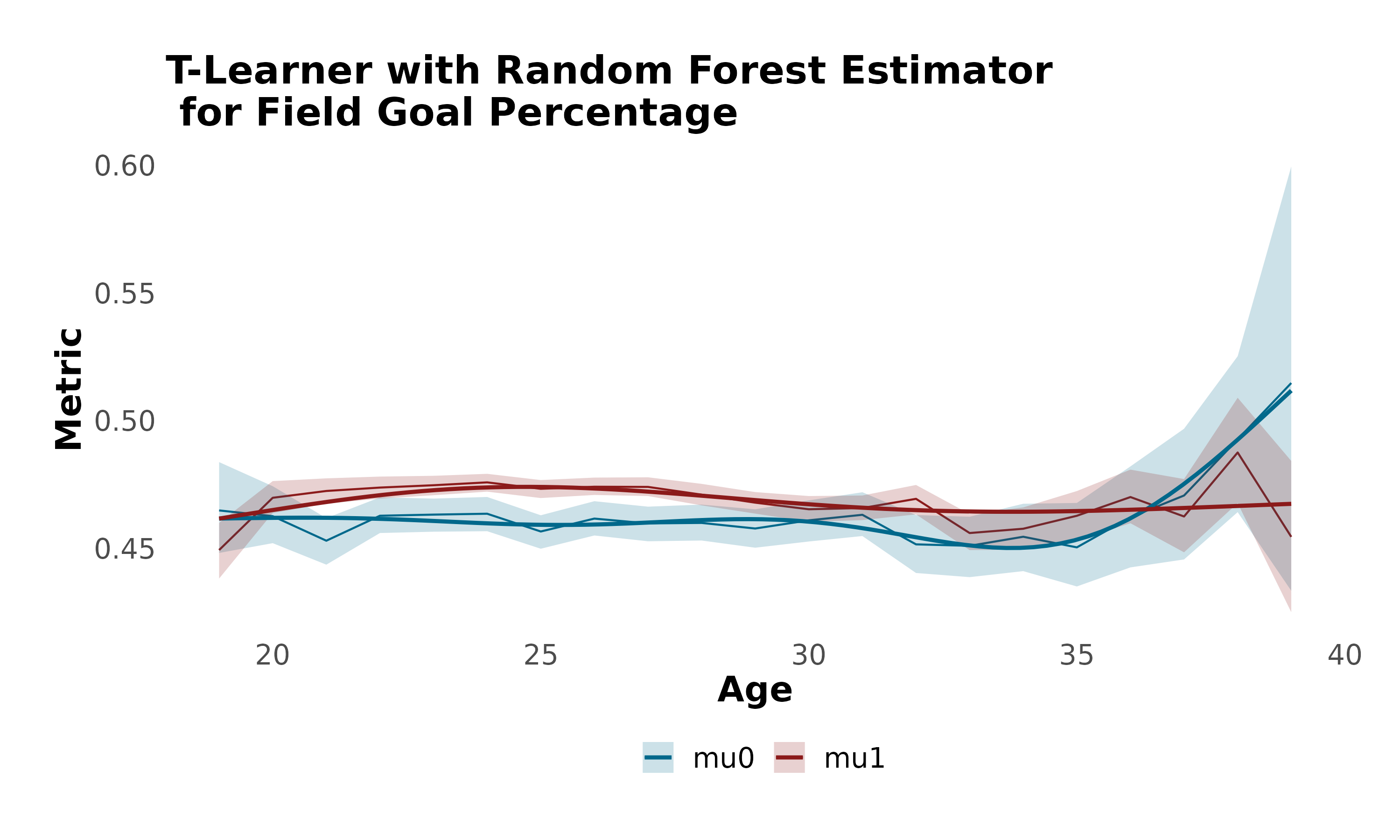} \includegraphics[scale=0.2]{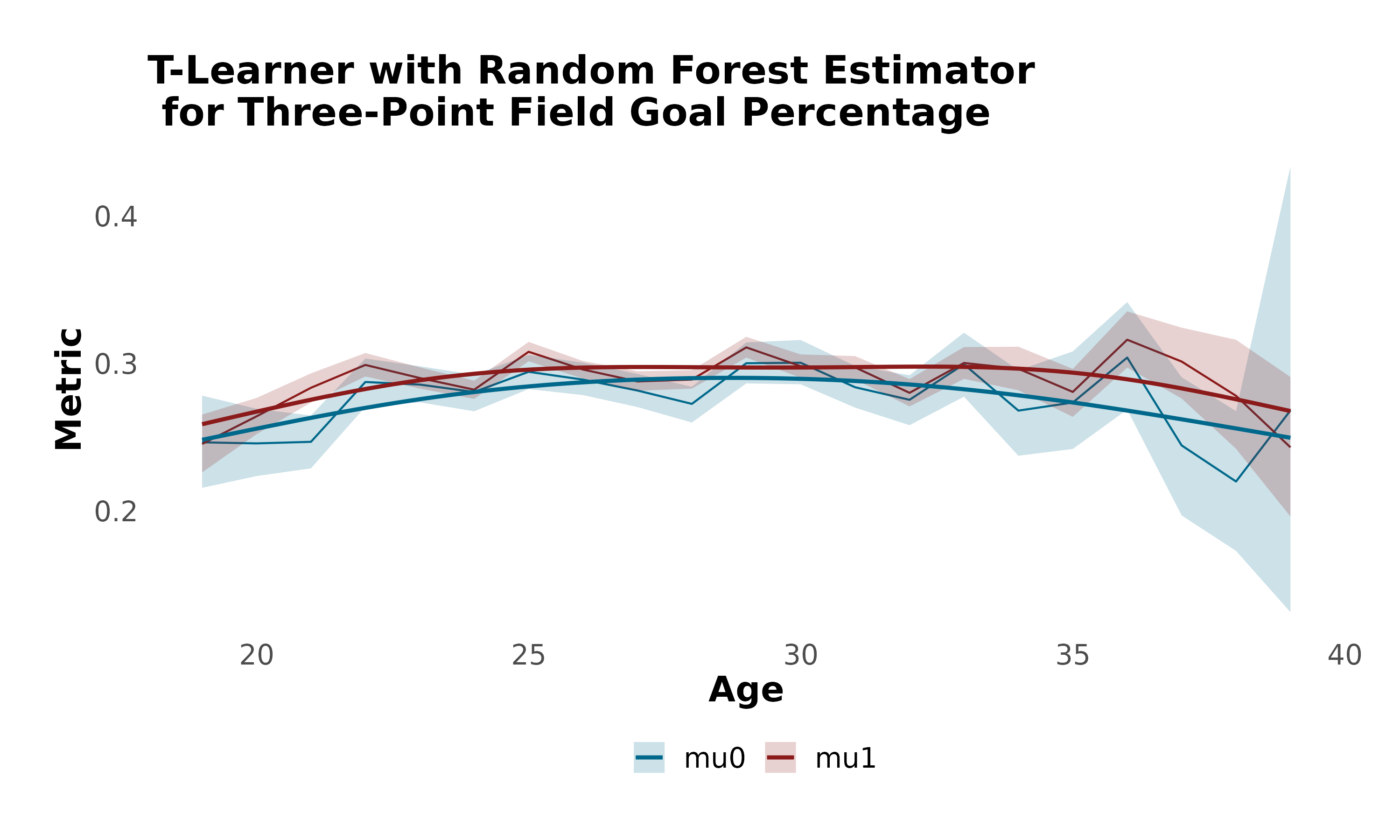}
    \includegraphics[scale=0.2]{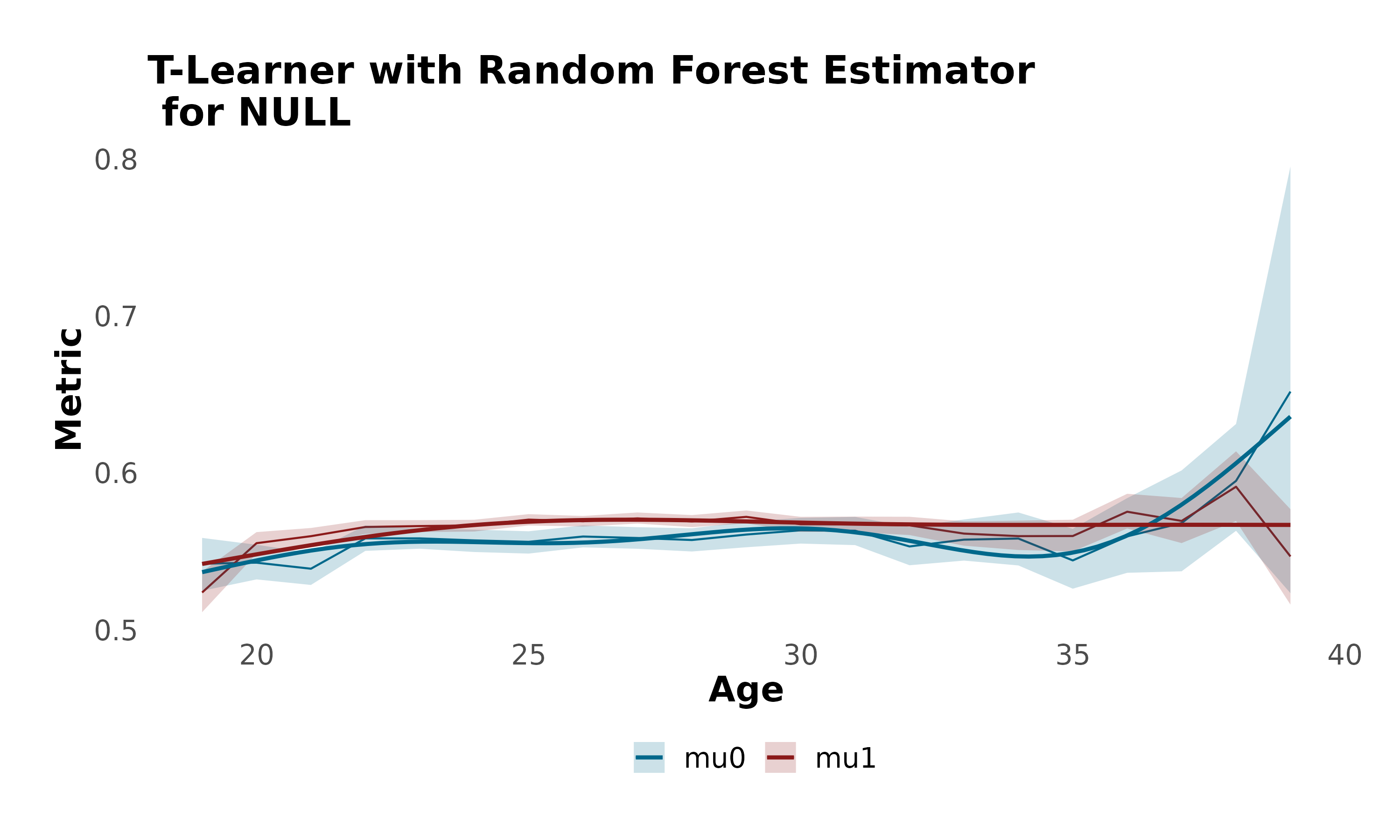}
    \caption{CEF using T-learner and RF for the field goal, three-point, and true shooting percentage}
    \label{fig:trf_shooting}
\end{figure}

\subsection{Age-Conditioned Treatment Effect (ACTE)}
In this section, we present the results obtained from the X-learner. Due to the algorithm's design, which 'crosses' inferential information between groups, it's not feasible to plot the two Conditional Expectation Functions (CEFs) for this model. Instead, we'll focus on each metric's Average Treatment Effect (ATE). It's important to note that in our specific covariate context, the X-learner's results using an OLS-spline approach would mirror those of the T-learner with an OLS-spline. Therefore, we will limit our discussion to the outcomes of employing a random forest with honesty. A significant benefit of this approach is that it incorporates information from both the treatment and control groups in the estimation of treatment effects.

In all the figures presented, various elements are strategically employed to illustrate the ACTE and its context. The finer line in each graph represents the ACTE itself, providing a direct visual representation of the treatment effect. Accompanying the ACTE line is its associated bandwidth, which delineates the 90\% bootstrap confidence interval. This interval is crucial for understanding the statistical significance of the observed effects. Additionally, a dashed horizontal line is plotted at zero. The relationship of the confidence interval with this dashed line is key: if the interval does not intersect the dashed line, it indicates that the treatment effect is statistically significant. Such visual demarcations in the figures are essential for a nuanced interpretation of the data, allowing for a more informed understanding of the significance and implications of the ACTE in various contexts.

In Figure \ref{fig:xrf_ratings}, the analysis upholds the main findings observed in the ACEFs. The net rating emerges as significantly impactful across various age groups, except older players where the limited sample size poses a challenge to conclusive analysis. A noteworthy aspect is the differential impact on the defensive and offensive ends; the ACTE in defense is notably greater in absolute value than in offense. This disparity suggests a more pronounced effect of rest on defensive capabilities. Consequently, we can infer that rest influences defensive performance more substantially than offensive output.

Upon analyzing Figure \ref{fig:xrf_box}, it is evident that the effects of rest on most box score statistics are minimal. However, notable exceptions can be observed in the domains of steals, blocks, and turnovers, particularly among younger athletes. This trend tends to be masked in the visualization due to the high variability present in older players' performance. For the older cohort, a clear adverse effect on scoring is observed, which could be attributed to older athletes possibly focusing more on offensive play during times of fatigue, as previously mentioned, or it may stem from a smaller sample size that renders the findings for this age group less definitive. Nevertheless, the reliability of these observations is limited by the overall small sample size, preventing a firm conclusion on the hypothesis.

In Figure \ref{fig:xrf_shooting}, our findings are in concordance with those observed by the S and T learners, indicating that the influence of rest is particularly pronounced in the areas of field goal and true shooting percentages. This trend can be attributed to the physically demanding nature of two-point plays, as discussed in previous sections. The agreement among various learning models strengthens the evidence that rest plays a crucial role in enhancing shooting efficiency, especially in situations that require higher physical effort. This thorough analysis highlights the complex impact of rest on different facets of player performance, providing critical insights for the development of effective player management strategies and rest schedules. To better illustrate the impact of rest, we focused on plotting data for ages between 20 to 38, as the effects were less discernible across all ages in the initial plot, allowing for a clearer observation of these effects.

\begin{figure}
    \centering
    \includegraphics[scale=0.5]{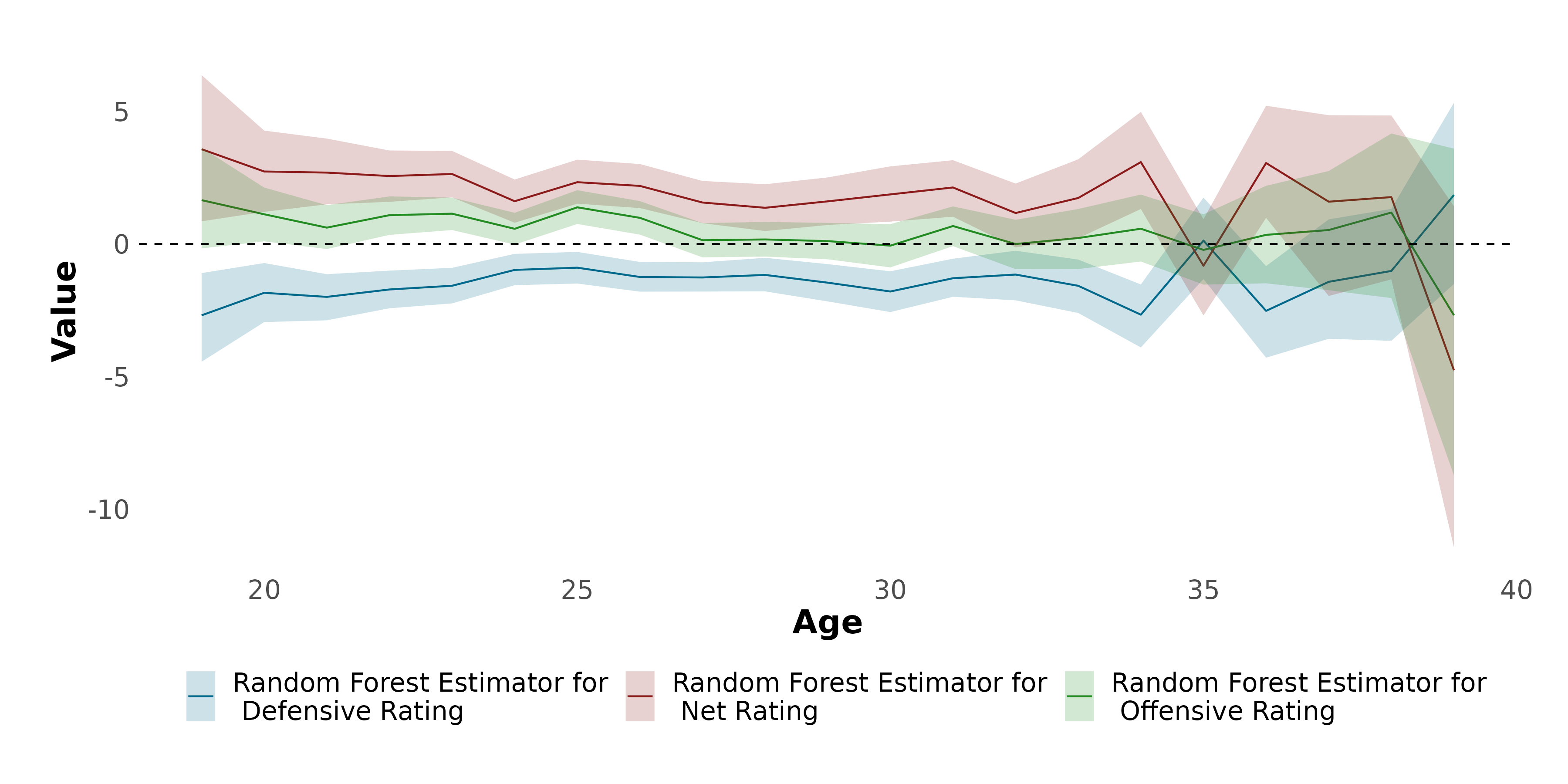} 
    \caption{ACTE using X-learner and RF for net, offensive, and defensive ratings}
    \label{fig:xrf_ratings}
\end{figure}

\begin{figure}
    \centering
    \includegraphics[scale=0.3]{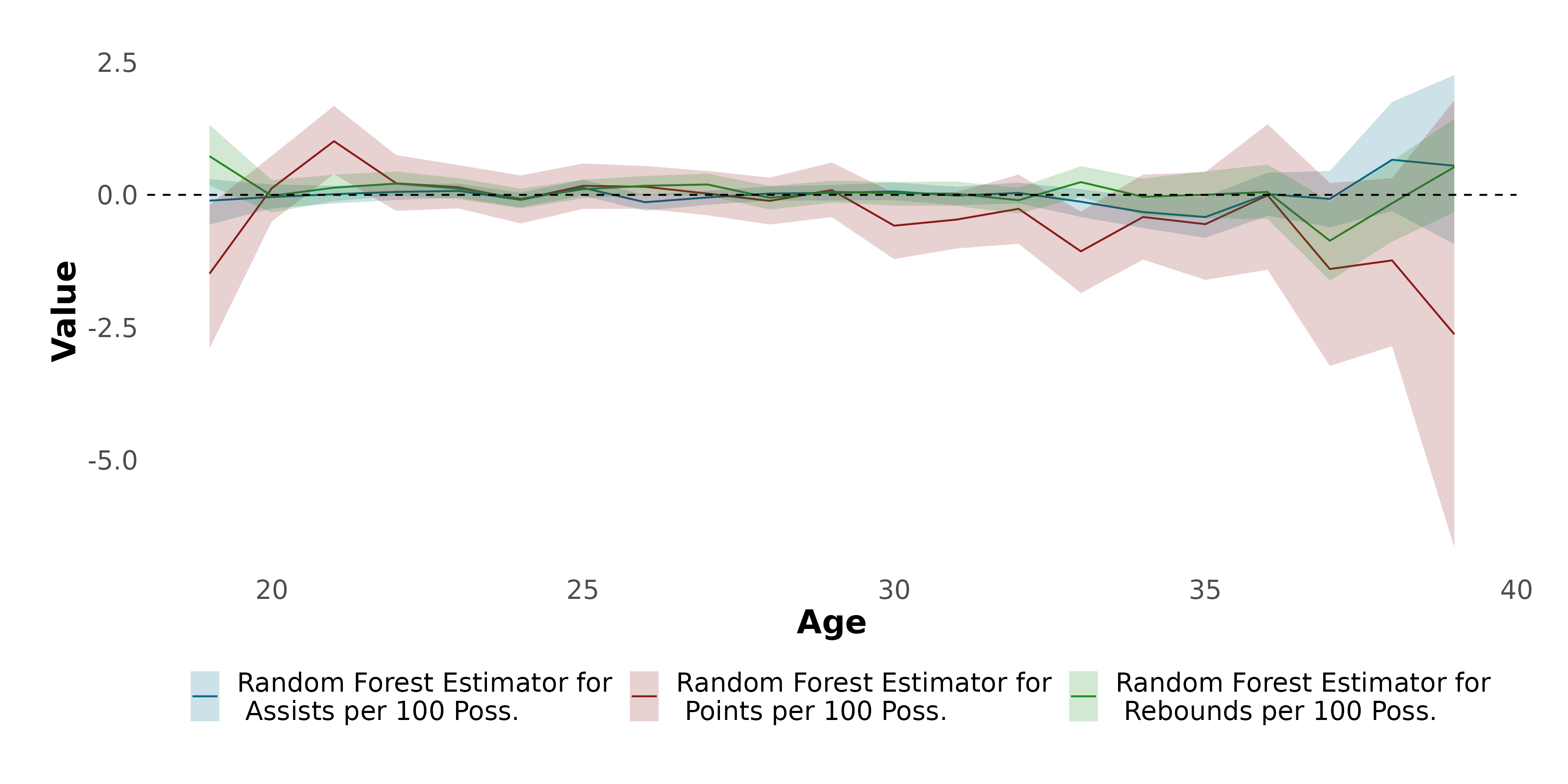} 
    \includegraphics[scale=0.3]{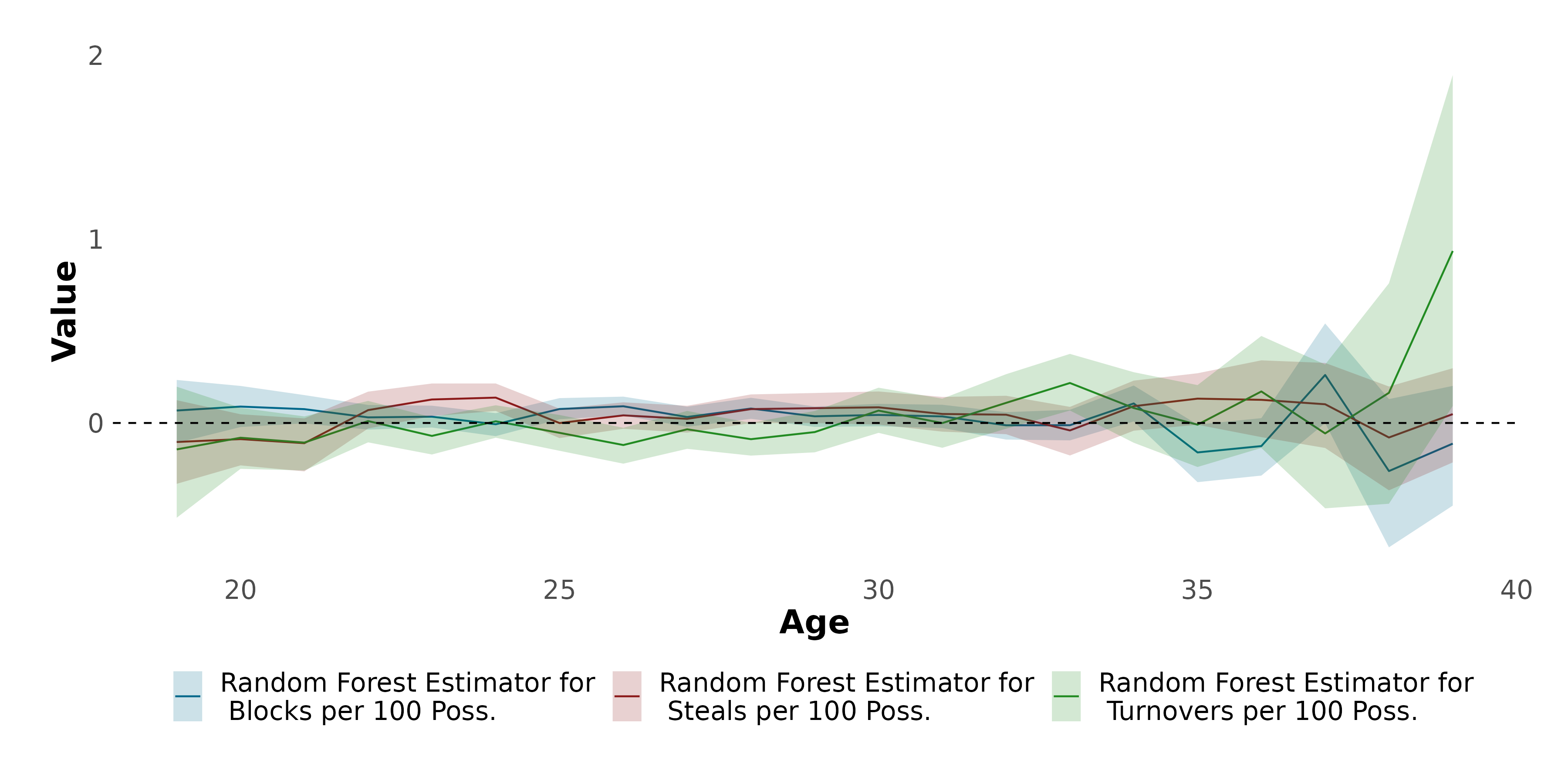} 
    \caption{ACTE using X-learner and RF for box score statistics per 100 possessions}
    \label{fig:xrf_box}
\end{figure}

\begin{figure}
    \centering
    \includegraphics[scale=0.3]{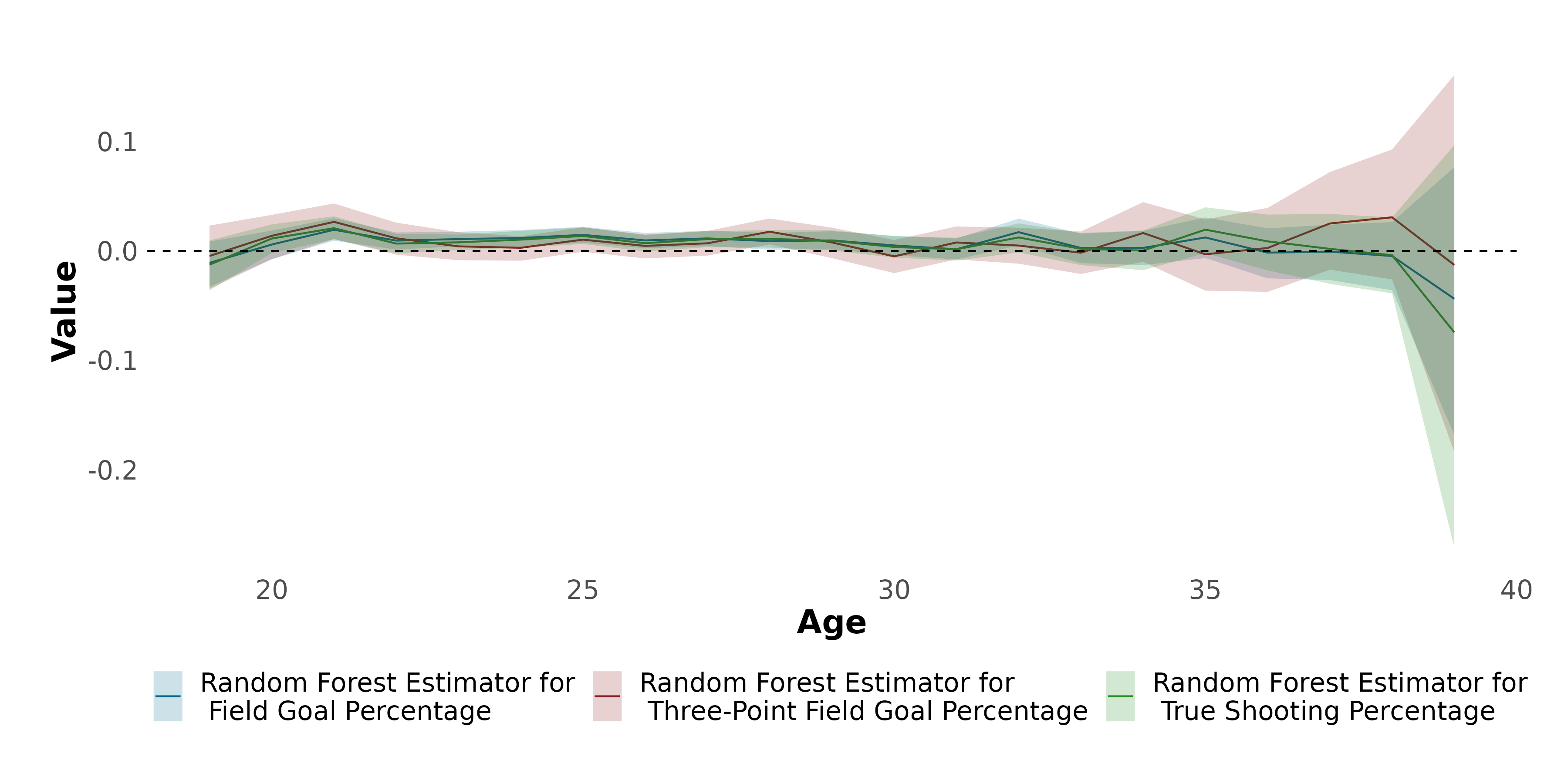} 
    \includegraphics[scale=0.3]{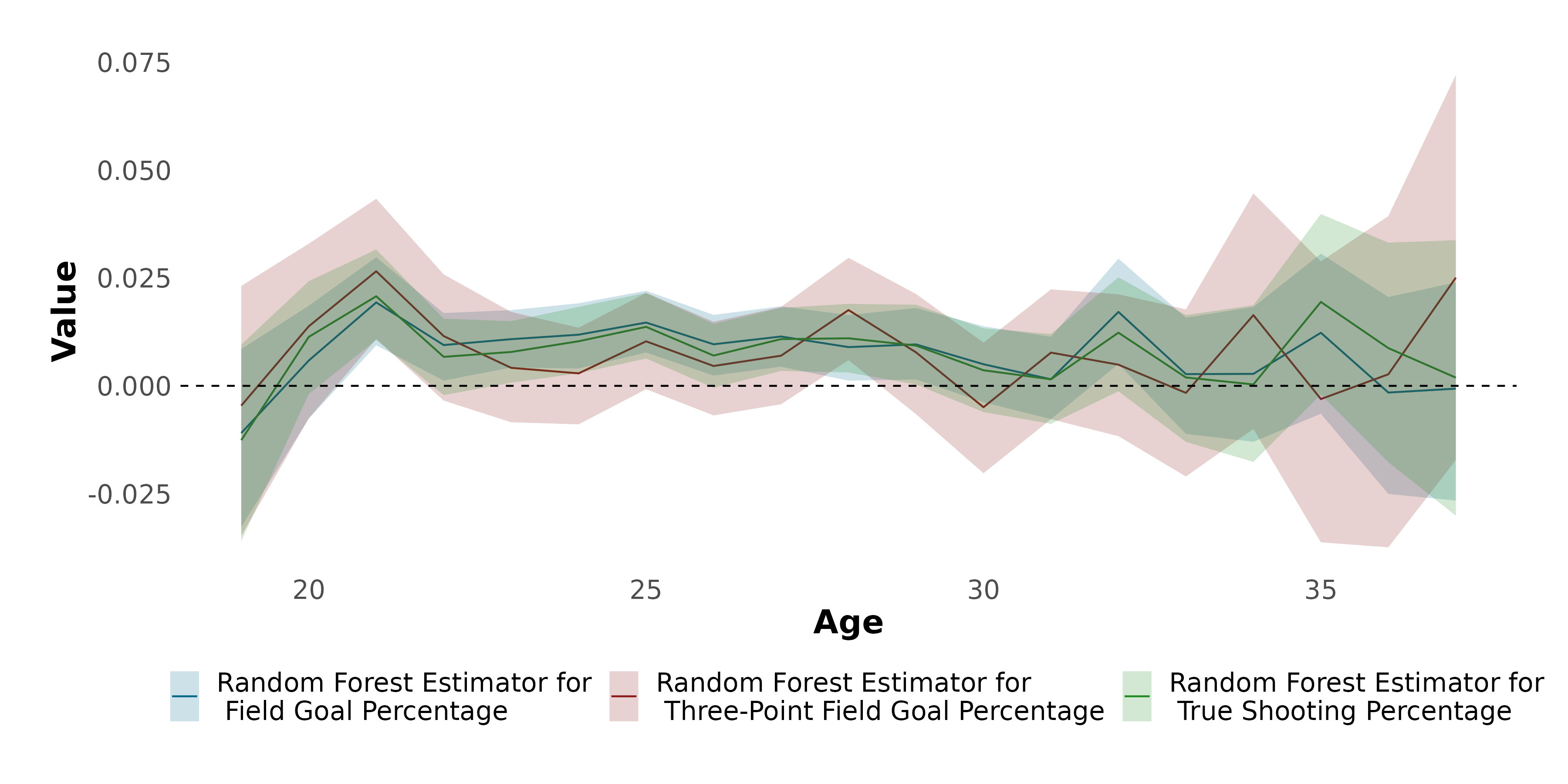} 
    \caption{ACTE using X-learner and RF for field goal, three-point, and true shooting percentage for age 20 to 40 (up) and 20 to 38 (bottom) }
    \label{fig:xrf_shooting}
\end{figure}

\section{Discussion}
\label{sec:5}
This study has analyzed the impact of load management strategies on player performance over time, with a particular focus on the NBA. By employing statistical methodologies, including meta-learners and various machine learning techniques, the research provides a nuanced understanding of how rest days influence player output, especially when considering different age groups. The findings suggest that strategic rest can enhance performance for several metrics, offering crucial insights for team management, coaches, and sports scientists aiming to optimize athlete longevity and effectiveness. This aligns with and expands upon existing literature that advocates for the careful management of player workload to prevent injury and maintain peak performance levels throughout a season.

Crucially, our work advances the age-curve literature by establishing a novel framework to estimate these curves using game-level data, as opposed to the traditional season-level approach. This innovation allows for a more granular analysis of how performance evolves with age and under different conditions, significantly enhancing our ability to understand and predict player trajectories. Additionally, by introducing a method to measure different treatment effects conditional on age, our study opens new avenues for testing multiple causal hypotheses. This represents a methodological leap forward, providing researchers, team analysts, and sports scientists with a powerful tool to explore a wide range of questions about athlete performance and management.

Finally, the research acknowledges its limitations, including the observational nature of the data and potential biases in performance measurement. These constraints underline the need for further studies, possibly incorporating controlled experiments or longitudinal studies across different sports, to validate and refine the findings. Future research could also explore the psychological impacts of load management, including athlete motivation and team dynamics, to provide a more comprehensive understanding of its effects. Such endeavors would not only enhance our grasp of optimal load management strategies but also contribute to developing healthier, more sustainable careers for professional athletes.

\clearpage
\appendix
\section{Proofs}\label{app:A}
The Stable Unit Treatment Value Assumption (SUTVA) consists of the following two sub-assumptions.

\begin{assumption}[No interference]\label{ass:interference}
The potential outcomes for any unit do not vary with the treatments assigned to other units. Formally, 
$$\forall \mathbf{w,w'}\in \textstyle\R^{N} \text{ such that } w_i=w'_i\Rightarrow Y_i(\mathbf{w})=Y_i(\mathbf{w'})$$
\end{assumption}
under this assumption we have $Y_i(\mathbf{w}) =Y_i(w_i)$, namely we index the potential outcome by the individual treatment $w_i$.

\begin{assumption}[Consistency]\label{ass:consistency}
There is only one version of each treatment $w$ and thus,
        $$Y_i^{obs}=Y_i(\mathbf{w})$$
\end{assumption}
where if we combine with assumtion \ref{ass:interference} we get $Y_i^{obs}=Y_i(w_i)$.
\subsection{Proof of Theorem \ref{them:IdentifiACTE}}
Under Assumptions \ref{ass:uncounfoundedness}, \ref{ass:interference}
    , and \ref{ass:consistency} we have
\begin{align*}
    &\textstyle\E_{\mathcal{X}}[\E[Y_i^{obs}|W_i=1, A_i=a, X_i=x]] - \textstyle\E_{\mathcal{X}}[\E[Y_i^{obs}|W_i=0, A_i=a, X_i=x]]\\
    =&\textstyle\E_{\mathcal{X}}[\E[Y_i(1)|W_i=1, A=a, X=x] - \E[Y_i(0)|W_i=0,A_i=a, X_i=x]]\\
    =&\textstyle\E_{\mathcal{X}}[\E[Y_i(1)|A_i=a, X_i=x] - \E[Y_i(0)|A_i=a, X_i=x]]\\
    =&\textstyle \E[Y_i(1)|A_i=a] - \E[Y_i(0)|A_i=a]\\
    =&g(a,1)-g(a,0)=\tau(a)\\
\end{align*}

where the second line follows from Assumption \ref{ass:interference} and \ref{ass:consistency}, and the third line follows from assumption \ref{ass:uncounfoundedness}, 

\section{Algorithms}\label{app:B}
In the subsequent text, we detail the algorithmic frameworks utilized in this study, in the form of pseudocode. The symbols $Y_0$ and $Y_1$ are employed to denote the observed outcomes for the control and treatment groups, respectively. For instance, $Y_{1i}$ represents the observed outcome for the $i$th subject within the treatment group. Similarly, $A_0, X_0$ and $A_1,X_1$ represent the attributes of the control and treated subjects, respectively, with $A^1_i,X^1_{i}$ being the age and attribute vector for the $i$th treated subject. We utilize the expression $M_k(Y \sim X)$ to denote a regression estimator that models regressing outcome $Y$ on $X$. This estimator may be any regression or machine learning technique. Following the format given by \cite{kunzel_metalearners_2019}.

\begin{algorithm}
\caption{S-learner}
\begin{algorithmic}
\Procedure{S-learner}{$A, X, Y, W$}
\State $\hat{\mu}_w = M_0(Y^{obs} \sim (A, X, W))$
\State $\hat{\tau}(a) = \textstyle \E_{\mathcal{X}}[\hat{\mu}_1(a,x) - \hat{\mu}_0(a,x)]$
\EndProcedure
\end{algorithmic}
\end{algorithm}

\begin{algorithm}
\caption{T-learner}
\begin{algorithmic}
\Procedure{T-learner}{$A, X, Y, W$}
\State $\hat{\mu}_0 = M_0(Y^0 \sim (A^0, X^0))$
\State $\hat{\mu}_1 = M_1(Y^1 \sim (A^1, X^1))$
\State $\hat{\tau}(a) = \textstyle \E_{\mathcal{X}}[\hat{\mu}_1(a,x) - \hat{\mu}_0(a,x)]$
\EndProcedure
\end{algorithmic}
\end{algorithm}

\begin{algorithm}
\caption{X-learner}
\begin{algorithmic}[1]
\Procedure{X-learner}{$A, X, Y, W, g$}
\State $\hat{\mu}_0 = M_1(Y^0 \sim (A^0,X^0))$ 
\State $\hat{\mu}_1 = M_2(Y^1 \sim (A^1,X^1))$
\State $D_i^1 = Y_i^1 - \hat{\mu}_0(A_i^1, X_i^1)$
\State $D_i^0 = \hat{\mu}_1(A_i^0, X_i^0) - Y_i^0$
\State $\hat{\tau}_1 = M_3(\bar{D}^1 \sim (A^1,X^1))$ 
\State $\hat{\tau}_0 = M_4(\bar{D}^0 \sim (A^0,X^0))$
\State $\hat{\tau}(a) = \textstyle \E_{\mathcal{X}}[g(a)\hat{\tau}_0(a,x) + (1 - g(a))\hat{\tau}_1(a,x)]$
\EndProcedure
\end{algorithmic}
\textit{The function $g(x)$, whose values range from 0 to 1, serves as a weighting mechanism selected to reduce the variance in the estimator $\hat{\tau}(x)$.}
\end{algorithm}

\begin{algorithm}[H]
\DontPrintSemicolon
\SetAlgoLined
\SetKwInput{KwData}{Input}
\SetKwInput{KwResult}{Output}
\KwData{$a$:age column of the training data , $x$: features of the training data, $w$: treatment assignments of the training data, $y$: observed outcomes of the training data, $p$: point of interest, $N$: number of observation}
\KwResult{Bootstrap Confidence Interval}
\BlankLine
S=\{1,...,N\}\;
\For{$b \in \{1, \dots, B\}$}{
    $s_b^* \gets \text{sample}(S, \text{replace} = \text{True}, \text{size} = N)$\;
    $x_b^* \gets x[s_b^*]$\;
    $w_b^* \gets w[s_b^*]$\;
    $y_b^* \gets y[s_b^*]$\;
    $\hat{\tau}_b^*(p) \gets \text{learner}(x_b^*, w_b^*, y_b^*)(p)$\;
}
\Return{($Q_{\alpha/2}(\{\hat{\tau}_b^*(p)\}_{b=1}^B), Q_{1-\alpha/2}(\{\hat{\tau}_b^*(p)\}_{b=1}^B)$)}\;
\caption{Procedure to compute bootstrap confidence intervals}
\end{algorithm}

\clearpage
\bibliographystyle{plainnat}  

\bibliography{ref}

\end{document}